\documentclass[a4paper,12pt]{article}
\usepackage{amsmath, bm}
\usepackage{amssymb,amsthm,graphicx}
\usepackage{enumitem}
\usepackage{xcolor}
\usepackage{epsfig}
\usepackage{graphics}
\usepackage{subcaption}
\usepackage[font=small]{caption}
\usepackage[hang,flushmargin]{footmisc} 
\usepackage{float}
\usepackage{rotating,tabularx}
\usepackage{booktabs}
\usepackage[mathscr]{euscript}
\usepackage{natbib}
\usepackage{setspace}
\usepackage{placeins}
\usepackage[left=2.7cm,right=2.7cm,bottom=2.7cm,top=2.3cm]{geometry}
\usepackage{parskip}
\numberwithin{equation}{section}

\allowdisplaybreaks[3]

\usepackage{hyperref}
\hypersetup{
hidelinks,
}

\newcommand{\reals}{\mathbb{R}}

\newcommand{\naturals}{\mathbb{N}}

\newcommand{\pr}{\mathbb{P}}        
\newcommand{\ex}{\mathbb{E}}        

\newcommand{\ind}{\boldsymbol{1}}          
\newcommand{\norm}[1]{\left\| #1 \right\|} 
\newcommand{\bs}[1]{\boldsymbol{#1}}

\newcommand{\haz}{\alpha^*}                  
\newcommand{\hazvec}{\boldsymbol{\alpha}^*}  
\newcommand{\cumhaz}{A^*}                    
\newcommand{\quadvar}[1]{\langle #1 \rangle} 
\newcommand{\ngrid}{n}

\newcommand{\rpack}[1]{\texttt{#1}}

\theoremstyle{plain}
\newtheorem{theorem}{Theorem}[section]
\newtheorem{prop}[theorem]{Proposition}
\newtheorem{lemma}[theorem]{Lemma}

\newtheorem*{theorem-unnumbered}{Theorem}

\theoremstyle{definition}

\newtheorem*{remark-nonumber}{Remark}
\newtheorem{setting}{Setting}

\makeatletter
\let\save@mathaccent\mathaccent
\newcommand*\if@single[3]{%
  \setbox0\hbox{${\mathaccent"0362{#1}}^H$}%
  \setbox2\hbox{${\mathaccent"0362{\kern0pt#1}}^H$}%
  \ifdim\ht0=\ht2 #3\else #2\fi
  }
\newcommand*\rel@kern[1]{\kern#1\dimexpr\macc@kerna}
\newcommand*\widebar[1]{\@ifnextchar^{{\wide@bar{#1}{0}}}{\wide@bar{#1}{1}}}
\newcommand*\wide@bar[2]{\if@single{#1}{\wide@bar@{#1}{#2}{1}}{\wide@bar@{#1}{#2}{2}}}
\newcommand*\wide@bar@[3]{%
  \begingroup
  \def\mathaccent##1##2{%
    \let\mathaccent\save@mathaccent
    \if#32 \let\macc@nucleus\first@char \fi
    \setbox\z@\hbox{$\macc@style{\macc@nucleus}_{}$}%
    \setbox\tw@\hbox{$\macc@style{\macc@nucleus}{}_{}$}%
    \dimen@\wd\tw@
    \advance\dimen@-\wd\z@
    \divide\dimen@ 3
    \@tempdima\wd\tw@
    \advance\@tempdima-\scriptspace
    \divide\@tempdima 10
    \advance\dimen@-\@tempdima
    \ifdim\dimen@>\z@ \dimen@0pt\fi
    \rel@kern{0.6}\kern-\dimen@
    \if#31
      \overline{\rel@kern{-0.6}\kern\dimen@\macc@nucleus\rel@kern{0.4}\kern\dimen@}%
      \advance\dimen@0.4\dimexpr\macc@kerna
      \let\final@kern#2%
      \ifdim\dimen@<\z@ \let\final@kern1\fi
      \if\final@kern1 \kern-\dimen@\fi
    \else
      \overline{\rel@kern{-0.6}\kern\dimen@#1}%
    \fi
  }%
  \macc@depth\@ne
  \let\math@bgroup\@empty \let\math@egroup\macc@set@skewchar
  \mathsurround\z@ \frozen@everymath{\mathgroup\macc@group\relax}%
  \macc@set@skewchar\relax
  \let\mathaccentV\macc@nested@a
  \if#31
    \macc@nested@a\relax111{#1}%
  \else
    \def\gobble@till@marker##1\endmarker{}%
    \futurelet\first@char\gobble@till@marker#1\endmarker
    \ifcat\noexpand\first@char A\else
      \def\first@char{}%
    \fi
    \macc@nested@a\relax111{\first@char}%
  \fi
  \endgroup
}
\makeatother

\newcommand{\heading}[2]
{  \setcounter{page}{1}
   \begin{center}

   {\LARGE \textbf{#1}}
   \vspace{0.2cm}

   {\LARGE \textbf{#2}}
   \end{center}
}

\begin{document}

\heading{Piecewise constant hazard estimation}{with the fused lasso}
\bigskip
\begin{center}
{\large Manuel Rosenbaum, Jan Beyersmann \& Michael Vogt\footnote{Address: Institute of Statistics, Department of Mathematics and Economics, Ulm University, 89081 Ulm, Germany. Email addresses: \texttt{\{manuel.rosenbaum,jan.beyersmann,m.vogt\}@uni-ulm.de}. Corresponding author is Michael Vogt.}}\\[0.4cm] 
{\large \textit{Ulm University}}
\end{center}
\vspace{-0.7cm}
\renewcommand{\abstractname}{}
\begin{abstract}
\noindent In applied time-to-event analysis, a flexible and useful approach is to model the hazard rate as a piecewise constant function of time. However, the change points and values of the piecewise constant hazard are usually unknown and need to be estimated. In this paper, we develop a fully data-driven procedure for piecewise constant hazard estimation. We work in a general counting process framework which nests a wide range of popular models in time-to-event analysis including Cox's proportional hazards model with potentially high-dimensional covariates, competing risks models as well as more general multi-state models. To construct our estimator, we set up a regression model for the increments of the Breslow estimator and then use fused lasso techniques to approximate the piecewise constant signal in this regression model. In the theoretical part of the paper, we derive the convergence rate of our estimator as well as some results on how well the change points of the piecewise constant hazard are approximated by our method. We complement the theory by both simulations and a real data example, illustrating that our results apply in rather general event histories such as multi-state models.
\end{abstract}
\renewcommand{\baselinestretch}{1.2}\normalsize

\section{Introduction}\label{sec:intro}

Hazard rates are a central quantity in time-to-event analysis. A flexible and useful approach is to model the hazard as a piecewise constant function of time. 
In practice, the change points and function values of the piecewise constant hazard are unknown and need to be estimated. In this paper, we develop a fully data-driven method for estimating piecewise constant hazard functions. Our method is based on techniques from high-dimensional statistics, in particular on fused lasso techniques \citep{Tibshirani2005}, which -- unlike many other machine learning techniques -- are theoretically tractable and thus allow us to back up our estimation approach by quite comprehensive theory.

We work in a general counting process framework with a multiplicative intensity structure. Specifically, the intensity process is the product of a deterministic function of time -- the (baseline) hazard rate we want to estimate -- and a predictable process (which possibly incorporates covariate effects and may thus depend on unknown parameters). This framework nests a wide range of popular models in time-to-event analysis, including the Cox model with potentially high-dimensional covariates, competing risks models as well as more general multistate models. The counting process framework under consideration is introduced formally in Section \ref{sec:model} along with some leading examples of models that are nested in it.

Piecewise constant hazard models are regularly presented in the textbook literature on time-to-event analysis as a very flexible parametric modelling approach.
However, guidance on choosing the number and location of the change points is either missing or rather informal.
In the research literature, the majority of articles is restricted to simple survival models where the hazard only has a single jump or change point.
Early examples are \cite{MatthewsFarewell1982}, \cite{NguyenRogersWalker1984} and \cite{Loader1991} who study maximum likelihood methods for estimating the change point. Further studies include \cite{ChangChenHsiung1994}, \cite{GijbelsGuerler2003} and \cite{ZhaoWuZhou2009}. Closely related to the problem of estimating the change point is the problem of testing the null hypothesis of a constant hazard against the alternative of a change point. This test problem is investigated in \cite{MatthewsFarewellPyke1985}, \cite{Yao1986}, \cite{Worsley1988}, \cite{Henderson1990} and \cite{Loader1991} among others.

Only very few studies allow the piecewise constant hazard to have multiple change points. \cite{GoodmanLiTiwari2011} and \cite{HanSchellKim2014} propose sequential testing procedures to detect multiple change points in a simple survival model with right-censoring and no covariates. An extension of the approach in \cite{GoodmanLiTiwari2011} to a two-sample setting is provided in \cite{HeFangSu2013}. Apart from these frequentist methods, Bayesian approaches to estimate a piecewise constant hazard with multiple change points can be found in \cite{ArjasGasbarra1994} and \cite{CooneyWhite2021,CooneyWhite2023}.

An alternative approach to fitting piecewise constant hazards is based on piecewise exponential models (PEMs). 
PEMs discretize time into a large number $N$ of subintervals and assume that the hazard is constant within each interval. The hazard can thus be parameterized by a vector of $N$ interval-specific hazard levels, which are estimated jointly with the remaining model parameters by maximizing a Poisson likelihood. For sufficiently large $N$, this approach produces a very flexible estimator of the hazard. However, it has the drawback that the resulting estimator is not parsimonious: the estimated hazard levels are typically all different from each other. Consequently, even if the underlying hazard is a simple piecewise constant function with a single change point, the estimated hazard will exhibit $N$ change points in general. 
Another drawback of PEMs is their high dimensionality, as at least $N$ hazard parameters must be estimated. To address this issue, spline-based extensions such as piecewise exponential additive models (PAMs) have been proposed; see \cite{BenderGrollScheipl2018} for an overview. Rather than estimating $N$ unrelated hazard parameters, PAMs model the hazard by a smooth spline function and estimate the corresponding spline coefficients together with the remaining model parameters by penalized maximum likelihood with a quadratic penalty. While this approach substantially reduces the effective dimensionality of the problem, it tends to smooth out piecewise constant structures in the hazard and is therefore not specifically designed to recover piecewise constant hazards.

Our approach is based on the fused lasso. Similar to PEMs, we start from a fine discretization of the time axis and allow for a separate hazard level in each interval. However, instead of imposing smoothness through a quadratic penalty as in PAMs, we employ a fused lasso penalty on adjacent hazard levels. This penalty encourages neighboring hazard levels to coincide and therefore promotes piecewise constant hazard estimates with only a small number of change points. As a result, the proposed approach combines the flexibility of high-dimensional PEMs with automatic complexity reduction. In particular, when the underlying hazard is piecewise constant with only a few change points, our fused lasso estimator produces a parsimonious approximation with only a small number of change points in the vicinity of the true jumps.


To construct our estimator, we reformulate hazard estimation as a regression problem to which fused lasso methodology can be applied. First, we estimate the cumulative hazard nonparametrically using an off-the-shelf estimator, namely the Breslow estimator (which reduces to the Nelson-Aalen estimator in the absence of covariates). Next, we compute increments of the Breslow estimator over a fine time grid. These increments can be shown to satisfy a simple regression model with the following property: the regression function is a discretized version of the underlying piecewise constant hazard. This allows us to construct an estimator of the hazard by applying techniques for estimating piecewise constant regression curves -- in particular, fused lasso techniques -- to the Breslow increments.

The main merits of our approach are as follows:
\begin{enumerate}[label=(\roman*),leftmargin=0.95cm]
\item The approach is not restricted to a specific time-to-event model. It rather works in a general counting process framework which nests a wide range of popular models in time-to-event analysis. 
\item The approach does not presuppose any knowledge about the piecewise constant structure of the hazard. In particular, the location of change points, their number and the function values of the hazard are unknown.
\item The approach produces parsimonious solutions, i.e., piecewise constant reconstructions of the hazard with only few parameters.
\item The approach is fully daten-driven: there are no free tuning parameters that need to be chosen in an adhoc fashion.  
\end{enumerate}
We are not aware of any other established technique that comes with all these advantages and is underpinned by comprehensive theory.

Our fused-lasso-based estimation procedure is developed step by step in Section \ref{sec:estimation} and backed up by theory in Section \ref{sec:theory}. Implementation details are dealt with in Section~\ref{sec:impl}. The methods and theory are complemented by a simulation study in Section \ref{sec:sim} and a real data example in Section \ref{sec:app}. The data example considers a piecewise constant parametrization of the three-variate hazard measure in an illness-death multistate model, which is used to jointly model so-called progression-free and overall survival, major time-to-event outcomes in oncology. The piecewise constant fit is judged by comparison against standard nonparametric Kaplan-Meier estimators and may, e.g., be used for planning of randomised controlled trials and evaluation of operational trial characteristics.

\section{Model setting}\label{sec:model}

In its most general form, our model can be formulated as follows: We observe an $n$-dimensional counting process $N_{1:n} = \{N_{1:n}(t): t \in [0,\tau]\}$ on a probability space $(\Omega,\mathcal{F},\pr)$. The process $N_{1:n}$ is adapted to the filtration $\{\mathcal{F}_t\}$, where $\mathcal{F}_t$ is the $\sigma$-algebra generated by the observed data up to time $t$. Hence, $\mathcal{F}_t$ represents the information available up to time $t$. Writing $N_{1:n} = (N_1,\ldots,N_n)$, we regard $N_i = \{ N_i(t): t \in [0,\tau]\}$ as the counting process of the $i$-th observed subject in a sample of size $n$. The Doob-Meyer decomposition yields that for each $i$, 
$N_i(t) = \Lambda_i(t) + M_i(t)$, 
where $\Lambda_i$ is a unique $\{\mathcal{F}_t\}$-predictable process and $M_i$ is a zero-mean local martingale with respect to $\{\mathcal{F}_t\}$. Assuming that $\Lambda_i$ is absolutely continuous, we can write $\Lambda_i(t) = \int_0^t \lambda_i(s) ds$ with some $\{\mathcal{F}_t\}$-predictable process $\lambda_i$. Plugging this into the Doob-Meyer decomposition gives
\begin{equation}\label{eq:Doob-Meyer-Ni}
N_i(t) = \int_0^t \lambda_i(s) ds + M_i(t),
\end{equation}  
where $\lambda_i$ is usually referred to as the intensity process of $N_i$. As common in the literature \citep[cp.][Chapter 3.1.2]{Aalen2008}, we assume that $\lambda_i$ has the multiplicative structure
\begin{equation}\label{eq:multiplicative-intensity-model}
\lambda_i(t) = Z_i(t,\bs{\beta}) \haz(t),
\end{equation}  
where $\haz$ is a non-negative deterministic function and $Z_i$ is an $\{\mathcal{F}_t\}$-predictable process that may depend on certain unknown parameters $\bs{\beta} \in \reals^d$ which need to be estimated from the data. The multiplicative structure \eqref{eq:multiplicative-intensity-model} arises very naturally in many popular models in time-to-event analysis. Some leading examples are discussed below. In summary, we consider the general model
\begin{equation}\label{eq:model-i}
N_i(t) = \int_0^t Z_i(s,\bs{\beta}) \haz(s) ds + M_i(t)
\end{equation}  
for each subject $i$ and the corresponding cumulative model
\begin{equation}\label{eq:model-cum}
\widebar{N}(t) = \int_0^t \widebar{Z}(s,\bs{\beta}) \haz(s) ds + \widebar{M}(t),
\end{equation}  
where $\widebar{N}(t) = \sum_{i=1}^n N_i(t)$, $\widebar{Z}(t,\bs{\beta}) = \sum_{i=1}^n Z_i(t,\bs{\beta})$ and $\widebar{M}(t) = \sum_{i=1}^n M_i(t)$.
Notably, the processes $\widebar{N}$, $\widebar{Z}$ and $\widebar{M}$ depend on the sample size $n$. However, for ease of notation, we suppress this dependence throughout the paper.

The central statistical object in model \eqref{eq:model-cum} besides the parameter vector $\bs{\beta}$ is the function $\haz: [0,\tau] \to \reals_{\ge 0}$. In this paper, we assume that $\haz$ has a piecewise constant structure: there are $K$ time points $0 < \tau_1 < \ldots < \tau_K < \tau$ such that
\begin{equation}\label{eq:pc-hazard}
\haz(t) =
\begin{cases}
\mathfrak{a}_1 & \text{for } t \in [0,\tau_1) \\
\mathfrak{a}_2 & \text{for } t \in [\tau_1,\tau_2) \\
\ \vdots & \\
\mathfrak{a}_{K+1} & \text{for } t \in [\tau_K,\tau], 
\end{cases}
\end{equation}
where $\mathfrak{a}_1,\ldots,\mathfrak{a}_{K+1}$ are non-negative constants. The jump locations $\tau_1,\ldots,\tau_K$, their number $K$ as well as the function values $\mathfrak{a}_1,\ldots,\mathfrak{a}_{K+1}$ are unknown. Our main goal is to construct an estimator of the piecewise constant function $\haz$.

\begin{remark-nonumber}
Most probably, the assumption that the function $\haz$ is piecewise constant does not hold exactly but only approximately in practice. We could reflect this in our methods and theory by allowing $\haz$ to be a general function which can be approximated sufficiently well by a sequence of piecewise constant functions $\haz_n$ with $K_n$ jumps (where $K_n$ grows with the sample size $n$ and the approximation quality gets better with increasing $n$). It is possible to adapt our theory to this more general setting. However, as this would make the already quite intricate formulation of our results and proofs even more involved, we have decided to stick with the simpler assumption that $\haz$ is exactly piecewise constant.     
\end{remark-nonumber}

\begin{remark-nonumber}
In some applications, it may be desirable to allow the jump locations $\tau_1,\ldots$ $\ldots,\tau_K$ to depend on covariates. Within our framework, however, this cannot be accommodated as the baseline hazard is multiplicatively separated from the covariate effects.
\end{remark-nonumber}

We now discuss some special cases of the general framework introduced above. In all considered settings, we assume that the data are independent and identically distributed (i.i.d.) across units $i$.

\begin{setting}\label{settingA}
To start with, we consider a simple survival model. Let $T_i^*$ denote the survival time of subject $i$ and assume that $T_i^*$ has a density $f^*$. Our goal is to estimate the density $f^*$ or the associated distribution function $F^*(t) = \int_0^ t f^*(s) ds$. The situation is complicated by the fact that we do not fully observe the survival times $T_i^*$ but only right-censored versions of them. More specifically, we observe the variables $T_i = T_i^* \land C_i$, where $\land$ denotes the minimum and $C_i$ is the censoring time, together with the variables $\delta_i = \ind(T_i^* \le C_i)$, which indicate whether there is censoring or not. The data sample thus has the form $\{ (T_i,\delta_i): i=1,\ldots,n \}$. \newline  
For theoretical analysis, a counting process formulation of the model is commonly used. In particular, we consider the counting processes $N_i$ defined by 
\[ N_i(t) = \ind(T_i \le t, \delta_i = 1) \]
for $i=1,\ldots,n$. Under standard conditions discussed below, the model equations \eqref{eq:model-i} and \eqref{eq:model-cum} arise with 
\begin{equation}\label{eq:hazard-rate}
\haz(t) := \frac{f^*(t)}{1 - F^*(t)} = \lim_{\delta \searrow 0} \frac{1}{\delta} \, \pr(t \le T_i^* < t+\delta \, | \, T_i^* \ge t)
\end{equation}  
and $Z_i(t) := \ind(T_i \ge t)$ (which does not depend on any unknown parameters $\bs{\beta}$, that is, $Z_i(t,\bs{\beta}) = Z_i(t)$). Hence, the function $\haz$ is identical to the hazard rate in the model. The hazard rate contains all the information we need. In particular, $f^*$ and $F^*$ can be retrieved from $\haz$ by the formulas $f^*(t) = \haz(t) \exp(- \int_0^t \haz(s) ds)$ and $F^*(t) = 1 - \exp(- \int_0^t \haz(s) ds)$. \newline
The key assumption to guarantee \eqref{eq:hazard-rate} -- i.e., to guarantee that the function $\haz$ is identical to the hazard parameter of interest -- is called \emph{independent censoring} in the counting process literature; see e.g.\ \cite{AndersenGill1993} or \cite{MartinussenScheike2006}. A special case is the stronger assumption of \emph{random censoring} where $T_i^*$ and $C_i$ are supposed to be stochastically independent. 
\end{setting}

\begin{setting}\label{settingB}
We now extend Setting \ref{settingA} by incorporating covariates. There are various ways to do so. Presumably the most prominent way proposed by Cox leads to the proportional hazards model \citep{Cox1972}. Let the variables $(T_i,\delta_i)$ be defined as above and assume that we additionally observe a vector of $d$ covariates $\bs{W}_i = (W_{i1},\ldots,W_{id})^\top$, which are time-independent for simplicity. The data sample is thus given by $\{(T_i,\delta_i,\bs{W}_i): i=1,\ldots,n \}$. Notably, we do not assume the number of covariates $d$ to be small relative to the sample size $n$. We rather allow for high-dimensional (sparse) cases where $d$ is potentially much larger than $n$. A precise description of the high-dimensional setting under consideration is provided in Section \ref{subsec:theory-settings}. As in Setting \ref{settingA}, we consider the counting processes
\[ N_i(t) = \ind(T_i \le t, \delta_i = 1) \]
for $i=1,\ldots,n$. In the Cox model, the intensity process $\lambda_i$ of $N_i$ is assumed to have the form
\[ \lambda_i(t) = \ind(T_i \ge t) \exp(\bs{\beta}^ \top \bs{W}_i) \haz(t). \]
Setting $Z_i(t,\bs{\beta}) = \ind(T_i \ge t) \exp(\bs{\beta}^ \top \bs{W}_i)$, the process $\lambda_i$ thus has the multiplicative structure \eqref{eq:multiplicative-intensity-model}, implying that model equations \eqref{eq:model-i} and \eqref{eq:model-cum} are satisfied. As in Setting \ref{settingA}, we assume the right-censoring to be independent. A special case is once again random censoring, meaning that $T_i^*$ and~$C_i$ are stochastically independent given the covariates. Under the assumption of independent censoring, the conditional hazard 
\begin{equation}\label{eq:cond-hazard-rate}
\haz(t \, | \, \bs{w}) := \lim_{\delta \searrow 0} \frac{1}{\delta} \, \pr(t \le T_i^* < t+\delta \, | \, T_i^* \ge t, \bs{W}_i=\bs{w})
\end{equation}  
is identical to $\haz(t \, | \, \bs{w}) = \exp(\bs{\beta}^ \top \bs{w}) \haz(t)$ and $\haz$ plays the role of a baseline hazard in the model. 
\end{setting}

\begin{setting}\label{settingC}
In Settings \ref{settingA} and \ref{settingB}, the observed subjects $i$ only face one risk, e.g., the risk of dying when $T_i^*$ is time-to-death. The competing risks model, in contrast, allows to incorporate different (competing) risks. Technically speaking, the survival times~$T_i^*$ have associated marks~$\varepsilon_i\in\{1, 2, \dots R\}$ which correspond to the $R$ possible risks in the model. For instance, if~$T_i^*$ is time-to-death, there are $R$ causes of death such as cardiovascular death ($\varepsilon_i = 1$), cancer death ($\varepsilon_i = 2$), and so on. \newline
As in Settings \ref{settingA} and \ref{settingB}, we only observe a right-censored version $T_i = T_i^* \land C_i$ of the survival time $T_i^*$, where $C_i$ denotes the censoring time. In addition, we allow for left-truncation in the model, where the left-truncation time $L_i$ is assumed to be strictly smaller than $C_i$ with probability $1$. The time $L_i$ is usually interpreted as delayed study entry, a common phenomenon in observational data. This means that subject $i$ only becomes observable after time $L_i$ \textit{provided that} $T_i^* > L_i$. Thus, in the presence of left-truncation, the data are sampled from a conditional probability distribution, in particular, from the distribution given that study entry occurs. Defining $\delta_i$ as before and allowing for covariates $\bs{W}_i$ as in Setting~\ref{settingB}, the observed data sample has the form $\{(T_i,\delta_i \cdot \varepsilon_i, L_i, \bs{W}_i): i=1,\ldots,n\}$. \newline
We now treat each event $r \in \{1, 2, \ldots R\}$ separately. We thus fix $r$ and consider the event-specific counting processes $N_{ir}$ for $i=1,\ldots,n$ defined by
\[ N_{ir}(t) = \ind(L_i < T_i \le t, \delta_i\cdot\varepsilon_i = r). \]
As in Setting \ref{settingB}, we assume a Cox model such that the intensity $\lambda_{ir}$ of $N_{ir}$ has the form 
\begin{displaymath}
\lambda_{ir}(t) = \ind(L_i< t\le T_i) \exp(\bs{\beta}_r^\top \bs{W}_i) \haz_r(t),
\end{displaymath}
where $\bs{\beta}_r$ is a vector of event-specific regression coefficients. For given $r$, the $n$-dimen\-sio\-nal counting process with the components $N_{ir}$ for $i=1,\ldots,n$ thus satisfies the general model equations \eqref{eq:model-i} and \eqref{eq:model-cum}. Under the assumption of independent right-censoring and independent left-truncation, the event-specific conditional hazard 
\begin{displaymath}
\haz_r(t \, | \, \bs{w}) := \lim_{\delta \searrow 0} \frac{1}{\delta} \, \pr\big(t \le T_i^* < t+\delta, \varepsilon_i=r \, \big| \, T_i^* \ge t, \bs{W}_i=\bs{w}\big)
\end{displaymath}
is identical to $\haz_r(t \, | \, \bs{w}) = \exp(\bs{\beta}_r^\top \bs{w})\haz_r(t)$ and the function $\haz_r$ has the interpretation of an event-specific baseline hazard. 
\end{setting}

\begin{setting}\label{settingD}
As a final setting, we consider a general multi-state model where we observe a finite-state Markov process for each subject $i=1,\ldots,n$. More specifically, we observe a nonhomogeneous, time-continuous Markov process $X_i$ with state space~$\mathcal{R} = \{0, 1, 2, \ldots R\}$, right-continuous sample paths and, for ease of presentation, $\pr(X_i(0) = 0) = 1$ for each $i$. A generalization to a non-degenerate initial distribution is immediate by conditioning on the initial states \citep[][Section~IV.4]{AndersenGill1993}. \newline
We treat each direct $\ell \to m$ transition with $\ell, m \in \mathcal{R}, \ell \neq m$, separately. For a given $\ell \to m$ transition, we define the counting process $N_{i,\ell \to m}$ for $i=1,\ldots,n$ by
\[ N_{i,\ell \to m}(t) = \# \big\{s \le t:  X_i(s-) = \ell \text{ and } X_i(s) = m \big\}, \]
which counts the number of direct $\ell \to m$ transitions in the time interval $[0,t]$. By Theorem II.6.8 in \cite{AndersenGill1993}, the process $N_{i,\ell \to m}$ has the decomposition
\[ N_{i,\ell \to m}(t) = \int_0^t \lambda_{i,\ell \to m} (s) ds + M_{i,\ell \to m}(t), \]
where $M_{i,\ell \to m}$ is a zero-mean local martingale with respect to the filtration $\{\mathcal{F}_t\}$ generated by the observed data and the intensity process $\lambda_{i,\ell \to m}$ has the form $\lambda_{i,\ell \to m}(s) = Z_{i,\ell}(s) \haz_{\ell \to m}(s)$ with $Z_{i,\ell}(s) = \ind(X_i(s-) = \ell)$ and the transition hazard
\[ \haz_{\ell \to m}(t) = \lim_{\delta \searrow 0} \frac{1}{\delta} \, \pr \big( X_i(t+\delta)=m \, \big| \, X_i(t-)=\ell \big), \]
assuming that the limit in the above display exists. Hence, the processes $N_{i,\ell \to m}$ satisfy model equations \eqref{eq:model-i} and \eqref{eq:model-cum}. Right-censoring and left-truncation can be incorporated as additional states in the model \citep[see Example III.3.3 in][for further details]{AndersenGill1993}. Moreover, covariates can be incorporated, for example, by imposing a Cox model on the intensity process $\lambda_{i,\ell \to m}$ as in Settings \ref{settingB} and \ref{settingC}. For simplicity, we have however ignored covariates in our exposition. \newline
Notably, Settings \ref{settingA}--\ref{settingC} can be formulated as special cases of the multi-state framework just described (with covariates added). Another popular multi-state model in applied survival analysis is the so-called illness-death model without recovery. This model is of major interest in oncology and underlies our application example in Section \ref{sec:app}, where the model is introduced in detail.  
\end{setting}

\enlargethispage{0.25cm}
\section{Estimation strategy}\label{sec:estimation}

We now describe how to construct an estimator of the piecewise constant hazard $\haz: [0,\tau] \to \infty$. In the presence of left-truncation and right-censoring, the data usually become very sparse close to the end points of the interval $[0,\tau]$. Consider for instance a hypothetical clinical study with $0$ and $\tau$ being the starting and closing date of the study, respectively. In such a situation, there are usually very few patients with event times close to $\tau$. In case of delayed study entry, this issue arises (most probably) close to $0$ as well. Hence, in the presence of left-truncation and right-censoring, reliable estimation of $\haz$ is only possible on a subinterval $[\tau_{\min},\tau_{\max}]$ of $[0,\tau]$. Notably, this is not a shortcoming of our methodology. It is rather an issue any (nonparametric) approach has to deal with. For example, it is well-known that in the presence of right-censoring, the Kaplan-Meier estimator becomes very unreliable towards the right-end point $\tau$. Usually, it levels out and the associated confidence intervals become huge, reflecting the sparsity of the data close to $\tau$.


In what follows, we thus restrict attention to estimation of $\haz$ on a suitable subinterval $[\tau_{\min}, \tau_{\max}]$. If there is only right-censoring, we can set $\tau_{\min} = 0$ and choose $\tau_{\max}$ somewhat smaller than $\tau$, leaving us with the interval $[0,\tau_{\max}]$. If there is left-truncation as well, we also need to choose $\tau_{\min}$ somewhat larger than $0$. For convenience, we assume that the interval $[\tau_{\min},\tau_{\max}]$ (i) has length $1$ and (ii) is so large that all change points of $\haz$ lie in its interior, that is, $\tau_k \in (\tau_{\min}, \tau_{\max})$ for all $k=1,\ldots,K$. (i) can always be achieved by re-normalizing the data and (ii) is imposed solely to avoid additional notation specifying which change points lie in $(\tau_{\min},\tau_{\max})$ and which do not. 
In practice, the choice of the estimation interval $[\tau_{\min},\tau_{\max}]$ may (but need not) be guided by expert knowledge. For example, experts may have prior knowledge that change points can only occur within a specific time window $[\tau^*_{\min},\tau^*_{\max}]$ with $0 < \tau^*_{\min} <\tau^*_{\max} < \tau$. In this case, it is natural to choose the estimation window $[\tau_{\min}, \tau_{\max}]$ equal to (or slightly larger than) $[\tau_{\min}^*, \tau_{\max}^*]$.

We now turn to the construction of our estimator of $\haz$. We proceed in several steps:

\textbf{Step 1.} To start with, we require an estimator of the cumulative hazard function $A^*(t) = \int_0^t \haz(s) ds$. We work with the standard Breslow estimator which is constructed as follows: By \eqref{eq:model-cum}, it holds that $d\widebar{N}(t) = \haz(t) \widebar{Z}(t,\bs{\beta}) dt + d\widebar{M}(t)$. With $J(t,\bs{\beta}) = \ind(\widebar{Z}(t,\bs{\beta}) > 0)$ and the convention that $0/0 := 0$, we thus get that
\begin{equation}\label{eq:Breslow1}
\int_0^t \frac{J(s,\bs{\beta})}{\widebar{Z}(s,\bs{\beta})} d\widebar{N}(s) = \int_0^t J(s,\bs{\beta}) \haz(s) ds + \int_0^t \frac{J(s,\bs{\beta})}{\widebar{Z}(s,\bs{\beta})} d\widebar{M}(s).
\end{equation}  
Now suppose we have an estimator $\hat{\bs{\beta}}$ of the parameter vector $\bs{\beta} \in \reals^d$ at hand. How to construct such an estimator depends on the specific setting under consideration. In many cases, it is possible to use a standard estimator off the shelf. In the Cox model of Settings \ref{settingB} and \ref{settingC}, for instance, $\hat{\bs{\beta}}$ can be chosen to be the partial likelihood estimator of $\bs{\beta}$ or an $\ell_1$-penalized version of it. From \eqref{eq:Breslow1}, it follows that
\begin{equation}\label{eq:Breslow2}
\int_0^t \frac{J(s,\hat{\bs{\beta}})}{\widebar{Z}(s,\hat{\bs{\beta}})} d\widebar{N}(s) = \int_0^t \haz(s) ds + \Delta^J(t) + \Delta^\beta(t) + \eta(t),
\end{equation}  
where $A^*(t) = \int_0^t \haz(s) ds$ is the cumulative hazard and 
\begin{align*}
\Delta^J(t) & = \int_0^t (J(s,\bs{\beta})-1) \haz(s) ds \\ 
\Delta^\beta(t) & = \int_0^t \frac{J(s,\hat{\bs{\beta}})}{\widebar{Z}(s,\hat{\bs{\beta}})} d\widebar{N}(s) - \int_0^t \frac{J(s,\bs{\beta})}{\widebar{Z}(s,\bs{\beta})} d\widebar{N}(s) \\ 
\eta(t) & = \int_0^t \frac{J(s,\bs{\beta})}{\widebar{Z}(s,\bs{\beta})} d\widebar{M}(s).
\end{align*}
The left-hand side of \eqref{eq:Breslow2},
\[ \hat{A}(t) := \int_0^t \frac{J(s,\hat{\bs{\beta}})}{\widebar{Z}(s,\hat{\bs{\beta}})} d\widebar{N}(s), \]
is the Breslow estimator of $A^*(t)$. As integrating over a counting process is the same as summing the integrand over the jump times of the process, it can be written as
$\hat{A}(t) = \sum_{\{k: S_k \le t\}} J(S_k,\hat{\bs{\beta}}) / \widebar{Z}(S_k,\hat{\bs{\beta}})$, 
where $S_1 < S_2 < S_3 < \ldots$ are the jump times of $\widebar{N}$. According to \eqref{eq:Breslow2}, the distance between the Breslow estimator $\hat{A}(t)$ and the cumulative hazard $A^*(t)$ is the sum of three error terms: the bias term $\Delta^J(t)$, the approximation error $\Delta^\beta(t)$ which results from estimating $\bs{\beta}$ by $\hat{\bs{\beta}}$ and the term $\eta(t)$ which (by using standard theory for counting processes) can be shown to be a zero-mean local martingale.

\textbf{Step 2.} We next derive a simple regression model for the increments of the Breslow estimator $\hat{A}$. Cover the interval $[\tau_{\min},\tau_{\max}]$ with an equidistant grid of step length $1/n$. The grid points are given by $t_j = \tau_{\min} + j/n$ for $j=0,\ldots,n$ (taking into account that the length of $[\tau_{\min},\tau_{\max}]$ is $1$). Now consider the increments
\[ Y_j := \frac{\hat{A}(t_j) - \hat{A}(t_{j-1})}{t_j - t_{j-1}} = n \{\hat{A}(t_j) - \hat{A}(t_{j-1})\}. \]
Using \eqref{eq:Breslow2}, we immediately obtain that 
\begin{equation}\label{eq:signal-plus-noise-model}
Y_j = \haz_j + u_j \quad \text{with} \quad u_j = \Delta_j^\alpha + \Delta_j^J + \Delta_j^\beta + \eta_j
\end{equation}  
and $\haz_j = \haz(t_j)$ for $j=1,\ldots,n$, where 
\begin{align*}
\Delta_j^\alpha & = \frac{\int_{t_{j-1}}^{t_j} \haz(s) ds}{t_j - t_{j-1}} - \haz(t_j) \\
\Delta_j^\ell & = \frac{\Delta^\ell(t_j) - \Delta^\ell(t_{j-1})}{t_j - t_{j-1}} \quad \text{for } \ell \in \{J,\beta\} \\
\eta_j & = \frac{\eta(t_j) - \eta(t_{j-1})}{t_j - t_{j-1}}. 
\end{align*}
Hence, the increments $Y_j$ follow a simple signal-plus-noise model, where the signal $\haz_j$ is the hazard rate $\haz(t_j)$ at time point $t_j = \tau_{\min} + j/n$ and $u_j$ is the noise term. Put differently, the signal vector $\hazvec =(\haz_1,\ldots,\haz_n)^\top$ in model \eqref{eq:signal-plus-noise-model} is identical to the discretized hazard rate $\hazvec = (\haz(t_1),\ldots,\haz(t_n))^\top$. Notably, the variables $Y_j$, $\haz_j$ and $u_j$ in model \eqref{eq:signal-plus-noise-model} depend on $n$. A more precise model formulation thus reads $Y_{j,n} = \haz_{j,n} + u_{j,n}$ with $\haz_{j,n} = \haz(\tau_{\min} + j/n)$. For ease of notation, we however suppress the dependence of $Y_j$, $\haz_j$ and $u_j$ on $n$ in what follows.

\textbf{Step 3.} By assumption, the discretized hazard rate $\hazvec = (\haz(t_1),\ldots,\haz(t_n))^\top$ is a piecewise constant vector. In order to estimate $\hazvec$ in model \eqref{eq:signal-plus-noise-model}, we may thus use techniques to recover a piecewise constant signal in a regression model. One such technique is the fused lasso. Applying the fused lasso to the constructed data sample $\{Y_1,\ldots,Y_n\}$ yields the estimator $\hat{\bs{\alpha}}_\lambda = (\hat{\alpha}_{\lambda,1},\ldots,\hat{\alpha}_{\lambda,n})^\top$ defined by
\begin{equation}\label{eq:fused-lasso}
\hat{\bs{\alpha}}_\lambda \in \underset{\bs{a} \in \reals^n}{\textnormal{argmin}} \bigg\{ \frac{1}{n} \sum_{j=1}^n (Y_j - a_j)^2 + \lambda \sum_{j=2}^n |a_j - a_{j-1}| \bigg\},
\end{equation} 
where $\lambda > 0$ is a tuning parameter. An estimator of the hazard rate $\haz$ on the interval $[\tau_{\min},\tau_{\max}]$ is immediately obtained from $\hat{\bs{\alpha}}_\lambda$ by constant interpolation. Specifically, we define the estimator $\hat{\alpha}_\lambda: [\tau_{\min},\tau_{\max}] \to \reals$ by setting  
\begin{equation}\label{eq:fused-lasso-interpolated} 
\hat{\alpha}_\lambda(t) =
\begin{cases}
\hat{\alpha}_{\lambda,1} & \text{for } t \in [t_0,t_1) \\
\hat{\alpha}_{\lambda,j} & \text{for } t \in [t_j, t_{j+1}) \text{ and } j \in \{1,\ldots,n-1\} \\
\hat{\alpha}_{\lambda,n} & \text{for } t = t_n
\end{cases}
\end{equation}
with $t_j = \tau_{\min} + j/n$.

\section{Theoretical results}\label{sec:theory}

In what follows, we investigate the estimator $\hat{\bs{\alpha}}_\lambda$ (and thus also implicitly the interpolated version $\hat{\alpha}_\lambda$) from a theoretical point of view. More specifically, we derive results on the convergence rate of $\hat{\bs{\alpha}}_\lambda$ as well as results on how well the change points of the hazard rate $\haz$ are approximated by those of $\hat{\bs{\alpha}}_\lambda$. Sections \ref{subsec:ass} and \ref{subsec:theory-general} provide the theory for the general counting process model of Section \ref{sec:model}. In Section \ref{subsec:theory-settings}, we verify that the conditions required for the general theory are satisfied in Settings \ref{settingA}--\ref{settingD} under standard assumptions. All proofs are relegated to the Appendix. 
Asymptotic statements are always to be understood in the sense that the sample size $n$ tends to infinity. Notably, the processes $\widebar{N}$, $\widebar{M}$ and $\widebar{Z}$ as well as the processes derived from them (in particular, $J$, $\Delta^J$, $\Delta^\beta$ and $\eta$) and the elements of the regression model \eqref{eq:signal-plus-noise-model} (namely, $Y_j$, $\haz_j$ and $u_j$) all depend on the sample size $n$. To keep the notation concise, we however suppress their dependence on $n$ throughout.

\subsection{Assumptions}\label{subsec:ass}

We impose the following assumptions on the general model from Section \ref{sec:model}:
\begin{enumerate}[label=(C\arabic*),leftmargin=1cm]
 
\item \label{C0} The filtered probability space $(\Omega,\mathcal{F},\{\mathcal{F}_t\}_{t \ge 0}, \pr)$ satisfies the ``usual conditions''.
  
\item \label{C1} The counting processes $N_i$ are i.i.d.\ across $i$. 

\item \label{C2} For each $i$, the compensator $\Lambda_i$ of the process $N_i$ has the form $\Lambda_i(t) = \int_0^t \lambda_i(s) ds$ with some $\{\mathcal{F}_t\}$-predictable process $\lambda_i$. Moreover, the process $\lambda_i$ has the multiplicative structure \eqref{eq:multiplicative-intensity-model}, that is, $\lambda_i(t) = Z_i(t,\bs{\beta}) \haz(t)$ with a non-negative deterministic function $\haz$ and an $\{\mathcal{F}_t\}$-predictable process $Z_i(\cdot,\bs{\beta})$. 
  
\item \label{C3} The processes $\widebar{Z}(\cdot,\bs{\beta})$ and $J(\cdot,\bs{\beta})/\widebar{Z}(\cdot,\bs{\beta})$ are locally bounded for each $n$.
  
\item \label{C4} It holds that $\pr ( \inf_{t \in [\tau_{\min},\tau_{\max}]} J(t,\bs{\beta}) = 0 ) = o(1)$. 

\item \label{C5} For some natural number $\nu > 4$, 
\[  \ex \left[ \bigg\| \frac{J(\cdot,\bs{\beta})}{\widebar{Z}(\cdot,\bs{\beta})} \bigg\|_{\infty}^{\nu} \right] = \ex \left[ \sup_{t \in [\tau_{\min},\tau_{\max}]} \bigg| \frac{J(t,\bs{\beta})}{\widebar{Z}(t,\bs{\beta})} \bigg|^{\nu} \right] \le \frac{C_{\infty,\nu}}{n^\nu}, \]
where $C_{\infty,\nu} < \infty$ is a fixed constant independent of $n$.

\item \label{C6} The approximation error $\Delta^\beta(t)$ in \eqref{eq:Breslow2} is such that 
\[ \max_{1 \le j \le n} \bigg| \frac{\Delta^\beta(t_j) - \Delta^\beta(t_{j-1})}{t_j - t_{j-1}} \bigg| = O_p(\rho_n), \]
where $\rho_n = o(n^{-\xi})$ for some arbitrarily small but fixed $\xi > 0$. 

\end{enumerate}
\ref{C0}--\ref{C3} are standard conditions in the context of statistical models based on counting processes. \ref{C4}--\ref{C6} are quite abstract high-order conditions. We show below that they are fulfilled in Settings \ref{settingA}--\ref{settingD} under common assumptions.

\subsection{General results}\label{subsec:theory-general}

Our first theoretical result bounds the distance between the estimator $\hat{\bs{\alpha}}_\lambda$ and the discretized hazard rate $\hazvec$. More specifically, it provides a bound on $|\hat{\alpha}_{\lambda,j} - \haz_j|$ for any $j \in \{1,\ldots,\ngrid\}$. To formulate the result, we use the following notation: For $k \in \{1,\ldots,K\}$, we let $n_k = \lceil (\tau_k - \tau_{\min}) \ngrid \rceil$ denote the index where the discretized hazard rate $\hazvec$ has the $k$-th jump. Put differently, $n_1 < \ldots < n_K$ are the $K$ indices with $\haz_{n_k-1} \ne \haz_{n_k}$. (Note that $1 < n_1$ and $n_K < n$ for sufficiently large $n$ because the change points of $\haz$ all lie in the interior of $[\tau_{\min},\tau_{\max}]$ by assumption.) In what follows, we call $n_1,\ldots,n_K$ the jump indices of the discretized hazard $\hazvec$. For notational convenience, we additionally set $\ngrid_0 = 1$ along with $\ngrid_{K+1} = \ngrid + 1$. For any $j$, we let $n_{k(j)}$ and $n_{k(j)+1}$ be the two jump indices with $n_{k(j)} \le j \le n_{k(j)+1}-1$ and define $d_j = \min\{j+1-n_{k(j)},n_{k(j)+1}-j\}$, which essentially gives the distance of $j$ to the nearest jump index. Finally, we let $r_{k(j)} = n_{k(j)+1} - n_{k(j)}$ be the length of the interval between the $k(j)$-th and the $(k(j)+1)$-th jump index. 
\begin{theorem}\label{theo:elementwise-bound} 
Let \ref{C0}--\ref{C6} be satisfied and let $\kappa_n = c_n \max \{ \ngrid^{2/\nu}, \sqrt{\ngrid} \rho_n \}$ with $\nu > 4$ from \ref{C5}, where $\{c_n\}$ is a slowly diverging sequence (e.g.\ $c_n = c_0 \log \log n$ with some constant $c_0 > 0$). Then for any positive value of the regularization parameter $\lambda$, the following holds with probability tending to $1$:
\[ |\hat{\alpha}_{\lambda,j} - \haz_j| \le \max \left\{ \frac{\kappa_n}{\sqrt{d_j}}, \frac{\kappa_n^2}{4\ngrid\lambda},\frac{2\ngrid\lambda}{r_{k(j)}} + \frac{2\kappa_n}{\sqrt{r_{k(j)}}} \right\} \quad \text{for } j \in \{1,\ldots,\ngrid\}. \]
\end{theorem}

The elementwise bound of Theorem \ref{theo:elementwise-bound} allows us to derive the following bound on the $\ell_2$-error $\norm{\hat{\bs{\alpha}}_\lambda - \hazvec}_2^2/n$.
\begin{theorem}\label{theo:l2-bound}
Let \ref{C0}--\ref{C6} be satisfied and choose $\lambda = \kappa_n / \sqrt{\ngrid}$, where $\kappa_n$ is defined as in Theorem \ref{theo:elementwise-bound}. Then 
\[ \frac{1}{n} \norm{\hat{\bs{\alpha}}_\lambda - \hazvec}_2^2 = O_p \Big( c_n^2 \log(n) \max\big\{  n^{-1 + \frac{4}{\nu}}, \rho_n^2 \big\} \Big). \]
\end{theorem}
We briefly gives some remarks on Theorem \ref{theo:l2-bound}:
\begin{enumerate}[label=(\alph*),leftmargin=0.7cm]
\item As we will see in our analysis of Settings \ref{settingA}--\ref{settingD}, a leading case is that $\nu$ can be chosen as large as desired and $\rho_n = O(\log(n)/\sqrt{n})$. In this case, Theorem \ref{theo:l2-bound} yields that
\[ \frac{1}{n} \norm{\hat{\bs{\alpha}}_\lambda - \hazvec}_2^2 = O_p \Big( \frac{1}{n^{1-\xi}} \Big) \]
with $\xi > 0$ arbitrarily small but fixed. This is a fast convergence rate not far from optimal \citep[see the discussion in][for details on sharp rates for the fused lasso]{LinSharpnackRinaldoTibshirani2017}, which suggests that the general rate established in Theorem \ref{theo:l2-bound} is quite sharp. 
\item Theorem \ref{theo:l2-bound} can alternatively be formulated in terms of the hazard function $\haz: [\tau_{\min},\tau_{\max}] \to \reals$ rather than the discretized hazard rate $\hazvec$. To do so, let $\hat{\alpha}_\lambda: [\tau_{\min},\tau_{\max}] \to \reals$ be the interpolated version of $\hat{\bs{\alpha}}_\lambda$ defined in \eqref{eq:fused-lasso-interpolated}. Then Theorem \ref{theo:l2-bound} and some straightforward calculations imply that 
 \[ \int_{\tau_{\min}}^{\tau_{\max}} \{\hat{\alpha}_\lambda(t) - \haz(t)\}^2 dt = O_p \Big( c_n^2 \log(n) \max\big\{  n^{-1 + \frac{4}{\nu}}, \rho_n^2 \big\} \Big). \]
\end{enumerate}

We next turn to change point estimation. The change points of the estimator $\hat{\bs{\alpha}}_\lambda$ are given by the set 
\[ \hat{\mathcal{S}}_\lambda = \left\{ \tau_{\min} + \frac{j}{n} \in [\tau_{\min},\tau_{\max}]: \hat{\alpha}_{\lambda,j-1} \neq \hat{\alpha}_{\lambda,j} \text{ for } j \in \{2,\ldots,n\} \right\}. \]
Our goal is to understand how well the change points $\mathcal{S}^* = \{\tau_1,\ldots,\tau_K\}$ of the hazard rate $\haz$ are approximated by those in $\hat{\mathcal{S}}_\lambda$. Informally speaking, we show that the following two statements hold with high probability:
\begin{enumerate}[label=(CP\arabic*),leftmargin=1.25cm]
\item \label{eq:change-point-res-I} For each change point in $\mathcal{S}^*$, there exists at least one change point in $\hat{\mathcal{S}}_\lambda$ ``close'' to it. 
\item \label{eq:change-point-res-II} Let $\hat{\tau} \in \hat{\mathcal{S}}_\lambda$ be a change point which is ``far away'' from any change point in $\mathcal{S}^*$. Then the jump size of $\hat{\bs{\alpha}}_\lambda$ at the change point $\hat{\tau}$ is ``small''. 
\end{enumerate}
According to \ref{eq:change-point-res-I}, $\hat{\bs{\alpha}}_\lambda$ has a jump close to each change point of $\haz$. In general, however, we cannot guarantee that it has exactly one jump close to each change point. This means that our fused lasso estimator may reconstruct a change point in $\haz$ by multiple jumps. Figure \ref{fig:illustration} gives an illustration. The black line is the true hazard rate $\haz$, which has a single change point at $0.5$, and the red line is (the interpolated version $\hat{\alpha}_\lambda$ of) the estimator $\hat{\bs{\alpha}}_\lambda$. As can be seen, $\hat{\bs{\alpha}}_\lambda$ has two jumps close to the time point $0.5$ which together reconstruct the change point of $\haz$. According to \ref{eq:change-point-res-II}, $\hat{\bs{\alpha}}_\lambda$ may have additional jumps far away from any change point of $\haz$. However, these additional jumps are negligible in the sense of being small. Figure \ref{fig:illustration} once again gives an illustration. The estimator $\hat{\bs{\alpha}}_\lambda$ (in red) has a jump between $0.8$ and $0.9$. Nevertheless, it yields a good overall reconstruction of the step function $\haz$ as this additional jump is small. In principle, it is possible to get rid of additional jumps far away from the true change points by applying post-processing methods to the fused lasso \citep[see e.g.\ the Haar wavelet based method in][]{LinSharpnackRinaldoTibshirani2017}. However, these methods are rather of theoretical than practical value as they depend on additional tuning parameters which are extremely hard to choose in practice. We thus do not consider any such methods. 
In summary, \ref{eq:change-point-res-I} and \ref{eq:change-point-res-II} imply that the estimator $\hat{\bs{\alpha}}_\lambda$ produces an accurate reconstruction of the piecewise constant hazard $\haz$, even though it may have more change points than $\haz$ in general. Hence, our method is a suitable tool for (parsimonious) fitting of piecewise constant hazard curves, which is the main goal in many applications in time-to-event analysis (rather than precise estimation and interpretation of location and number of change points). 

\begin{figure}
\centering
\includegraphics[width=0.35\textwidth]{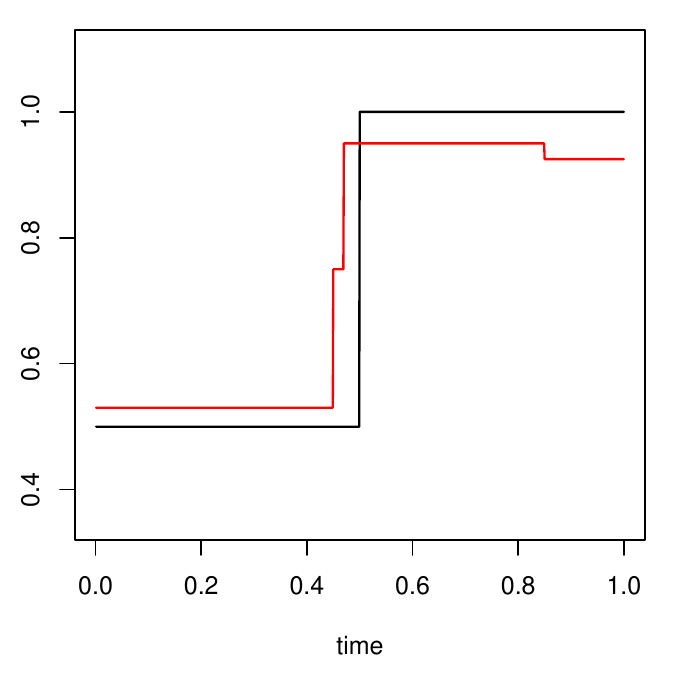}
\caption{Example of a piecewise constant hazard rate $\haz$ (in black) and the fused lasso estimator $\hat{\alpha}_\lambda$ (in red).}\label{fig:illustration}
\end{figure}

We now formalize the statements \ref{eq:change-point-res-I} and \ref{eq:change-point-res-II}. To do so, we measure the distance between two discrete sets $A$ and $B$ by 
\begin{equation*}
d(A \, | \, B) = \max_{b \in B} \min_{a \in A} |a-b|. 
\end{equation*}
The quantity $d(A \, | \, B)$ gives the maximal distance of an element in $B$ to its closest element in $A$. Hence, it is small if for any element in $B$, there is an element in $A$ close to it.
In order to formalize \ref{eq:change-point-res-I}, we analyze the quantity $d(\hat{\mathcal{S}}_\lambda \, | \, \mathcal{S}^*)$.
The following theorem shows that $d(\hat{\mathcal{S}}_\lambda \, | \, \mathcal{S}^*)$ is small for large $n$ in the following sense:
\begin{theorem}\label{theo:change-point-I}
Under the conditions of Theorem \ref{theo:l2-bound}, it holds that
\[ d(\hat{\mathcal{S}}_\lambda \, | \, \mathcal{S}^*) = O_p  \Big( c_n^2 \log(n) \max\big\{  n^{-1 + \frac{4}{\nu}}, \rho_n^2 \big\} \Big). \]
\end{theorem}
In the leading case where $\nu$ can be chosen as large as desired and $\rho_n = O(\log(n)/\sqrt{n})$, Theorem \ref{theo:change-point-I} in particular yields that
\[ d(\hat{\mathcal{S}}_\lambda \, | \, \mathcal{S}^*) = O_p\Big(\frac{1}{n^{1-\xi}}\Big) \]
with $\xi > 0$ arbitrarily small but fixed. Hence, there exists a small neighbourhood of order $1/n^{1-\xi}$ around each change point in $\mathcal{S}^*$ which contains an element of $\hat{\mathcal{S}}_\lambda$. The next theorem formalizes \ref{eq:change-point-res-II}.  
\begin{theorem}\label{theo:change-point-II}
Let the conditions of Theorem \ref{theo:l2-bound} be satisfied, let $\{\Delta_n\}$ be any sequence of positive real numbers and let $\hat{\mathcal{S}}_\lambda^{\textnormal{far}}$ be the set of all change points $\hat{\tau} \in \hat{\mathcal{S}}_\lambda$ with $\min_{1\le k \le K} |\hat{\tau} - \tau_k| \ge \Delta_n/n$. Then with probability tending to $1$, 
\[ |\hat{\alpha}_{\lambda,(\hat{\tau}-\tau_{\min})n} - \hat{\alpha}_{\lambda,(\hat{\tau}-\tau_{\min})n-1}| \le C \frac{\kappa_n}{\sqrt{\Delta_n}} \quad \text{for all } \hat{\tau} \in \hat{\mathcal{S}}_\lambda^{\textnormal{far}}, \]
where $C$ is a sufficiently large constant. 
\end{theorem}
Theorem \ref{theo:change-point-II} essentially says the following: if the distance of a change point $\hat{\tau} \in \hat{\mathcal{S}}_\lambda$ to any change point $\tau_k$ of the hazard rate $\haz$ is larger than $\Delta_n/n$, then the estimator $\hat{\bs{\alpha}}_\lambda$ cannot have a jump larger than $C \kappa_n/\sqrt{\Delta_n}$ at the change point $\hat{\tau}$, i.e., at the jump index $(\hat{\tau}-\tau_{\min})n$. In the leading case where $\nu$ can be chosen as large as desired and $\rho_n = O(\log(n)/\sqrt{n})$, Theorem \ref{theo:change-point-II} simplifies slightly to: with probability tending to $1$,
\[ |\hat{\alpha}_{\lambda,(\hat{\tau}-\tau_{\min})n} - \hat{\alpha}_{\lambda,(\hat{\tau}-\tau_{\min})n-1}| \le C \frac{n^\xi}{\sqrt{\Delta_n}} \quad \text{for all } \hat{\tau} \in \hat{\mathcal{S}}_\lambda^{\textnormal{far}} \]
with $\xi > 0$ arbitrarily small but fixed. This in particular implies the following: if $\hat{\tau}$ remains bounded away from any change point $\tau_k$ of $\haz$, that is, if $\min_{1\le k \le K} |\hat{\tau} - \tau_k| > c > 0$ for some $c > 0$ (with probability tending to $1$), then the corresponding jump size $|\hat{\alpha}_{\lambda,(\hat{\tau}-\tau_{\min})n} - \hat{\alpha}_{\lambda,(\hat{\tau}-\tau_{\min})n-1}|$ is smaller than $C n^\xi/\sqrt{n}$ (with probability tending to $1$).

\subsection{Results for Settings \ref{settingA}--\ref{settingD}}\label{subsec:theory-settings}

We now show that the general theory from the previous Section \ref{subsec:theory-general} applies to Settings \ref{settingA}--\ref{settingD}. To do so, we verify that the high-order conditions \ref{C0}--\ref{C6} are fulfilled in these four settings under standard assumptions. As Setting \ref{settingA} is nested as a special case in the Cox model of Setting \ref{settingB}, we directly start with Setting \ref{settingB} and do not treat Setting \ref{settingA} separately. A formal result for Setting \ref{settingA} follows as a simple corollary to Proposition \ref{prop:settingB} below.

The Cox model of Setting \ref{settingB} incorporates a vector of $d$ covariates. Rather than restricting attention to the standard low-dimensional case where the dimension $d$ is small and fixed, we consider the following high-dimensional setup: (i) The dimension $d$ may be much larger than the sample size $n$. In particular, $d$ may grow as any polynomial of $n$, that is, $d = O(n^a)$ with $a$ arbitrarily large but fixed. (ii) The parameter vector $\bs{\beta}$ is such that $\|\bs{\beta}\|_1$ remains bounded as $n \to \infty$.  In principle, it is possible to allow $\| \bs{\beta} \|_1$ to grow (sufficiently slowly) with $n$. However, to avoid certain technical complications, we make the stronger assumption that it is bounded. Condition (ii) can be regarded as a sparsity constraint. It is in particular fulfilled if the number of non-zero coefficients of $\bs{\beta}$ does not grow with the sample size and the entries of $\bs{\beta}$ are all bounded by a fixed constant $C < \infty$ in absolute value.
\begin{prop}\label{prop:settingB}
Consider the Cox model from Setting \ref{settingB} and assume the following:   
\begin{enumerate}[label=\textnormal{(C$_B$\arabic*)},leftmargin=1.25cm]
\item \label{C1:settingB} The data points $(T_i,\delta_i,\bs{W}_i)$ are i.i.d.\ across $i$.  
\item \label{C2:settingB} $\max_{1 \le k \le n} \pr(t_{k-1} < T_i^* \le t_k) \le C/n$ with some sufficiently large constant $C$.
\item \label{C3:settingB} The components $N_{i}$ of the multivariate counting process $N_{1:n} = (N_{1},\ldots,N_{n})$ have $\{\mathcal{F}_t\}$-intensities $\lambda_{i}(t) = \ind(t \le T_i) \exp(\bs{\beta}^\top \bs{W}_i) \haz(t)$. 
\item \label{C4:settingB}$\pr(T_i \ge \tau_{\max}) = p$ with some constant $p \in (0,1)$.
\item \label{C5:settingB} The dimension $d$ of the covariate vector $\bs{W}_i$ grows at most polynomially in the sample size $n$, that is, $d = O(n^a)$ with some arbitrarily large but fixed constant $a$. Moreover, the covariates $\bs{W}_i$ have bounded support, that is, $\| \bs{W}_i \|_\infty \le C_W$ for all $i$ and some constant $C_W < \infty$. 
\item \label{C6:settingB} It holds that $\| \bs{\beta} \|_1 \le C_\beta < \infty$ for some constant $C_\beta < \infty$ and $\hat{\bs{\beta}}$ is any estimator of $\bs{\beta}$ with the property that $\|\hat{\bs{\beta}} - \bs{\beta}\|_1 = O_p(\sqrt{\log(n) / n})$. 
\end{enumerate}
Then \ref{C0}--\ref{C6} are satisfied with $\nu$ as large as desired, $\tau_{\min} = 0$ and $\rho_n = \log(n) /\sqrt{n}$.   
\end{prop}
We briefly comment on the assumptions of Proposition \ref{prop:settingB}. The i.i.d.\ assumption in \ref{C1:settingB} is standard in the literature. \ref{C2:settingB} is rather mild as well. It is in particular fulfilled if the survival times $T_i^*$ have a bounded density $f^*$. \ref{C3:settingB} is merely the definition of the Cox model. \ref{C4:settingB} requires that the probability of observing a right-censored survival time $T_i$ (weakly) larger than $\tau_{\max}$ is strictly positive. This ensures that the data do not become too sparse towards the right end-point of the interval $[0,\tau_{\max}]$. Hence, \ref{C4:settingB} essentially encapsulates that the data points $T_i$ lie sufficiently dense in the interval $[0,\tau_{\max}]$ to perform reasonable estimation. The first part of \ref{C5:settingB} restricts the dimension $d$ of the covariate vector as already discussed above. The second part imposes boundedness on the covariates, which is a very common assumption in the literature and not very problematic in practice as most covariates can be regarded as having bounded support. As shown in \cite{Huang2013}, \ref{C6:settingB} is satisfied (under appropriate regularity conditions not spelt out here) by an $\ell_1$-penalized version (i.e., a lasso version) of the partial likelihood estimator of $\bs{\beta}$. In the low-dimensional case where $d$ does not grow with $n$, one may of course use the standard (unpenalized) partial likelihood estimator, which has the property that $\|\hat{\bs{\beta}} - \bs{\beta}\|_1 = O_p(1/\sqrt{n})$ (once again under certain regularity conditions not spelt out here).

We next turn to the competing risks model of Setting \ref{settingC}, which allows for both right-censoring and left-truncation. In the presence of left-truncation, or put differently, delayed study entry, the data points $(T_i, \delta_i \cdot \varepsilon_i, L_i, \bs{W}_i)$ are drawn from the conditional probability distribution given study entry. Technically speaking, this means that the underlying probability space is not $(\Omega, \mathcal{F}, \pr)$ as in Settings \ref{settingA} and \ref{settingB} but rather the space $(\Omega,\mathcal{F},\pr^B)$, where $\pr^B$ is the conditional measure given study entry defined by $\pr^{B}(A) = \pr(A \cap B)/\pr(B)$ for all $A \in \mathcal{F}$ with $B = \bigcap_{i=1}^n \{T_i^* > L_i\}$. We thus replace the unconditional measure $\pr$ by the conditional one $\pr^B$. In particular, all probabilistic statements in the following proposition are with respect to $\pr^B$.
\begin{prop}\label{prop:settingC}
Consider a specific event $r \in \{1,\ldots,R\}$ in the competing risks model from Setting \ref{settingC} and assume the following:   
\begin{enumerate}[label=\textnormal{(C$_C$\arabic*)},leftmargin=1.25cm]
\item \label{C1:settingC} The data points $(T_i, \delta_i \cdot \varepsilon_i, L_i, \bs{W}_i)$ are i.i.d.\ across $i$. 
\item \label{C2:settingC} $\max_{1 \le k \le n} \pr^B(t_{k-1} < T_i^* \le t_k) \le C/n$ with some sufficiently large constant $C$.
\item \label{C3:settingC} The components $N_{ir}$ of the multivariate counting process $N_{1:n,r} = (N_{1r},\ldots,N_{nr})$ have $\{\mathcal{F}_t\}$-intensities $\lambda_{ir}(t) = \ind(L_i< t\le T_i) \exp(\bs{\beta}_r^\top \bs{W}_i) \haz_r(t)$. 
\item \label{C4:settingC} $\pr^B(L_i < \tau_{\min} \land T_i \ge \tau_{\max}) = p$ with some constant $p \in (0,1)$.
\item \label{C5:settingC} The dimension $d$ of the covariate vector $\bs{W}_i$ grows at most polynomially in the sample size $n$, that is, $d = O(n^a)$ with some arbitrarily large but fixed constant $a$. Moreover, the covariates $\bs{W}_i$ have bounded support, that is, $\| \bs{W}_i \|_\infty \le C_W$ for all $i$ and some constant $C_W < \infty$. 
\item \label{C6:settingC} It holds that $\| \bs{\beta}_r \|_1 \le C_\beta < \infty$ for some constant $C_\beta < \infty$ and $\hat{\bs{\beta}}_r$ is any estimator of $\bs{\beta}_r$ with the property that $\|\hat{\bs{\beta}}_r - \bs{\beta}_r\|_1 = O_p(\sqrt{\log(n) / n})$. 
\end{enumerate}
Then \ref{C0}--\ref{C6} are satisfied for the event $r$ with $\nu$ as large as desired and $\rho_n = \log(n) /\sqrt{n}$.
\end{prop} 
The assumptions of Proposition \ref{prop:settingC} closely parallel those required for the analysis of Setting \ref{settingB} in Proposition \ref{prop:settingB}. We thus only give a brief comment on \ref{C4:settingC}. Compared to \ref{C4:settingB}, the statement in \ref{C4:settingC} is a bit more complicated. This stems from the fact that we now also need to deal with left-truncation. However, the underlying idea is the same: \ref{C4:settingC} ensures that the data on the interval $[\tau_{\min},\tau_{\max}]$ are sufficiently dense to perform reasonable estimation.

We finally analyze the Markov model with state space $\mathcal{R} = \{0,\ldots,R\}$ from Setting \ref{settingD}.  
\begin{prop}\label{prop:settingD}
Consider a specific $\ell \to m$ transition with $\ell,m \in \mathcal{R}$, $\ell \ne m$, in the multi-state model from Setting \ref{settingD} and assume the following:  
\begin{enumerate}[label=\textnormal{(C$_D$\arabic*)},leftmargin=1.25cm]
\item \label{C1:settingD} For each $i$, $X_i$ is a time-continuous Markov process with absolutely continuous intensity measure, right-continuous sample paths and $\pr(X_i(0) = 0) = 1$. The processes $X_i$ are i.i.d.\ across $i$. 
\item \label{C2:settingD} It holds that $\pr(X_i(t-) = \ell \text{ for all } t \in [\tau_{\min},\tau_{\max}]) = p_\ell$ with some constant $p_\ell \in (0,1)$.
\end{enumerate}
Then \ref{C0}--\ref{C5} are satisfied for the considered $\ell \to m$ transition with $\nu$ as large as desired. Moreover, \ref{C6} is trivially satisfied with $\rho_n = 0$. 
\end{prop}
As the assumptions in \ref{C1:settingD} are standard in the literature, we only comment on \ref{C2:settingD}. According to \ref{C2:settingD}, with positive probability, subject $i$ is in state $\ell$ at (or more precisely, just before) time point $t$ for any $t \in [\tau_{\min},\tau_{\max}]$. This ensures that there are sufficiently many subjects in the risk set at any $t \in [\tau_{\min},\tau_{\max}]$ to perform estimation. \ref{C2:settingD} thus plays the same role as \ref{C4:settingB} and \ref{C4:settingC} in the results for Settings \ref{settingB} and \ref{settingC}.

\section{Implementation}\label{sec:impl}

The estimator $\hat{\bs{\alpha}}_\lambda$ depends on the tuning parameter $\lambda$. In what follows, we derive a data-driven rule for selecting $\lambda$ in practice. To do so, we exploit the fact that the fused lasso can be expressed in standard lasso form. This allows us to adapt tuning parameter calibration techniques developed for the standard lasso.

Cross-validation is presumably the most prominent technique for selecting the tuning parameters of lasso-type estimators. In our context, however, it should be treated with caution. The main reason is that our data (i.e., the constructed variables $Y_j$) are not independent but have a quite complicated dependence structure. As is well-known \citep[see e.g.][]{Altman1990, OpsomerWangYang2001, RabinowiczRosset2022}, cross-validation may perform very poorly when the data are dependent, and adjustments are needed to deal with the specific dependence structure at hand. Hence, simply using a cross-validation algorithm off-the-shelf will most likely lead to unsatisfactory results in our case. Indeed, we have run different versions of cross-validation in our simulation exercises and have persistently obtained extremely bad results.

We thus pursue a different strategy for tuning parameter calibration. Specifically, we work with techniques based on the lasso's effective noise. In a standard linear regression model, the lasso's effective noise is given by $2 \| X^\top u \|_\infty / n$, where $X$ denotes the design matrix and $u$ the error vector of the linear model. Theory for the lasso suggests to choose the tuning parameter $\lambda$ as a high quantile of the effective noise; see e.g. \cite{BelloniChernozhukov2013} and \cite{VogtLederer2021}. These quantiles, however, are usually not known in practice (as the distribution of the error vector $u$ is unknown). Nevertheless, they can be approximated in many cases. Hence, a general strategy for selecting $\lambda$ in our framework is this: Reformulate the fused-lasso-based hazard estimator in standard lasso form, approximate a high quantile (say the $90\%$-quantile) of the effective noise of the reformulated estimator and set the tuning parameter $\lambda$ equal to the computed quantile.
The technical details to implement this strategy are as follows: 
\begin{enumerate}[label=(\roman*),leftmargin=0.85cm]

\item We first repara\-metrize the model $\bs{Y} = \hazvec + \bs{u}$, where $\bs{Y} = (Y_1,\ldots,Y_n)^\top$, $\hazvec = (\haz_1,\ldots,\haz_n)^\top$ and $\bs{u} = (u_1,\ldots,u_n)^ \top$. In particular, we set $\theta_1 = \haz_1$ and $\theta_j = \haz_j - \haz_{j-1}$ for $2 \le j \le n$. This implies that $\hazvec = X \bs{\theta}$, where $\bs{\theta} = (\theta_1,\ldots,\theta_n)^\top$ and the design matrix $X = (X_{ij}: 1 \le i,j \le n)$ has the entries $X_{ij} = \ind(i \ge j)$. With this reparametrization, we can reformulate the model as $\bs{Y} = X\bs{\theta} + \bs{u}$. The reparametrized fused lasso is
\begin{equation}\label{eq:min-theta-1}
\hat{\bs{\theta}}_\lambda = (\hat{\theta}_{\lambda,1},\ldots,\hat{\theta}_{\lambda,n})^\top \in \underset{\bs{\vartheta} \in \reals^n}{\textnormal{argmin}} \Big\{ \frac{1}{n} \| \bs{Y} - X \bs{\vartheta} \|_2^2 + \lambda \| \bs{\vartheta}_{-1} \|_1 \Big\},
\end{equation}
where $\bs{\vartheta}_{-1}$ denotes the vector $\bs{\vartheta}$ without the first component. The estimator $\hat{\bs{\theta}}_\lambda$ differs from the standard lasso only in that the first coefficient is not penalized (that is, the penalty is $\| \bs{\vartheta}_{-1}\|_1$ rather than $\|\bs{\vartheta}\|_1$). Next let $\bs{Y}^{\text{c}} = (Y_1^{\text{c}},\ldots,Y_n^{\text{c}})^\top$ be defined by $Y_i^{\text{c}} = Y_i - \widebar{Y}$ with $\widebar{Y} = n^{-1} \sum_{i=1}^n Y_i$ and let $X^{\text{c}} = (\bs{X}_2^{\text{c}},\ldots,\bs{X}_n^{\text{c}})$ be the $n \times (n-1)$ matrix whose columns $\bs{X}_j^{\text{c}} = (X_{1j}^{\text{c}},\ldots,X_{nj}^{\text{c}})^\top$ have the entries $X_{ij}^{\text{c}} = X_{ij} - \widebar{X}_j$ with $\widebar{X}_j = n^{-1} \sum_{i=1}^n X_{ij}$. Straightforward arguments show that the minimization problem \eqref{eq:min-theta-1} is equivalent to
\begin{equation}\label{eq:min-theta-2}
\hat{\bs{\theta}}_{\lambda,-1} = (\hat{\theta}_{\lambda,2},\ldots,\hat{\theta}_{\lambda,n})^\top
 \in \underset{\bs{\vartheta} \in \reals^{n-1}}{\textnormal{argmin}} \Big\{ \frac{1}{n} \| \bs{Y}^{\text{c}} - X^{\text{c}} \bs{\vartheta} \|_2^2 + \lambda \|\vartheta\|_1 \Big\}
\end{equation}
and $\hat{\theta}_{\lambda,1} = \widebar{Y} - \sum_{j=2}^{n} \widebar{X}_j \hat{\theta}_{\lambda,j}$. Importantly, \eqref{eq:min-theta-2} is a standard lasso problem. Put differently, $\hat{\bs{\theta}}_{\lambda,-1}$ is the standard lasso in the model $\bs{Y}^{\text{c}} = X^{\text{c}} \bs{\theta}_{-1} + \bs{u}^{\text{c}}$, where $\bs{u}^c$ is defined analogously as $\bs{Y}^{\text{c}}$.

\item We now consider the effective noise of the standard lasso $\hat{\bs{\theta}}_{\lambda,-1}$, which is given by $2 \| (X^{\text{c}})^\top \bs{u}^{\text{c}} \|_\infty / n$. Simple calculations show that the effective noise is equal to the maximum statistic
\[ U_n := 2 \max_{2 \le j \le n} \bigg| -\frac{1}{n} \sum_{i < j} u_i + \frac{j-1}{n^2} \sum_{i=1}^n u_i \bigg|. \]
In order to approximate the quantiles of $U_n$, we run the following multiplier (or wild) bootstrap: We first compute the lasso residuals $\hat{u}_{\lambda_0,i} = Y_i - \hat{\alpha}_{\lambda_0,i}$ for $i=1,\ldots,n$, where $\lambda_0$ is a pilot tuning parameter whose choice is explained below. We then draw standard normal random vectors $\epsilon^{(\ell)} = (\epsilon^{(\ell)}_1,\ldots,\epsilon^{(\ell)}_n)$ for $\ell=1,\ldots,L$, where $L$ is a large number (say $L=1000$). The values $\{ \hat{u}_{\lambda_0,i} \cdot \epsilon_i^{(\ell)}: i=1,\ldots,n\}$ serve as the $\ell$-th bootstrap sample of the error terms $\{u_i: i=1,\ldots,n\}$. We next compute 
\begin{equation}\label{eq:impl-bootstrap}
U_n^{(\ell)} := 2 \max_{2 \le j \le n} \bigg| -\frac{1}{n} \sum_{i < j} \hat{u}_{\lambda_0,i} \epsilon^{(\ell)}_i + \frac{j-1}{n^2} \sum_{i=1}^n \hat{u}_{\lambda_0,i} \epsilon^{(\ell)}_i \bigg| 
\end{equation}  
for each $\ell$ and determine the empirical $q$-quantile $z_n(q)$ of the bootstrap sample $\{U_n^{(\ell)}: \ell =1,\ldots,L\}$. Regarding $z_n(q)$ as an approximation of the theoretical $q$-quantile of $U_n$, we finally set $\lambda = z_n(q)$.  

\item It remains to choose the starting value $\lambda_0$ in the above bootstrap procedure. The chosen value needs to be such that the lasso residuals $\hat{u}_{\lambda_0,i}$ give a reasonable approximation to the error terms $u_i$. This requires that the estimator $\hat{\bs{\alpha}}_{\lambda_0}$ is reasonably close to $\hazvec$. However, it is not necessary that the number of jumps in $\hat{\bs{\alpha}}_{\lambda_0}$ is close to the number of jumps in $\hazvec$. Indeed, even if $\hat{\bs{\alpha}}_{\lambda_0}$ has many spurious small jumps which do not correspond to any real jump in $\hazvec$, the discretized curve $\hat{\bs{\alpha}}_{\lambda_0}$ may still be close in shape to $\hazvec$. Hence, somewhat overestimating the number of change points in $\hazvec$ should not do any harm. In most applications, it is natural to assume that the discretized hazard rate $\hazvec$ has a rather small number of change points $K$. We may thus choose $\lambda_0$ such that $\hat{\bs{\alpha}}_{\lambda_0}$ has a moderately large number of jumps $K_{\max}$ (say $K_{\max}$ between $10$ and $20$) which most likely exceeds the true number of jumps $K$ a bit. This should equip us with an acceptable initial estimator $\hat{\bs{\alpha}}_{\lambda_0}$ of the discretized hazard $\hazvec$. 
\end{enumerate}

In summary, our algorithm for selecting $\lambda$ is as follows:
\begin{enumerate}[label=Step \arabic*., leftmargin=1.5cm]
\item Set $\lambda_0$ to be the largest value such that the resulting estimator $\hat{\bs{\alpha}}_{\lambda_0}$ has $K_{\max}$ change points and compute the lasso residuals $\hat{u}_{\lambda_0,i} = Y_i - \hat{\alpha}_{\lambda_0,i}$ for $i=1,\ldots,n$.
\item For $\ell=1,\ldots,L$, draw standard normal random vectors $\epsilon^{(\ell)} = (\epsilon^{(\ell)}_1,\ldots,\epsilon^{(\ell)}_n)$, compute $U_n^{(\ell)}$ as defined in \eqref{eq:impl-bootstrap} and determine the empirical $q$-quantile $z_n(q)$ of the computed values.
\item Set $\lambda = z_n(q)$ with $q$ close to $1$. 
\end{enumerate}

\section{Simulation study}\label{sec:sim}


The simulation study splits up into two parts. In the first, we investigate the performance of our fused lasso approach by a series of Monte Carlo experiments. In the second, we compare it to alternative methods.

\subsection{Finite sample performance of the fused lasso approach}

We simulate data from different designs that fit into Settings \ref{settingA} and \ref{settingB}:
\begin{itemize}[leftmargin=0.5cm]
\item We consider two different piecewise constant hazard functions whose form is loosely inspired by the estimated hazard functions in the application of Section \ref{sec:app}: 
\begin{align*}
\haz_{[1]}(t) & = 4 \cdot \ind(t < 0.25) + 1 \cdot \ind(t \ge 0.25) \\
\haz_{[2]}(t) & = 4 \cdot \ind(t < 0.2) + 1.5 \cdot \ind(0.2 \le t < 0.6) + 0.5 \cdot \ind(t \ge 0.6). 
\end{align*} 
\item We consider the model from Setting \ref{settingA} with no covariates at all as well as the Cox model from Setting \ref{settingB} with different covariate vectors, in particular, a low-dimensional version of the Cox model with only $p=2$ covariates and a higher-dimensional version with $p=100$ covariates. The covariate vector $\bs{W}_i = (W_{i1},\ldots$ $\ldots,W_{ip})^\top$ is defined as follows: the covariates $W_{ij}$ with odd indices $j=1,3,\ldots$  are binary variables which take the values $-1$ and $1$ with probability $0.5$ each, whereas the covariates $W_{ij}$ with even indices $j=2,4,\ldots$ are uniformly distributed on $[-1,1]$. For simplicity, we assume that the covariates are independent from each other. The parameter vector is chosen as $\bs{\beta} = (\beta_1,\beta_2)^\top = (0.25,1)^\top$ in the low-dimensional case and as $\bs{\beta} = (\beta_1,\beta_2,\beta_3,\ldots,\beta_{100})^\top = (0.25,1,0,\ldots,0)^\top$ in the higher-dimensional case. 
\item The right-censoring variables $C_i$ are assumed to be independent from $(T_i^*,\bs{W}_i)$ and are drawn from an exponential distribution with parameter $0.5$, resulting in around 20--25\% censored data points in all considered simulation scenarios.
\item For estimation, we consider the unit interval $[\tau_{\min},\tau_{\max}] = [0,1]$ and three different sample sizes $n \in \{500,1000,2000\}$.
\end{itemize}

\begin{table}[t]
\centering
\small{\begin{tabular}{rccc}
\hline\hline
Scenario & hazard  & sample size & covariates \\
\hline   
S$_0$[1] & $\haz_{[1]}$ & $n \in \{500,1000,2000\}$ & 0 \\
S$_2$[1] & $\haz_{[1]}$ & $n \in \{500,1000,2000\}$ & 2 \\
S$_{100}$[1] & $\haz_{[1]}$ & $n \in \{500,1000,2000\}$ & 100 \\
S$_0$[2] & $\haz_{[2]}$ & $n \in \{500,1000,2000\}$ & 0 \\
S$_2$[2] & $\haz_{[2]}$ & $n \in \{500,1000,2000\}$ & 2 \\
S$_{100}$[2] & $\haz_{[2]}$ & $n \in \{500,1000,2000\}$ & 100 \\
\hline\hline
\end{tabular}}
\caption{List of simulation scenarios.} \label{table:sim-scenarios}
\end{table}

The resulting simulation scenarios are listed in Table \ref{table:sim-scenarios}. Estimating the piecewise constant hazard in these scenarios is a non-trivial problem. This is reflected by fairly low signal-to-noise ratios in the regression model \eqref{eq:signal-plus-noise-model} which underlies our fused lasso method. The signal-to-noise ratio in model \eqref{eq:signal-plus-noise-model} is defined as the empirical variance of the signal vector $\hazvec = (\haz_1,\ldots,\haz_n)^\top$ divided by that of the noise vector $\bs{u} = (u_1,\ldots,u_n)^\top$. Formally, $\text{SNR} = \hat{V}(\hazvec)/\hat{V}(\bs{u})$, where $\hat{V}(\bs{w})$ is the empirical variance of a generic vector $\bs{w} = (w_1,\ldots,w_n)^\top$. In the considered simulation scenarios, the SNR (averaged over $1000$ simulation runs) ranges between $0.23$ and $0.33$. This shows that there is much more variation in the noise term than in the signal, implying that it is quite difficult to estimate the signal in this model. 

Throughout the simulation study, we implement our estimator as described in Section \ref{sec:impl} with $q=0.9$, $K_{\max} = 20$ and $L=100$. In order to estimate the parameter vector $\bs{\beta}$ in the simulation scenarios with a Cox model in the background, we proceed as follows: in the low-dimensional case with $p=2$, we employ a standard partial likelihood estimator as implemented in the \texttt{R}-package \rpack{survival}, whereas in the higher-dimensional case with $p=100$, we make use of an $\ell_1$-penalized (i.e.\ a lasso) version of the partial likelihood estimator as implemented in the \texttt{R}-package \rpack{glmnet} with default 10-fold cross-validated penalty parameter. The results of our simulation experiments are reported in Figures \ref{fig:settingS01}--\ref{fig:settingS1002} and Tables \ref{table:L2error}--\ref{table:asymdist}. As they are qualitatively very similar across the considered scenarios, we proceed as follows: we discuss the results for one of the scenarios -- in particular, for Scenario S$_0$[2] -- in detail while keeping the discussion of the other scenarios rather short.

\begin{figure}[t]
\centering
\includegraphics[width=\textwidth]{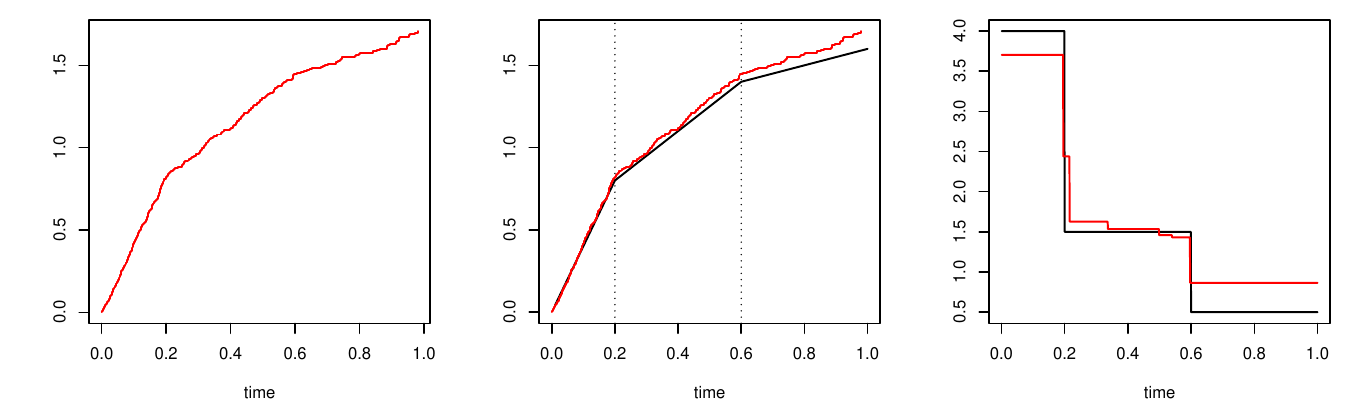}
\caption{Results of a representative simulation run in Scenario S$_0$[2] with $n=1000$. Left: Breslow (i.e., Nelson-Aalen) estimator. Middle: cumulative hazard in black and Breslow estimator in red. Right: hazard rate in black and its fused lasso estimator in red.}\label{fig:single-run} 
\end{figure}

To start with, we have a closer look at a representative simulation run in Scenario S$_0$[2] (with $n=1000$). Figure \ref{fig:single-run} summarizes the results of this run. The red curve in the left-hand panel is the Breslow (i.e., Nelson-Aalen) estimator $\hat{A}_{[2]}(t)$ of the cumulative hazard $A_{[2]}^*(t) = \int_0^t \haz_{[2]}(s) ds$. The middle panel shows the estimator $\hat{A}_{[2]}$ once again (in red) together with the piecewise linear cumulative hazard $A_{[2]}^*$ (in black), where the kinks of $A_{[2]}^*$ at $t=0.2$ and $t=0.6$ are indicated by vertical dotted lines. The right-hand panel depicts the hazard $\haz_{[2]}$ in black and its fused lasso estimator in red. The left-hand (and middle) panel of Figure \ref{fig:single-run} demonstrate that it is not so easy to spot the piecewise linear structure of the function $A_{[2]}^*$ in the estimator $\hat{A}_{[2]}$. This illustrates that estimating the kinks and slope parameters of $A_{[2]}^*$ -- or equivalently, the change points and function values of $\haz_{[2]}$ -- is not trivial at all. Nevertheless, the fused lasso estimator (depicted in the right-hand panel in red) does a pretty good job in reconstructing the piecewise constant hazard $\haz_{[2]}$ (depicted in black). Several typical features of the estimator become visible here:
\begin{enumerate}[label=(\roman*),leftmargin=0.85cm]

\item The fused lasso has a bias. This has two reasons: First, as any other member of the lasso family, it is a shrinkage estimator and thus is biased by construction. Second, our fused lasso method is based on the Breslow estimator. Hence, any imprecisions of this estimator are ``inherited'' by our method. In the case at hand, the Breslow estimator is sloping upwards more strongly than the cumulative hazard after $t=0.6$. This induces a substantial upward bias in our fused lasso estimator after $t=0.6$.
  
\item The fused lasso often reconstructs jumps in the underlying hazard by a sequence of consecutive jumps rather than one big jump. In the case at hand, the fused lasso does not mimic the true change point at $t=0.2$ by one big jump but rather by two jumps in quick succession.
  
\item The fused lasso may have additional jumps far away from the true change points. However, these jumps tend to be small and thus negligible. In the case at hand, the fused lasso has additional small jumps between the two true change points at $t=0.2$ and $t=0.6$.

\end{enumerate}
We will re-encounter these features of the fused lasso when discussing the results of our simulation experiments below.

Figure \ref{fig:settingS02} depicts the results of $100$ simulation runs in Scenario S$_0$[2] for different sample sizes. The upper three panels show the $100$ estimated hazard rates $\hat{\bs{\alpha}}_\lambda$ (in grey) together with the true underlying hazard (in black). The lower three panels give the locations and jump heights of the change points estimated over the $100$ simulation runs. Each grey circle represents one estimated change point. The $x$-axis specifies the location of the change point, whereas the $y$-axis gives the jump height. The larger red circles correspond to the location and jump height of the two true change points of the underlying hazard. As can be seen from the plots, the fused lasso estimates capture the piecewise constant step structure of the hazard pretty well, the approximation getting better as the sample size increases. The three features discussed in the context of Figure \ref{fig:single-run} also become visible here: 
\begin{enumerate}[label=(\roman*),leftmargin=0.85cm]

\item The fused lasso estimates have a bias which is most pronounced for $n=500$ and gets smaller as the sample size increases. As already discussed, this bias has two sources: the shrinkage done by the fused lasso estimator and bias inherited from the Breslow estimator.

\item The lower three panels show that many estimated jumps are located fairly close to the two true change points at $t=0.2$ and $t=0.6$. However, most of these jumps are substantially smaller in size than the true jumps (indicated by the two red dots). The reason is that the lasso often uses several smaller jumps in quick succession to reconstruct the true jumps at $t=0.2$ and $t=0.6$. 

\item Some of the fused lasso estimates have additional spurious jumps quite distant from the true change points $t=0.2$ and $t=0.6$. These jumps, however, tend to be comparatively small. (See also the plots for the scenarios with the hazard $\haz_{[1]}$ which give an even better illustration of this point.)
 
\end{enumerate}
Having summarized the main simulation results for Scenario S$_0$[2], we now briefly comment on the findings for the other two scenarios with the hazard $\haz_{[2]}$ (i.e., Scenarios S$_2$[2] and S$_{100}$[2]), which differ from S$_0$[2] only by incorporating covariates. The results for S$_2$[2] and S$_{100}$[2] in Figures \ref{fig:settingS22} and \ref{fig:settingS1002} are very similar to those for S$_0$[2] even though somewhat less precise. This is not surprising as the estimation of the parameter vector $\bs{\beta}$ produces additional approximation error, which gets reflected in somewhat less precise estimation results. It is worth noting that the results for the higher-dimensional Scenario S$_{100}$[2] are comparable to those for the low-dimensional Scenario S$_2$[2], despite the substantially larger number of covariates. This indicates that, in the setting at hand, the penalized partial likelihood estimator is able to recover the parameter vector $\bs{\beta}$ with sufficient accuracy, so that the increased dimensionality has only a minor impact on the estimation of the hazard function.
The results for the scenarios with the hazard $\haz_{[1]}$ (i.e., Scenarios S$_0$[1], S$_2$[1] and S$_{100}$[1] depicted in Figures \ref{fig:settingS01}, \ref{fig:settingS21} and \ref{fig:settingS1001}) are qualitatively very similar to those for the scenarios with the hazard $\haz_{[2]}$ (i.e., Scenarios S$_0$[2], S$_2$[2] and S$_{100}$[2]). In particular, the main points of our discussion apply to the scenarios with the hazard $\haz_{[1]}$ as well. 

\begin{figure}[p]
\centering
\includegraphics[width=\textwidth]{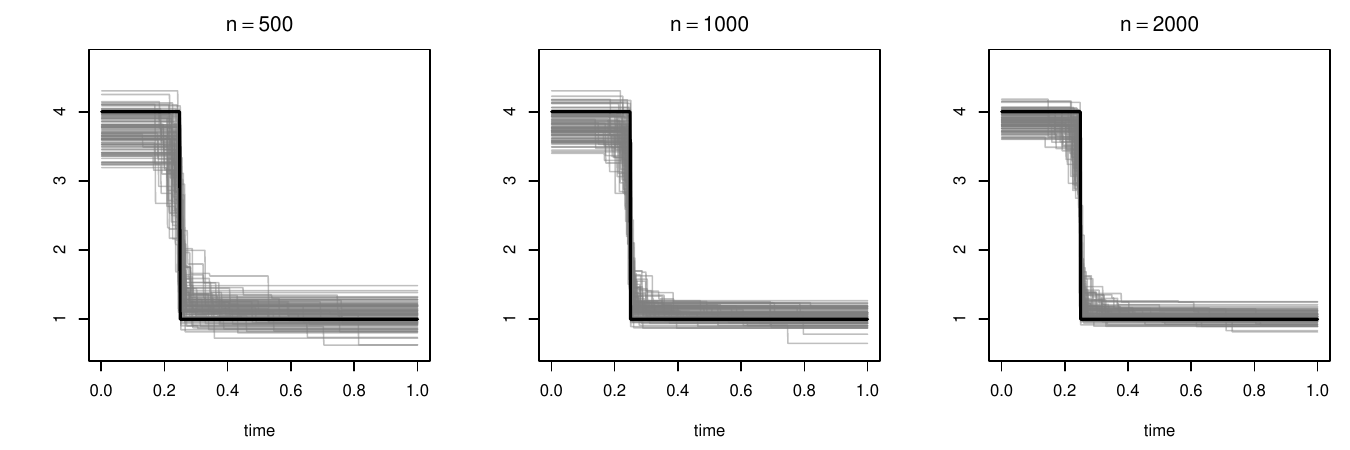}
\includegraphics[width=\textwidth]{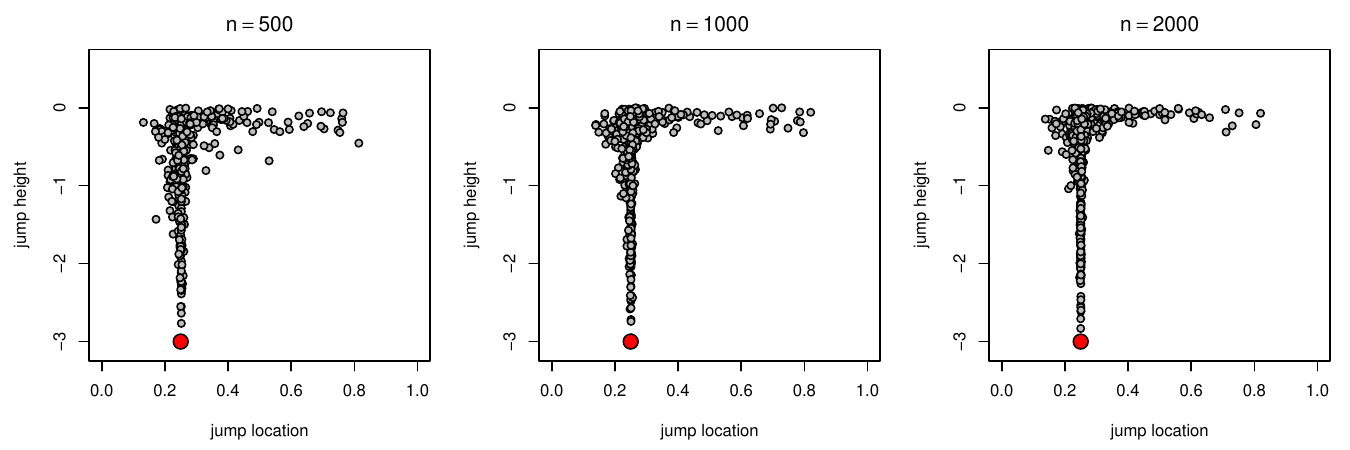}
\caption{Simulation results for Scenario S$_0$[1].}\label{fig:settingS01}
\vspace{1cm}

\includegraphics[width=\textwidth]{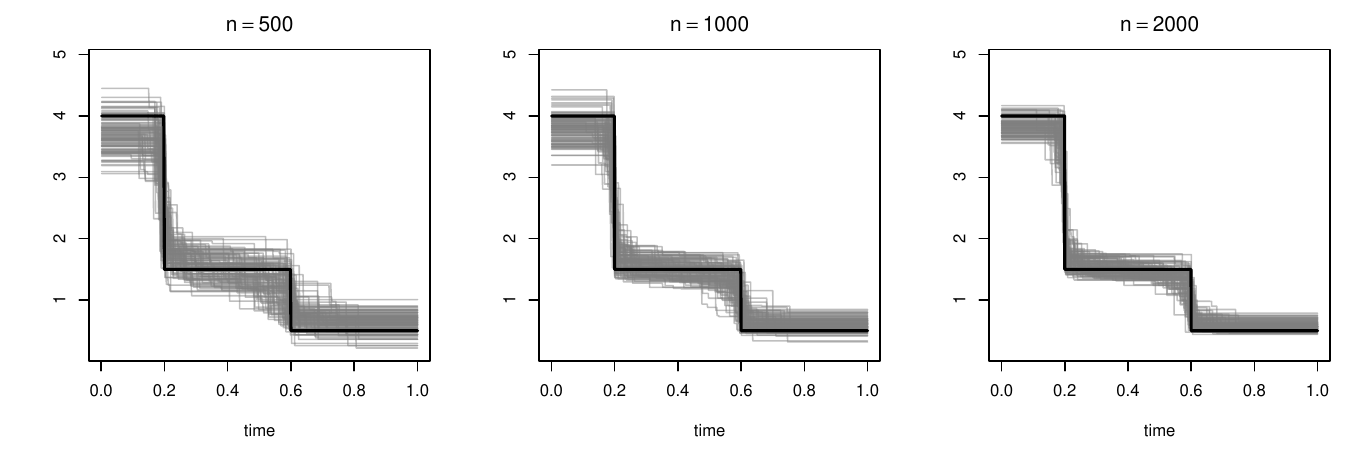}
\includegraphics[width=\textwidth]{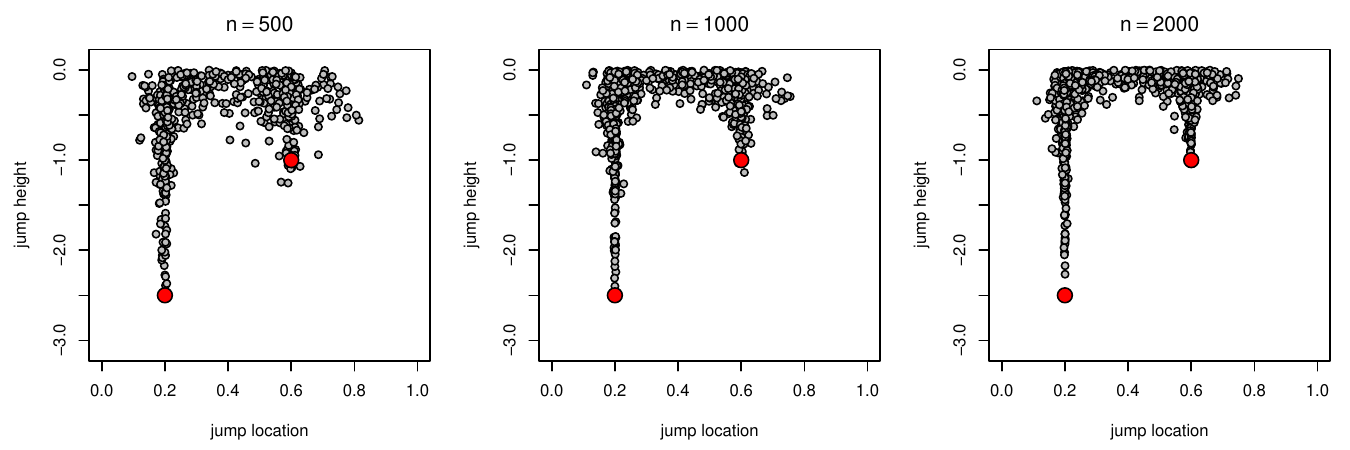}
\caption{Simulation results for Scenario S$_0$[2].}\label{fig:settingS02}
\end{figure}

\begin{figure}[p]
\includegraphics[width=\textwidth]{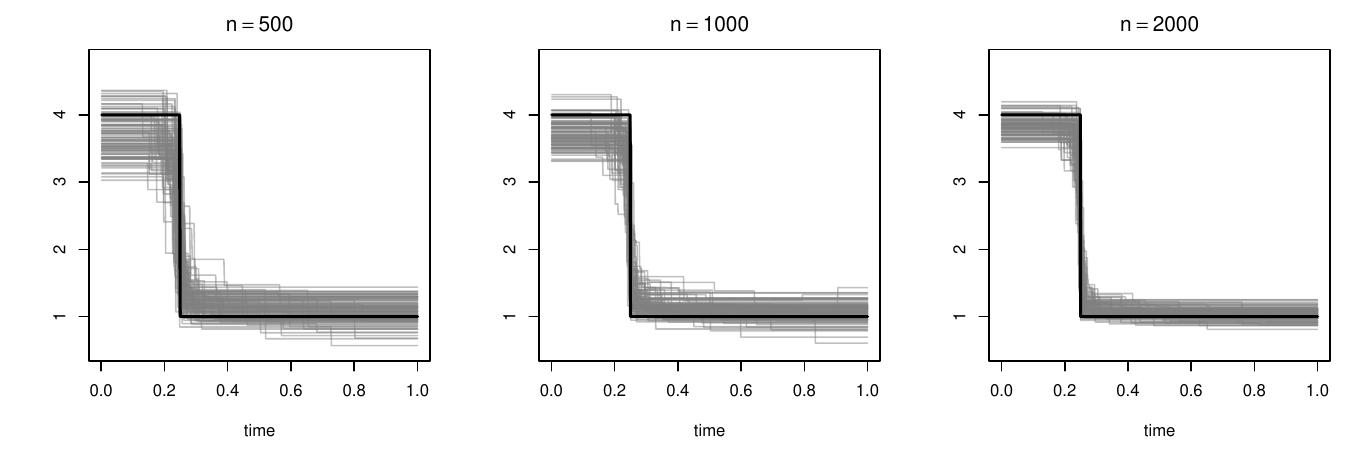}
\includegraphics[width=\textwidth]{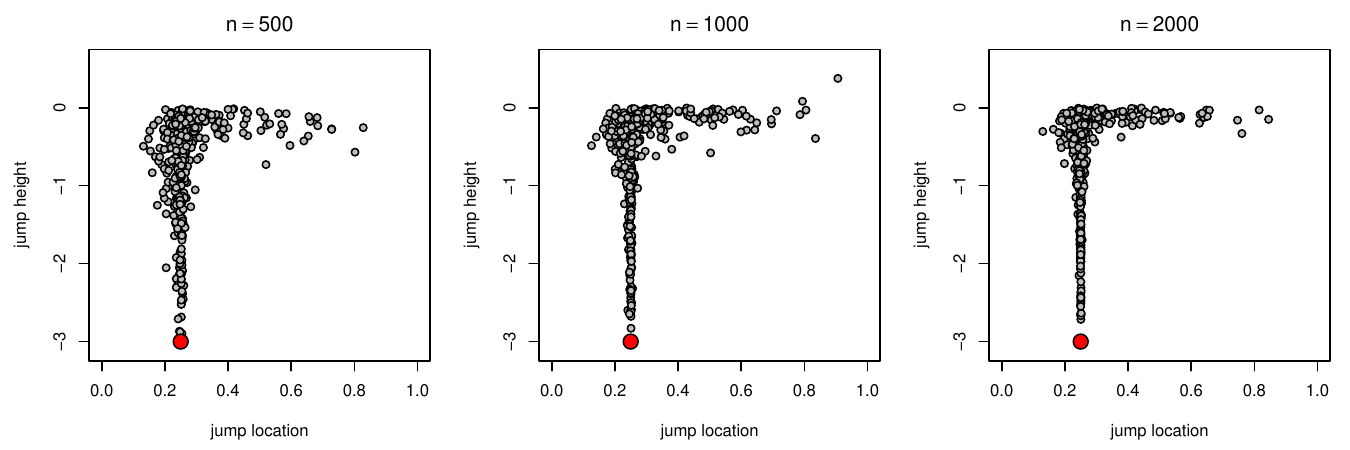}
\caption{Simulation results for Scenario S$_2$[1].}\label{fig:settingS21}
\vspace{1cm}

\includegraphics[width=\textwidth]{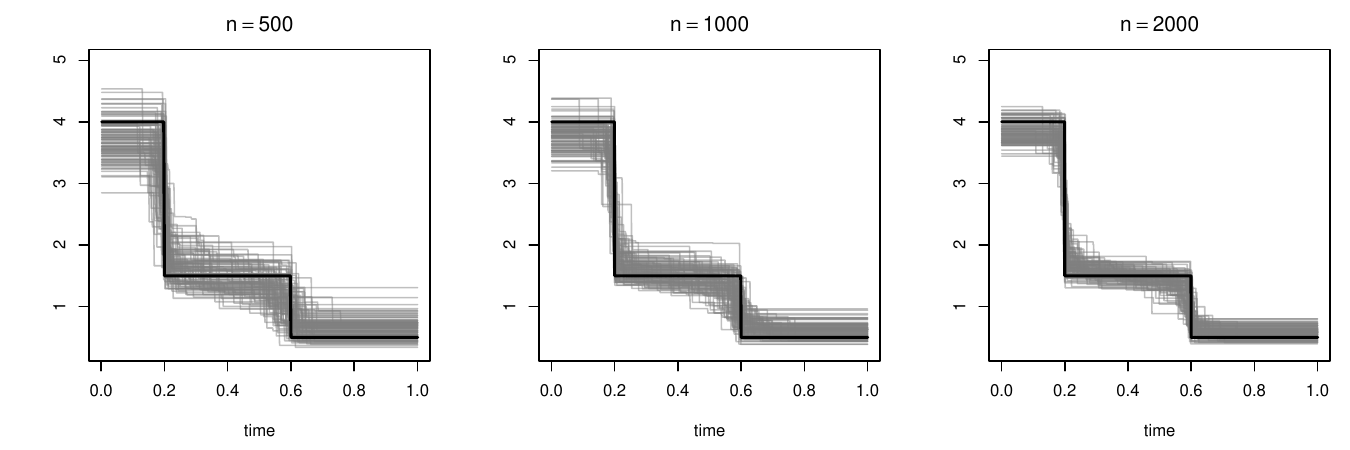}
\includegraphics[width=\textwidth]{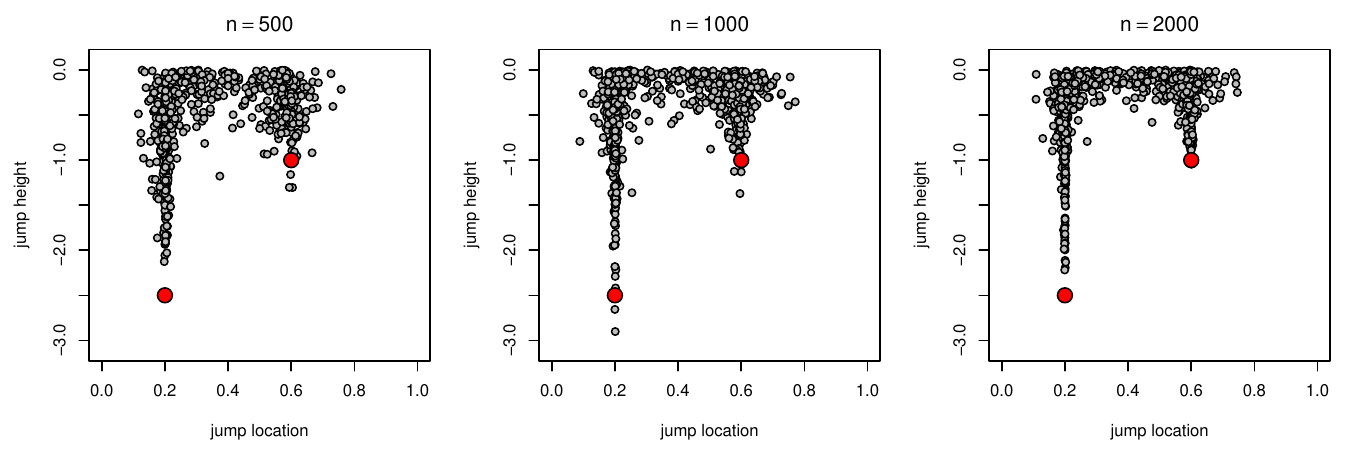}
\caption{Simulation results for Scenario S$_2$[2].}\label{fig:settingS22}
\end{figure}

\begin{figure}[p]
\centering
\includegraphics[width=\textwidth]{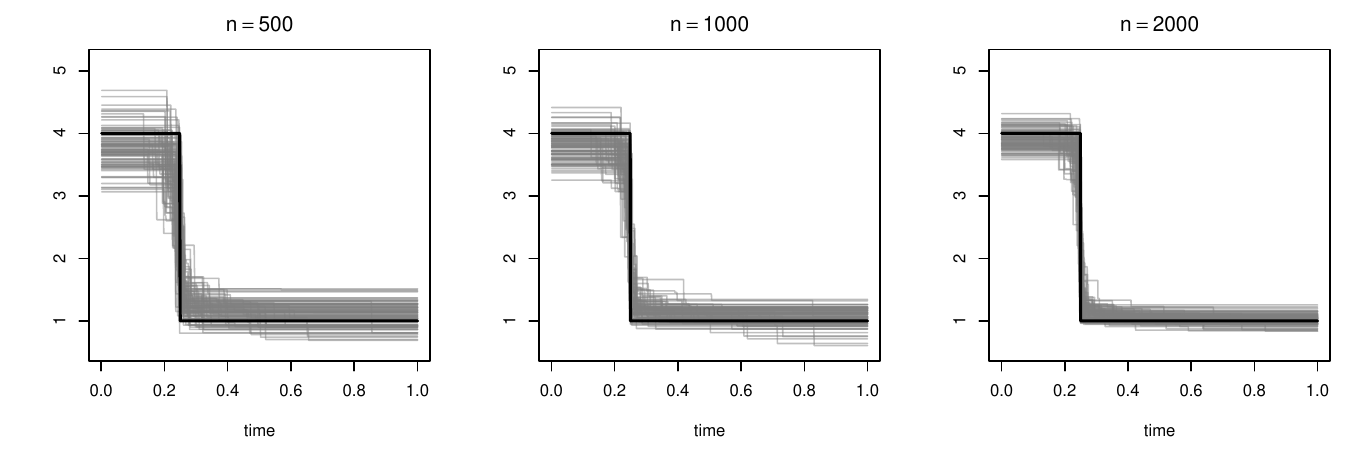}
\includegraphics[width=\textwidth]{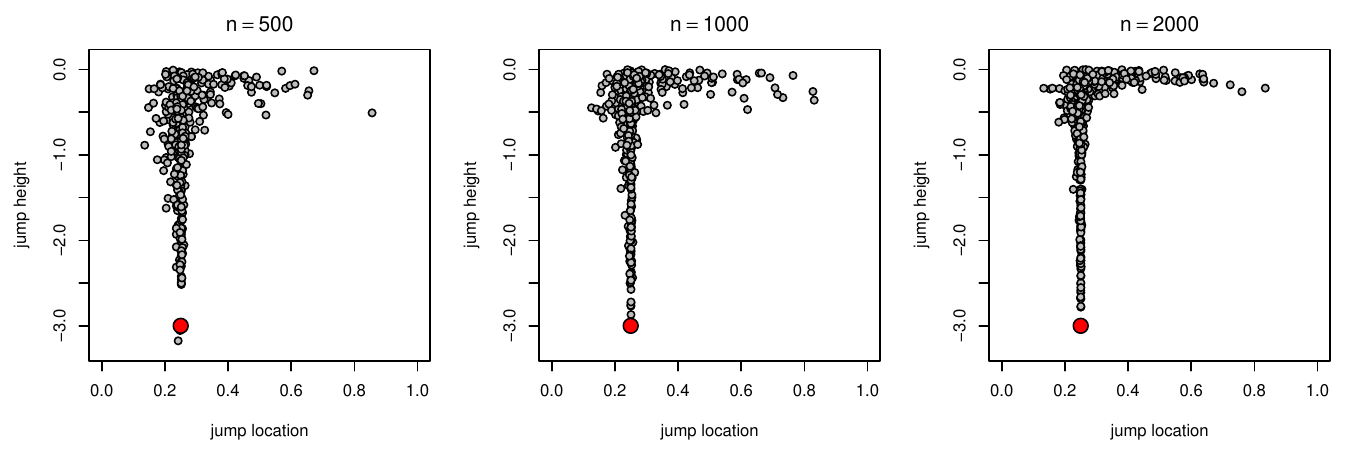}
\caption{Simulation results for Scenario S$_{100}$[1]. 
}\label{fig:settingS1001}
\vspace{1cm}

\includegraphics[width=\textwidth]{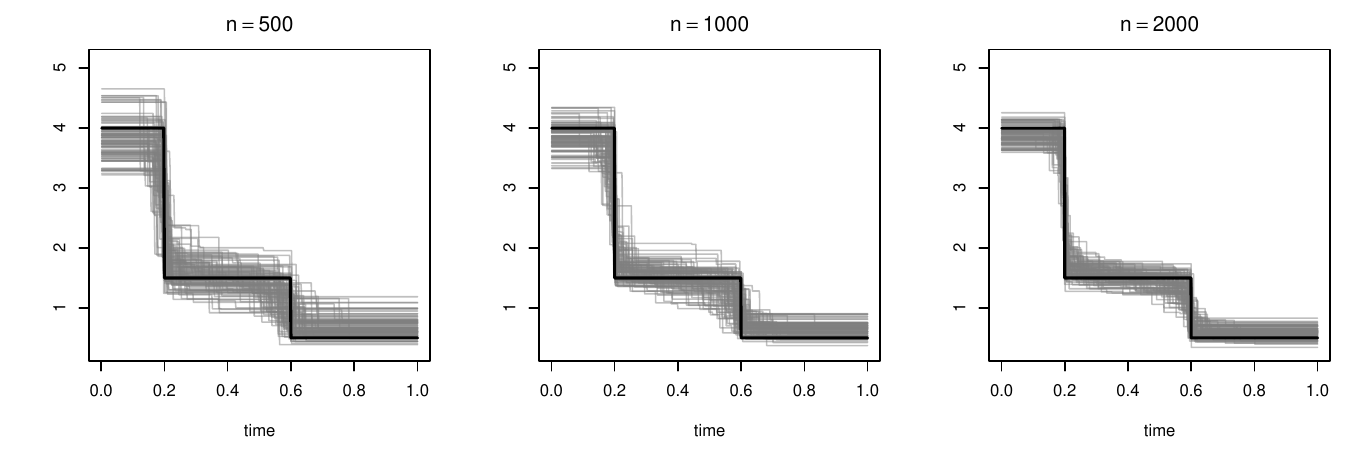}
\includegraphics[width=\textwidth]{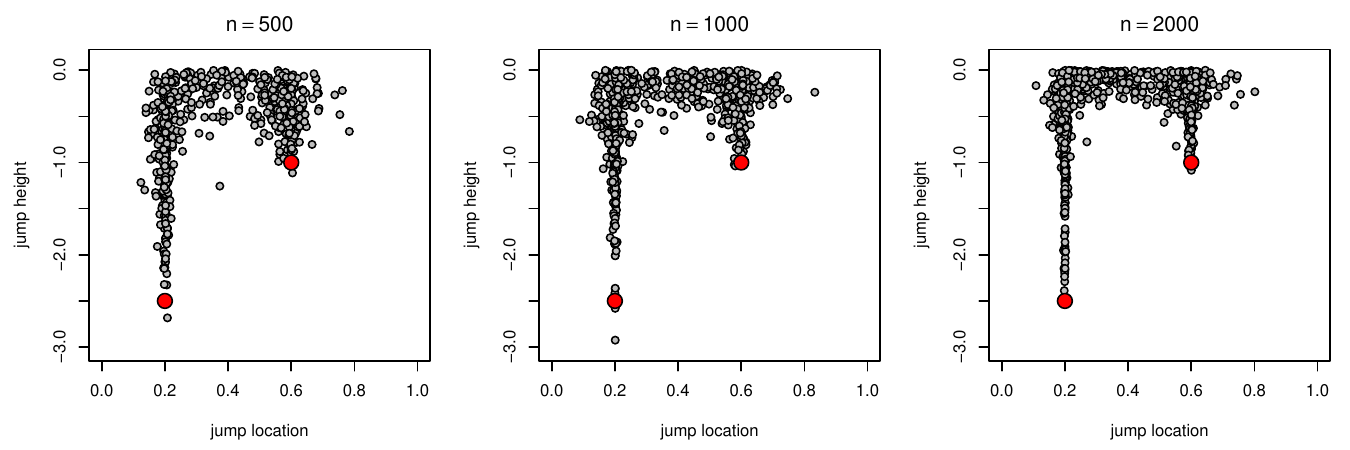}
\caption{Simulation results for Scenario S$_{100}$[2]. 
}\label{fig:settingS1002}
\end{figure}

Tables \ref{table:L2error} and \ref{table:asymdist} complement the visual results from Figures \ref{fig:settingS01}--\ref{fig:settingS1002} by summary statistics. All reported numbers are based on $1000$ simulation runs.
Table \ref{table:L2error} gives information on how accurately the estimator $\hat{\bs{\alpha}}_\lambda$ reconstructs the true underlying hazard curve: it reports the average of the relative $\ell_2$-error $\|\hat{\bs{\alpha}}_\lambda - \hazvec\|_2^2/\norm{\hazvec}_2^2$ over the $1000$ simulation runs along with the standard deviation in brackets. As expected, the average error (as well as the standard deviation) gets smaller as the sample size increases, implying that the quality of the reconstruction gets better with growing sample size. 
Table \ref{table:asymdist} is concerned with the accuracy of change-point estimation. 
Assessing change-point detection accuracy is somewhat challenging in our general framework. To illustrate this, consider a simplified version of our piecewise constant hazard model \eqref{eq:pc-hazard} where the number of change points $K$ is known. In this setting, we can tune our estimator (by appropriate choice of $\lambda$) to produce exactly $K$ estimated change points. Let $\tau_{\mathrm{vec}}=(\tau_1,\ldots,\tau_K)^\top$ with $\tau_1<\cdots<\tau_K$ denote the vector of true change points and $\hat{\tau}_{\mathrm{vec}}=(\hat{\tau}_1,\ldots,\hat{\tau}_K)^\top$ with $\hat{\tau}_1<\cdots<\hat{\tau}_K$ the vector of estimated change points. The change-point detection accuracy can be assessed very easily in this case, e.g., by computing the $\ell_1$- or $\ell_2$-distance between $\tau_{\mathrm{vec}}$ and $\hat{\tau}_{\mathrm{vec}}$. In our setting, such a simple approach is no longer applicable because the numbers of true and estimated change points generally differ. Consequently, the vectors $\tau_{\mathrm{vec}}$ and $\hat{\tau}_{\mathrm{vec}}$ need not have the same length, so that no natural one-to-one correspondence between true and estimated change points exists. We therefore require a more general measure of change-point detection accuracy. We in particular make use of the asymmetric distance measure $d(\hat{\mathcal{S}}_\lambda,\mathcal{S}^*)$, which was studied theoretically in Section~\ref{sec:theory} and has also been employed in previous work \citep[see e.g.][]{LinSharpnackRinaldoTibshirani2017}.
Table \ref{table:asymdist} reports the average value of $d(\hat{\mathcal{S}}_\lambda,\mathcal{S}^*)$ over the $1000$ simulation runs, together with its standard deviation (shown in parentheses). The interpretation of the measure $d(\hat{\mathcal{S}}_\lambda,\mathcal{S}^*)$ is straightforward: if $d(\hat{\mathcal{S}}_\lambda,\mathcal{S}^*)$ equals $\delta$, then every true change point has at least one estimated change point within distance $\delta$. The average values and standard deviations reported in Table \ref{table:asymdist} are very small already for $n=500$ and get even smaller with increasing sample size. Hence, our estimator tends to have change points very close to the true ones. This, however, does not exclude the existence of additional ``spurious'' change points far away from the true ones. According to Theorem \ref{theo:change-point-II}, such ``spurious'' change points are negligible in the sense of having small jump size. As turns out, it is quite difficult to underpin this theoretical finding numerically as this would require to define precisely what constitutes a ``spurious'' change point. For instance, one could call an estimated change point ``spurious'' if it is further away from any true change point than a pre-specified threshold value, say $0.1$. As such a definition is quite adhoc, we have not attempted anything of this kind. We rather think that the visual results in Figures \ref{fig:settingS01}--\ref{fig:settingS1002} speak for themselves: they illustrate quite clearly that ``spurious'' jumps far away from the true change points tend to be comparatively small.

\begin{table}[t]
\centering


\footnotesize{
\makebox[\textwidth][c]{%
\begin{tabular}{lcccccc}
\hline\hline  
Scenario & S$_0$[1] & S$_0$[2] & S$_2$[1] & S$_2$[2] & S$_{100}$[1] & S$_{100}$[2] \\
\hline
$n=500$  & 0.021 (0.013) & 0.030 (0.019) & 0.025 (0.012) & 0.032 (0.017) & 0.025 (0.014) & 0.029 (0.013) \\ 
$n=1000$ & 0.012 (0.007) & 0.017 (0.010) & 0.015 (0.007) & 0.019 (0.009) & 0.014 (0.008) & 0.018 (0.010) \\
$n=2000$ & 0.007 (0.004) & 0.009 (0.005) & 0.008 (0.004) & 0.010 (0.005) & 0.006 (0.003) & 0.010 (0.004) \\
\hline\hline
\end{tabular}}}
\caption{Average values of the relative squared $\ell_2$-error $\norm{\hat{\bs{\alpha}}_\lambda - \hazvec}_2^2/\norm{\hazvec}_2^2$, with standard deviations reported in brackets.}\label{table:L2error} 
\vspace{0.35cm}

\footnotesize{
\makebox[\textwidth][c]{%
\begin{tabular}{lcccccc}
\hline\hline    
Scenario & S$_0$[1] & S$_0$[2] & S$_2$[1] & S$_2$[2] & S$_{100}$[1] & S$_{100}$[2] \\
\hline
$n=500$  & 0.002 (0.003) & 0.016 (0.022) & 0.003 (0.003) & 0.019 (0.026) & 0.04 (0.003) & 0.021 (0.032) \\
$n=1000$ & 0.001 (0.001) & 0.008 (0.009) & 0.001 (0.002) & 0.009 (0.011) & 0.002 (0.002) & 0.009 (0.011) \\
$n=2000$ & 0.001 (0.001) & 0.004 (0.005) & 0.001 (0.001) & 0.004 (0.005) & 0.001 (0.001) & 0.004 (0.006) \\
\hline\hline
\end{tabular}}}
\caption{Average values of the asymmetric distance measure $d(\hat{\mathcal{S}}_\lambda,\mathcal{S}^*)$, with standard deviations reported in brackets.}\label{table:asymdist}   

\end{table}

\subsection{Comparison study}

We now compare our fused lasso approach with two alternative methods. The first is a sequential testing procedure for detecting multiple change points in piecewise constant hazards, which we refer to as the SeqTest method. Specifically, we consider the method proposed by \cite{GoodmanLiTiwari2011} and implemented in the \texttt{R}-package \rpack{eventTrack} by \cite{rufibach2023eventtrack}. The second is a spline-based method for estimating general hazard functions, in particular, Piecewise Exponential Additive Modelling (PAM) as proposed by \cite{BenderGrollScheipl2018} and implemented in the \texttt{R}-package \rpack{pammtools} \citep{bender2018pammtools}. The comparison with the former method allows us to assess the change-point detection performance of our approach, whereas the comparison with the latter method provides insight into its curve fitting performance. As the method of \cite{GoodmanLiTiwari2011} has only been developed for a simple survival model without covariates, we restrict attention to Scenarios S$_0$[1] and S$_0$[2] in what follows.

\begin{figure}[p]
\centering
\includegraphics[width=\textwidth]{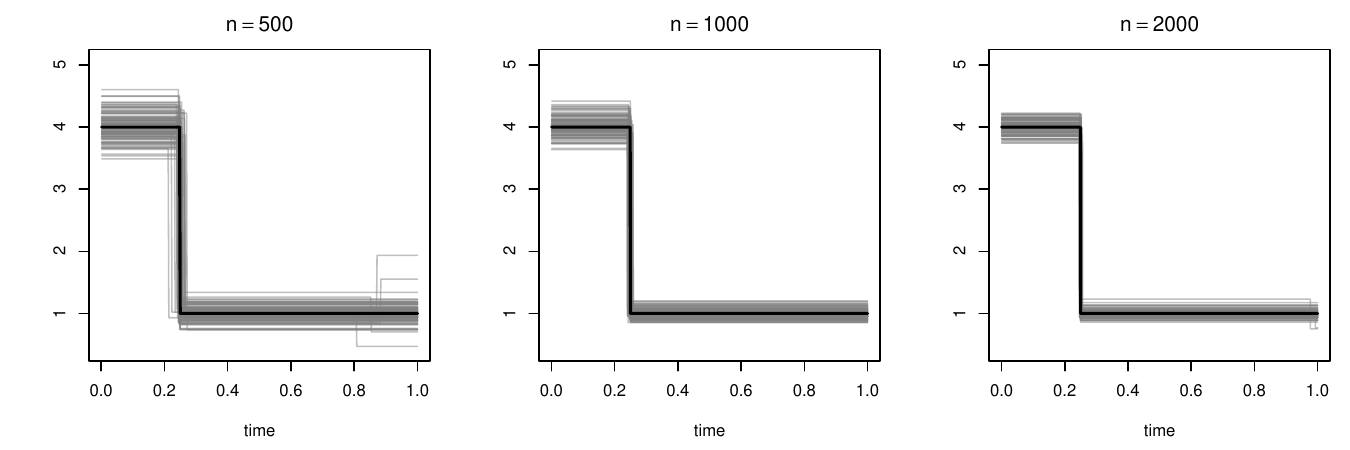}
\includegraphics[width=\textwidth]{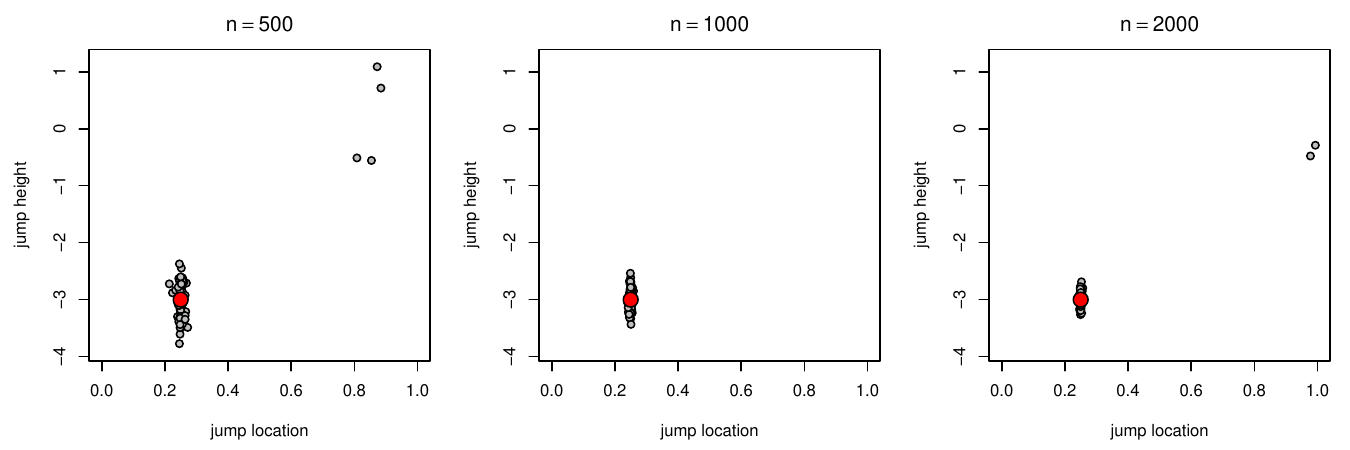}
\caption{Simulation results for the SeqTest method in Scenario S$_0$[1]. 
}\label{fig:settingS01_seqtest}
\vspace{0.5cm}

\includegraphics[width=\textwidth]{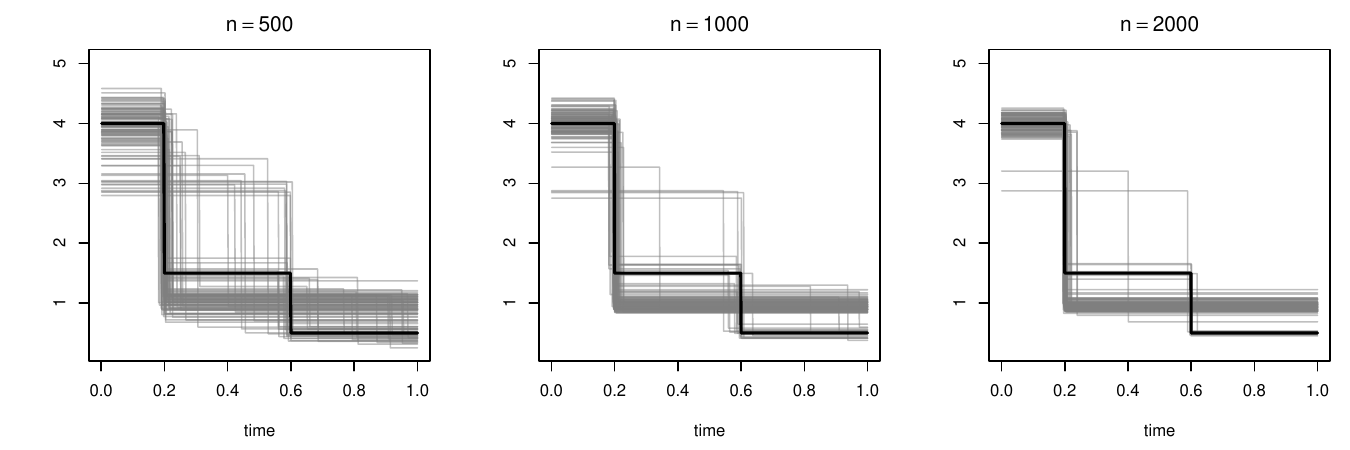}
\includegraphics[width=\textwidth]{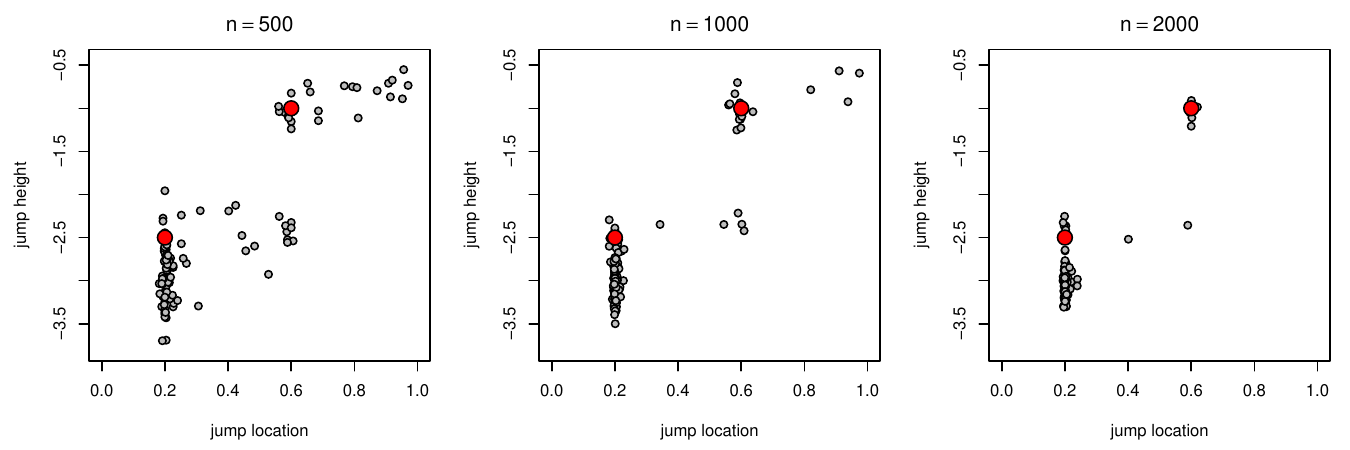}
\caption{Simulation results for the SeqTest method in Scenario S$_0$[2]. 
}\label{fig:settingS02_seqtest}
\end{figure}

\begin{figure}[t]
\centering
\includegraphics[width=\textwidth]{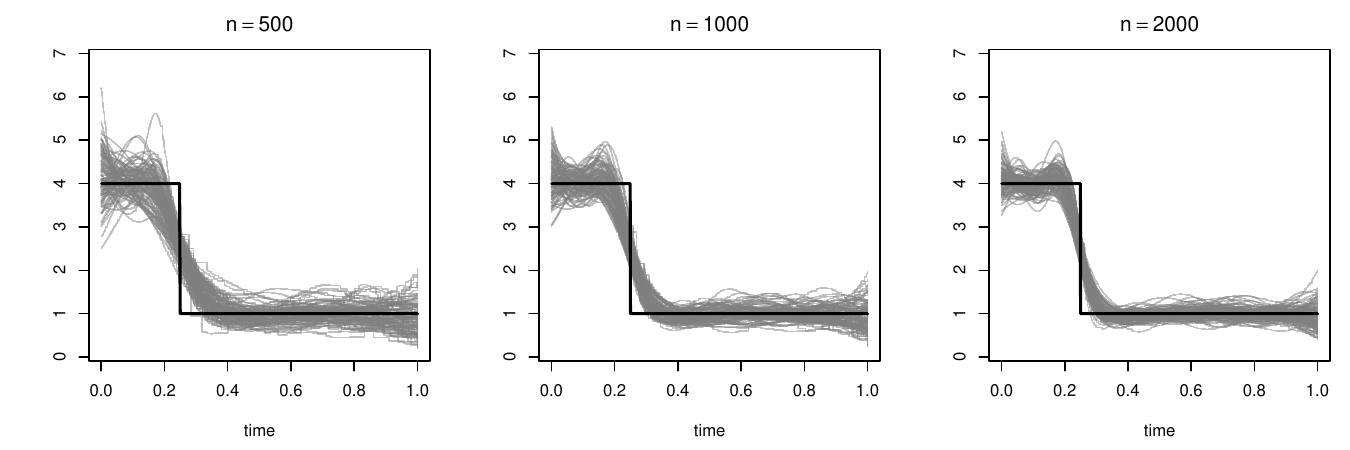}
\caption{Simulation results for PAM in Scenario S$_0$[1]. 
}\label{fig:settingS01_pam}
\vspace{0.5cm}

\includegraphics[width=\textwidth]{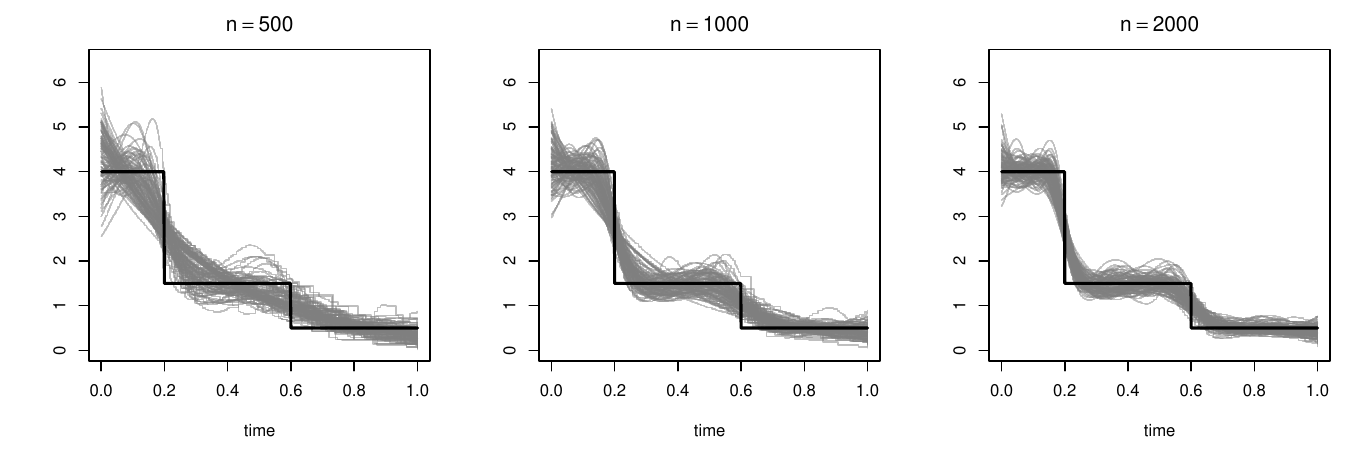}
\caption{Simulation results for PAM in Scenario S$_0$[2]. 
}\label{fig:settingS02_pam}
\end{figure}

The SeqTest method proceeds as follows: Consider the piecewise constant hazard model from Scenarios S$_0$[1] and S$_0$[2] and assume that a known upper bound $K_{\text{upper}}$ on the (unknown) number of change points $K$ is available. Throughout the comparison study, we set $K_{\text{upper}} = 5$. First, the model with $K_{\text{upper}}$ change points is fitted by maximum likelihood, yielding $K_{\text{upper}}$ estimated change-point locations and $K_{\text{upper}}+1$ estimated hazard levels $\hat{a}_1,\ldots,\hat{a}_{K_{\text{upper}}+1}$. Subsequently, $K_{\text{upper}}$ Wald-type tests are performed to assess whether the adjacent hazard levels $\hat{a}_k$ and $\hat{a}_{k+1}$ differ significantly for $k=1,\ldots,K_{\text{upper}}$. To account for multiple testing, an exponential alpha-spending rule is used, according to which the $k$-th test is conducted at significance level $\text{FWER}\cdot 2^{-k}$, where $\text{FWER} \in (0,1)$ denotes the desired family-wise error rate and is set to $0.05$ throughout. This alpha-spending rule reflects the fact that right-censored time-to-event data typically become sparser over time, making accurate estimation increasingly difficult at later time points. In the \texttt{R}-package \rpack{eventTrack}, the Wald-type tests are performed sequentially for $k=1,2,\ldots$ and the procedure terminates at the first test that fails to reject the null hypothesis. If this occurs at the $k_0$-th test, only the first $k_0-1$ estimated change points are included in the final model. In contrast, we include all change points whose corresponding Wald-type tests reject the null hypothesis. The maximum-likelihood estimators and Wald-type tests of this approach are described in detail in \cite{GoodmanLiTiwari2011}. 

The PAM method proceeds as follows: To start with, the time interval $[\tau_{\min},\tau_{\max}] = [0,1]$ is partitioned into a fine grid $[0,u_1),[u_1,u_2),\ldots,[u_{N-1},u_N),[u_N,1]$. 
We in particular follow the recommendation of \cite{bender2018pammtools} and use the unique observed event times as grid points, which is the default option in the package \rpack{pammtools}.
Approximating the hazard $\haz$ by a function which is constant on each interval $[u_{j-1},u_j)$, we obtain that $\haz(t) \approx \sum_{j=1}^{N+1} a_j \ind(t\in [u_{j-1},u_j))$ for each $t$, i.e., it can be parameterized by the coefficient vector $(a_1,\ldots,a_{N+1})$. For large $N$, this representation is very flexible but produces a high-dimensional estimation problem with a large number of coefficients (equal to $N+1$). To address this issue, PAM imposes a smoothness constraint on the coefficients. In particular, the hazard $\haz$ is approximated by a spline function $s$ such that $a_j = s(u_j)$ for $j=1,\ldots,N+1$. Hence, the hazard is no longer parameterized by an unrestricted vector of coefficients but by a comparatively small set of spline coefficients, which substantially reduces the dimensionality of the estimation problem. 
We implement PAM with cubic B-splines. 
PAM estimates the spline coefficients by penalized maximum likelihood with a quadratic (i.e., ridge-type) penalty. We use the \texttt{R}-packages \rpack{pammtools} and \rpack{mgcv} to carry out the estimation.


The simulation results for the SeqTest method are summarized in Figures \ref{fig:settingS01_seqtest} and \ref{fig:settingS02_seqtest}. These figures should be interpreted in the same way as Figures \ref{fig:settingS01} and \ref{fig:settingS02}, which present the corresponding results for our fused lasso approach.
The SeqTest method performs very well in the simple Scenario S$_0$[1], which contains only a single change point, clearly outperforming our fused lasso method in this scenario. In almost all simulation runs, it correctly identifies exactly one change point and estimates its location with high accuracy. However, its performance deteriorates substantially in the somewhat more challenging Scenario S$_0$[2], which contains two change points. In many runs, the method detects only a single change point, leading to a poor reconstruction of the underlying hazard function. In contrast, our method continues to recover the hazard function accurately in this setting. These results suggest that the fused lasso approach is more robust and better able to accommodate increasingly complex hazard structures.

The simulation results for PAM are presented in Figures \ref{fig:settingS01_pam} and \ref{fig:settingS02_pam}. The hazard estimator $\hat{\alpha}(t) = \sum_{j=1}^{N+1} \hat{a}_j \ind(t\in [u_{j-1},u_j))$ produced by PAM has the property that neighbouring hazard levels $\hat{a}_j$ and $\hat{a}_{j+1}$ are typically distinct for almost all $j$. As a consequence, the estimated hazard function exhibits a change point at virtually every grid point $u_1,\ldots,u_N$. For this reason, plots of the estimated change-point locations and jump heights -- such as those shown for the fused lasso and SeqTest approaches in the lower panels of Figures \ref{fig:settingS01}--\ref{fig:settingS02_seqtest} -- are not particularly informative. We therefore omit such plots. 
As expected, PAM smooths the underlying piecewise constant structure of the hazard function. Nevertheless, as the sample size increases, the resulting estimates provide a reasonable approximation of the true hazard and capture the main changes in its shape. Despite this, the reconstruction errors reported in Table \ref{table:L2error-comparison} are consistently larger for PAM than for the fused lasso method across all considered cases. This indicates that the fused lasso approach recovers the underlying hazard function more accurately. 
We view PAM as a natural baseline method for hazard estimation. The simulation results suggest that when the true hazard exhibits abrupt changes and is therefore not well described by a smooth function, the fused lasso approach has a clear advantage over classical spline-based methods such as PAM.

\begin{table}[t]
\centering
\footnotesize{\begin{tabular}{lcccccc}
\hline\hline    
 & \multicolumn{3}{c}{Scenario S$_0$[1]} \\
 & $n=500$  & $n=1000$ & $n=2000$ \\
\hline
fused lasso    			& 0.021 (0.013) & 0.012 (0.007) & 0.007 (0.004) \\
Seqtest             	& 0.013 (0.013) & 0.006 (0.005) & 0.003 (0.003) \\
PAM            			& 0.043 (0.018) & 0.029 (0.009) & 0.020 (0.005) \\
\hline\hline\\[-0.3cm]
 & \multicolumn{3}{c}{Scenario S$_0$[2]} \\
 & $n=500$  & $n=1000$ & $n=2000$ \\
\hline
fused lasso    			& 0.030 (0.019) & 0.017 (0.010) & 0.009 (0.005) \\
SeqTest             	& 0.087 (0.070) & 0.060 (0.041) & 0.055 (0.029) \\
PAM            			& 0.052 (0.022) & 0.032 (0.014) & 0.020 (0.006) \\
\hline\hline
\end{tabular}}
\caption{Average values of the relative squared $\ell_2$-error $\|\hat{\bs{\alpha}}^{[M]} - \hazvec\|_2^2/\|\hazvec\|_2^2$ over $1000$ simulation runs, with standard deviations reported in brackets. Here, $\hat{\bs{\alpha}}^{[M]} = (\hat{\alpha}_1^{[M]}, \ldots,\hat{\alpha}_n^{[M]})^\top$ denotes the vector of hazard levels at time points $j/n$ ($j=1,\ldots,n$) estimated by method $M \in \{$fused lasso, SeqTest, PAM$\}$.}\label{table:L2error-comparison} 
\end{table}

\section{Empirical application}\label{sec:app}

\subsection{Background}

The motivation behind the present empirical illustration is clinical trial planning in oncology. The gold standard to demonstrate clinical benefit of a novel therapy is evidence of delaying time-to-death, called `overall survival' (OS) in the field. A second endpoint that is both observed earlier and used to demonstrate treatment benefit is progression-free survival (PFS), i.e., the time until (diagnosed) progression of disease or death, whatever occurs first. In fact, it is not uncommon for contemporary trials to evaluate both endpoints, with PFS benefit leading to accelerated drug approval and authorities demanding subsequent follow-up of OS.

Trials are planned in terms of the OS and/or PFS number of events to be observed, typically assuming proportional hazards between treatment groups or even exponential distributions \citep{ohneberg2013sample,nochmehrsample}. If OS and PFS are co-primary endpoints, a common approach is to plan the number of events to be observed independently for OS and PFS, assuming proportional hazards. In what follows, we use the methodology developed earlier to come up with a more flexible and realistic trial design. We in particular propose a flexible multi-state model with piecewise constant transition hazards to jointly model PFS and OS. We then fit the proposed model to clinical data from a randomised controlled trial in patients with non-small-cell lung cancer \citep{rittmeyer2017atezolizumab}. Trial planning (e.g.\ calculating the number of events to be observed) could then be performed by simulating from the fitted model. In what follows, we focus on model fitting, in particular, on estimating the piecewise constant transition hazards, but we do not carry out any simulation exercises for trial planning. We refer to \cite{beyer2020multistate} and \cite{danzer2022confirmatory} as trial design examples using simulation in a multi-state context. Notably, we are not interested in a specific interpretation of the change points in the estimated transition hazards \citep[as opposed to other studies in the change point literature such as][]{Brazzale2019}. The main motivation for fitting a piecewise constant hazard model is that this yields a flexible yet simple parametrization which is well suited for trial planning.

\subsection{Multi-state modelling of PFS and OS}

Following \cite{meller2019}, we model PFS and OS jointly via an illness-death model without recovery. Formally speaking, for each patient $i$ in our sample, let $X_i$ be a time-continuous Markov process with state space $\mathcal{R} = \{0, 1, 2\}$ and possible transitions $0 \to 1$, $0 \to 2$ and $1 \to 2$. We interpret state $0$ as the initial state of being alive without progression, state $2$ as the terminal state of death and state $1$ as an intermediate state entered upon progression diagnosis. In this multi-state model, PFS is viewed as the waiting time in the initial state $0$ and OS as the waiting time until reaching the absorbing state $2$, i.e., $\textnormal{PFS}_i = \inf\{ t \ge 0: X_i(t) \ne 0 \}$ and $\textnormal{OS}_i = \inf\{ t \ge 0: X_i(t) = 2 \}$ for each patient $i$. In general, PFS and OS do not both follow proportional hazards at the same time in this model. Hence, our multi-state framework naturally accounts for the common concern in the field \citep{mukhopadhyay2022log} that proportional hazards for both OS and PFS appear to be too simple for adequate modelling.

We do not observe the Markov processes $X_i$ completely but right-censored versions of them, where $C_i$ denotes the right-censoring time for patient $i$. In order to estimate the transition hazards $\haz_{0 \to 1}$, $\haz_{0 \to 2}$ and $\haz_{1 \to 2}$ by our methods, we can treat each transition ($0 \to 1$, $0 \to 2$ and $1 \to 2$) separately and reformulate the data in terms of the survival framework from Setting \ref{settingA} as follows: 
\begin{itemize}[leftmargin=0.5cm]
\item Transition $0 \to 1$:  
$T_i^* = \textnormal{PFS}_i$, $T_i = T_i^* \land C_i$ and $\delta_i = \ind(T_i^* \le C_i \textnormal{ and } \textnormal{PFS}_i < \textnormal{OS}_i)$.  
\item Transition $0 \to 2$:  
$T_i^* = \textnormal{PFS}_i$, $T_i = T_i^* \land C_i$ and $\delta_i = \ind(T_i^* \le C_i \textnormal{ and } \textnormal{PFS}_i = \textnormal{OS}_i)$.
\item Transition $1 \to 2$: $T_i^* = \textnormal{OS}_i$, $T_i = T_i^* \land C_i$ and $\delta_i = \ind(T_i^* \le C_i)$. As patient $i$ is at risk only after entering state $1$, i.e., from time $\textnormal{PFS}_i$ onwards, we also have to deal with left-truncation. To do so, we only consider the subsample of patients $i$ for whom $\textnormal{PFS}_i < T_i$, thus treating $\textnormal{PFS}_i$ as a left-truncation time $L_i$. 
\end{itemize}
Figure~\ref{fig:scatter:oak} is a scatterplot of the observed transition times $T_i$ for the different transitions in the analyzed data set. The left-hand and middle plots indicate a possibly difficult estimation problem in that there are few observed $0 \to 2$ transitions and many rather early transitions into the intermediate state $1$.

\begin{figure}[t!]
\centering
\includegraphics[width=\textwidth]{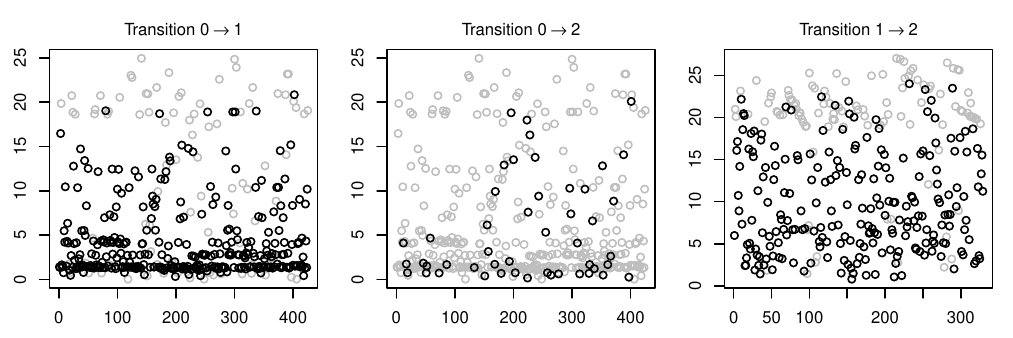}
\caption{Scatterplot of transition times in the atezolizumab arm of the data from \cite{rittmeyer2017atezolizumab}. Uncensored data points are depicted in black, right-censored points in grey. The $x$-axis displays arbitrary unit IDs, the $y$-axis displays time (in months).}\label{fig:scatter:oak}
\end{figure}

\subsection{Re-analysis of clinical trial data}

We consider data from \cite{rittmeyer2017atezolizumab}, available as supplement to \cite{gandara2018blood}, on 850 patients randomized in one-to-one fashion to either atezolizumab or docetaxel. In what follows, we restrict attention to the atezolizumab arm. (For the docetaxel arm, we found similar results not reported here.) As discussed by \cite{rittmeyer2017atezolizumab}, atezolizumab is a novel immunotherapy that aims at reestablishing anticancer immunity, while docetaxel is a more traditional therapy.

Our aim is to fit a parsimonious piecewise constant hazards parameterization of the illness-death model from above to the data at hand. The estimated piecewise constant transition hazards $\haz_{0 \to 1}$, $\haz_{0 \to 2}$ and $\haz_{1 \to 2}$ are then transformed into survival functions of PFS and OS, respectively, using well-known solutions to Kolmogorov forward differential equations \citep[see][Section II.6]{AndersenGill1993}, and are visually compared with standard Kaplan-Meier plots of PFS and OS.

To implement our methods, we make the following choices:
(i) To demonstrate that our methods are robust to the particular choice of the interval $[\tau_{\min},\tau_{\max}]$, we carry out the data analysis for different intervals: letting $\tau(p)$ be the empirical $p$-quantile of the uncensored event times $\{T_i: \delta_i = 1\}$ defined above for the different transitions, we set $\tau_{\min} = \tau(0)$ for the transitions $0 \to 1$ and $0 \to 2$, $\tau_{\min} = \tau(0.025)$ for transition $1 \to 2$ (taking into account the implicit left truncation) and $\tau_{\max} = \tau(p)$ with $p \in \{ 0.8, 0.975,0.99\}$. Note that for $p=0.8$, $\tau_{\max}$ and thus the estimation window is quite small. We have added the value $p=0.8$ for illustrative purposes but do not recommend to use such a small value in practice. 
(ii) To evaluate the influence of the tuning parameter $\lambda$ on our estimation results, we run the fused lasso with different choices of $\lambda$. In particular, $\lambda$ is chosen as explained in Section \ref{sec:impl} with $q \in \{0.1,0.5,0.9\}$ (as well as $K_{\max} = 20$ and $L=1000$). The larger $q$, the larger the tuning parameter $\lambda$ and thus the smaller the number of change points in the estimated hazard. Whereas the choice $q=0.9$ is in line with our theoretical considerations, $q=0.1$ and $q=0.5$ are quite low choices. In the application context at hand, it makes sense to consider $q$-values and thus $\lambda$-values considerably smaller than suggested by the theory for the following reason: while resulting in somewhat less parsimonious piecewise constant models with a larger number of change points, they may produce a better curve fit.

\begin{figure}[p]
\centering
\includegraphics[width=\textwidth]{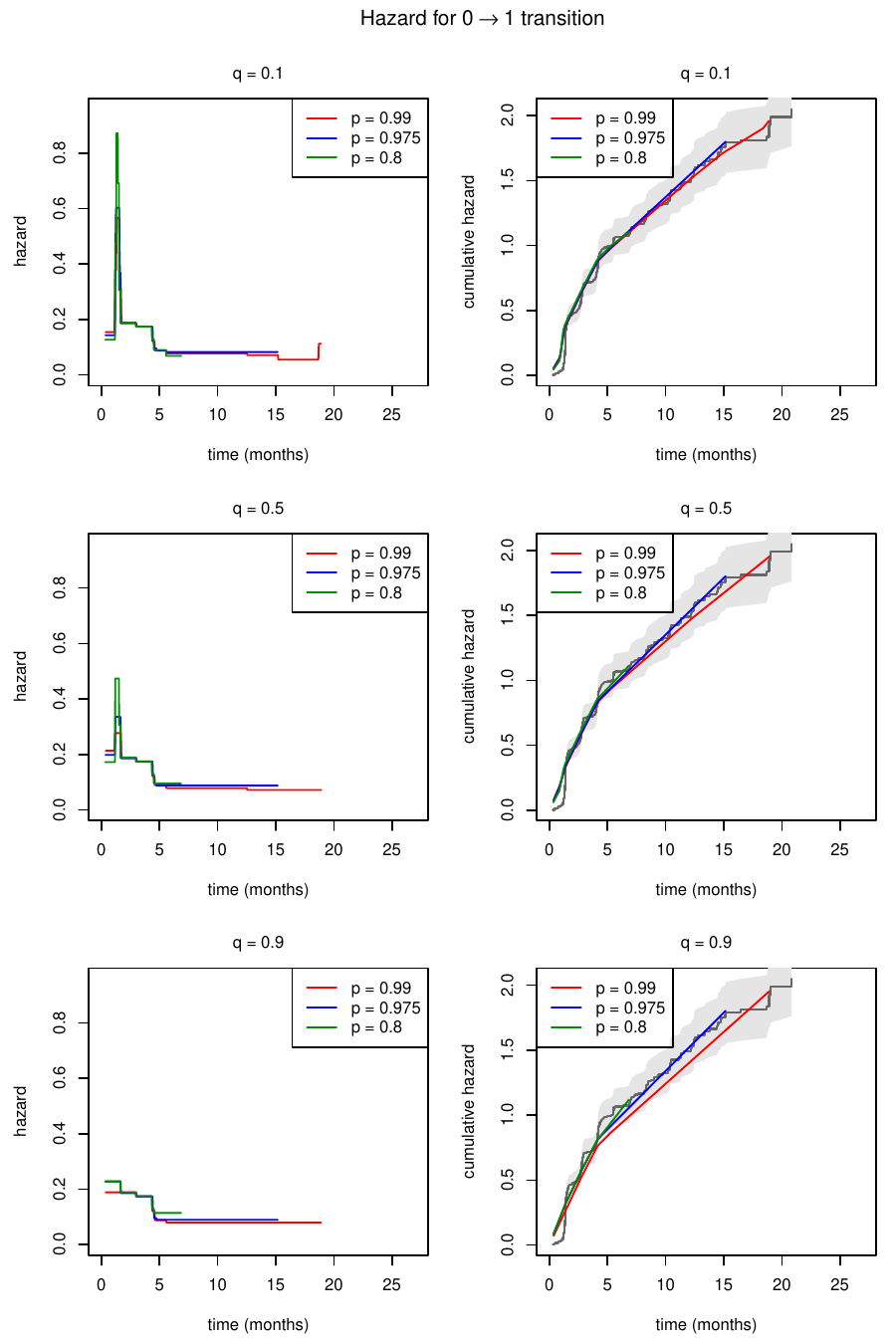}
\caption{Estimated hazard and cumulative hazard for the $0 \to 1$ transition with different choices of $p$ and $q$.}\label{fig:hazard01:oakMPDL}
\end{figure}

\begin{figure}[p]
\centering
\includegraphics[width=\textwidth]{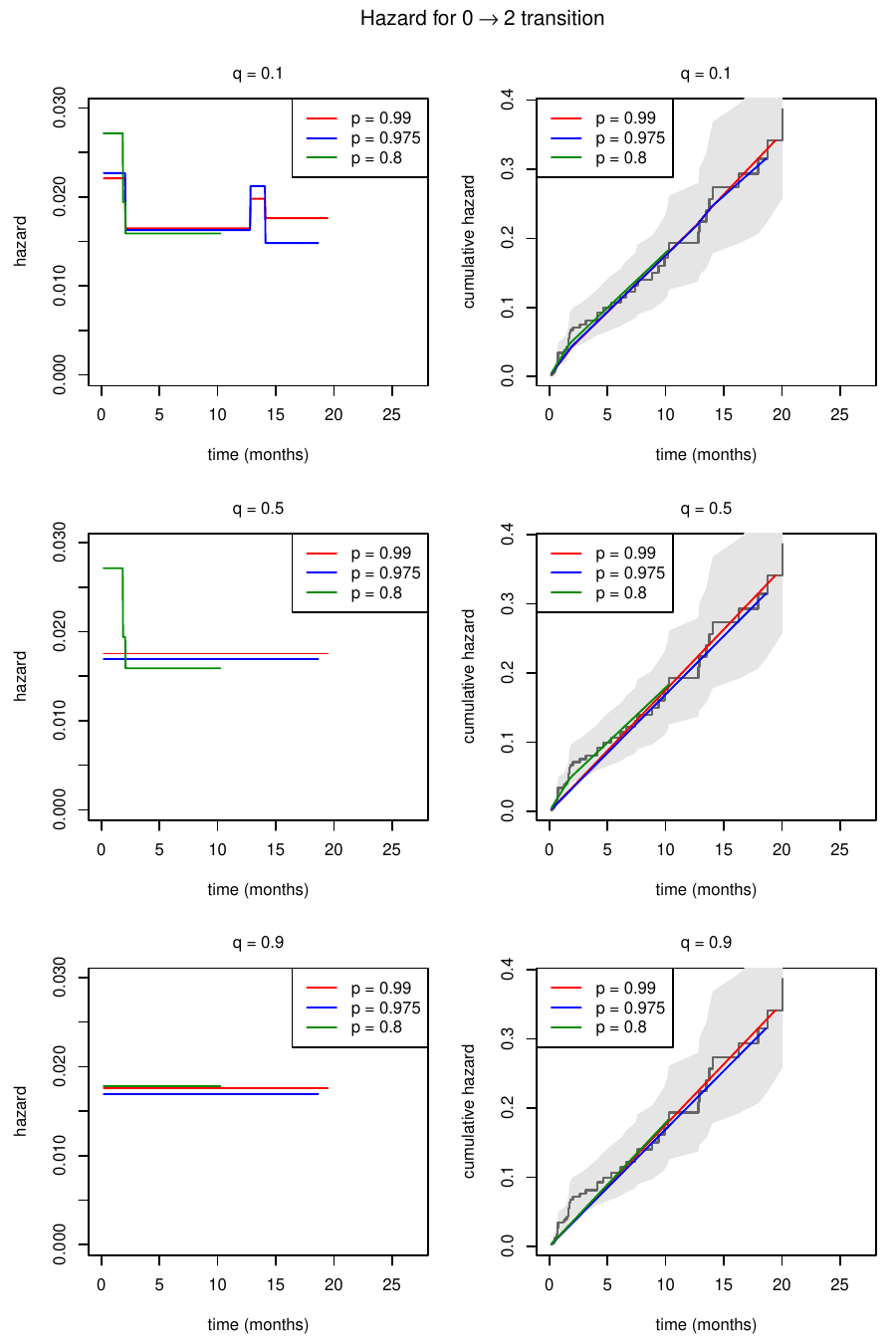}
\caption{Estimated hazard and cumulative hazard for the $0 \to 2$ transition with different choices of $p$ and $q$.}\label{fig:hazard02:oakMPDL}
\end{figure}

\begin{figure}[p]
\centering
\includegraphics[width=\textwidth]{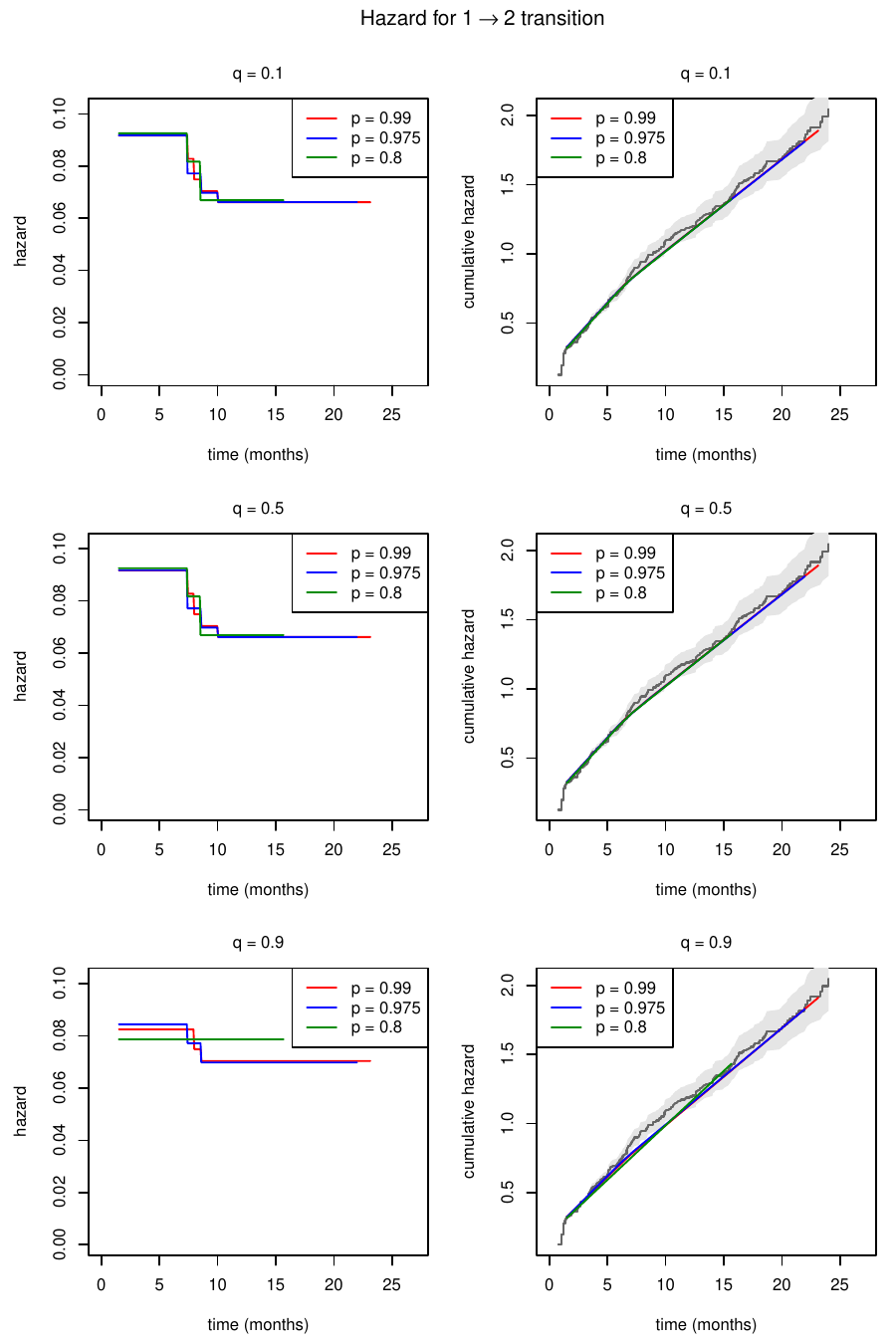}
\caption{Estimated hazard and cumulative hazard for the $1 \to 2$ transition with different choices of $p$ and $q$.}\label{fig:hazard12:oakMPDL}
\end{figure}

\begin{figure}[p]
\includegraphics[width=0.5\textwidth]{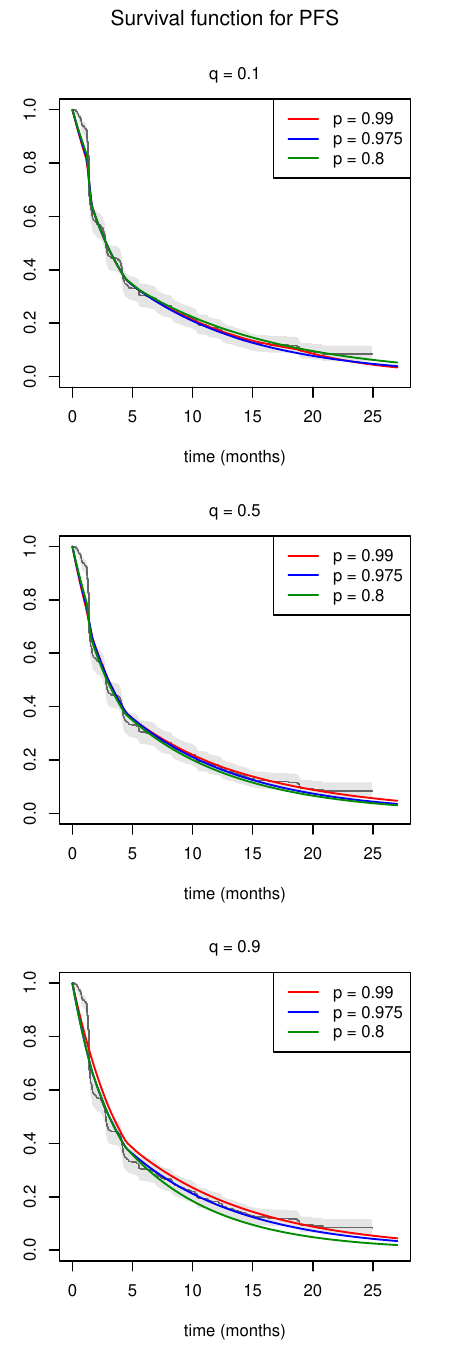}
\includegraphics[width=0.5\textwidth]{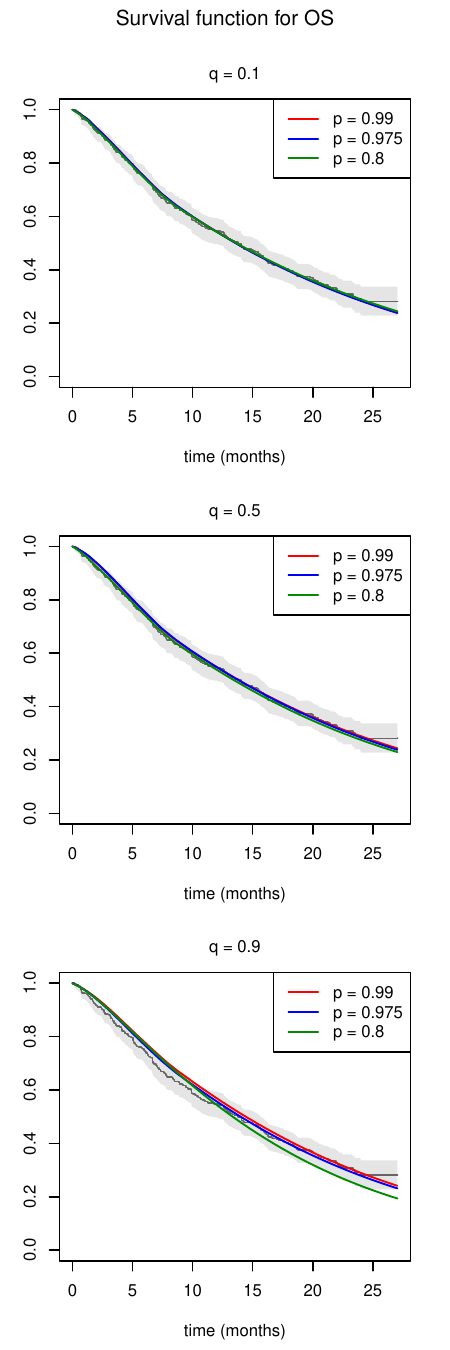}
\caption{Estimated survival curves for PFS and OS.}\label{fig:surv:oakMPDL}
\end{figure}

Figures \ref{fig:hazard01:oakMPDL}--\ref{fig:hazard12:oakMPDL} present the estimated piecewise constant hazards (left column) and the corresponding piecewise linear cumulative hazards (right column) for different choices of $p$ and $q$. Each figure is structured as follows: The plots in each row correspond to a specific choice of $q \in \{0.1,0.5,0.9\}$. The choice of $p$ is indicated by colour. In particular, the green, blue, red curve in each plot corresponds to the choice $p = 0.8,0.975,0.99$, respectively. The dark grey line in the right-hand plots is the Nelson-Aalen estimator of the cumulative hazard, the shaded area in light grey is the corresponding pointwise 95\% confidence band.
We briefly discuss the estimation results: 
(i) The estimated hazards in Figures \ref{fig:hazard01:oakMPDL}--\ref{fig:hazard12:oakMPDL} suggest that a model with constant transition hazards is too simple.
(ii) Overall, our approach produces a good fit to the data. This becomes visible in the piecewise linear estimates of the cumulative hazards which are very similar to their nonparametric Nelson-Aalen counterparts for all choices of $p$ and $q$. 
(iii) The estimated hazard curves are fairly similar for different choices of $p$.
More notable are the differences in the estimates across $q$. In particular, the larger $q$, the less pronounced the estimated jumps. This is not surprising as larger $q$ amounts to more shrinkage.
(iv) For the $0 \to 1$ transition, the Nelson-Aalen estimator has a staircase-like structure at earlier times (i.e., close to the time point $t=0$). This is most likely due to the fact that progression diagnosis comes in 'waves' during earlier times, a common feature in oncology trails. Notably, our piecewise linear estimates do not pick up this staircase structure but provide a smoother fit. 

An aggregated assessment of the estimation results is offered by the survival function plots in Figure \ref{fig:surv:oakMPDL}, which is structured analogously as the previous three figures. The plots depict estimators of the survival function for PFS (left column) and OS (right column) which are based on the estimated piecewise constant transition hazards. The dark grey line in the background is the Kaplan-Meier estimator of the survival function, the shaded area in light grey is the corresponding pointwise 95\% confidence band. \cite{titman2010model} have suggested comparing the pointwise confidence intervals of the Kaplan-Meier estimator with the survival curves resulting from the piecewise constant fit of the multi-state model as an informal model diagnostic. We here follow this suggestion:
The estimated PFS survival curves yield an overall good fit to the nonparametric Kaplan-Meier estimator, the fit being somewhat less precise at earlier times (i.e., close to $t=0$). Specifically, our estimates do not mimic the wave-like structure of the Kaplan-Meier estimator close to $t=0$. This is not surprising as this wave-like structure is most probably due to the staircase structure of the Nelson-Aalen estimator for the $0 \to 1$ transition (discussed above) which is not picked up by our estimates.
The estimated OS survival curves are quite close to the Kaplan-Meier estimator as well, the fit being somewhat worse for $q=0.9$, in particular, for $(q,p)=(0.9, 0.8)$. 
To summarize, Figure~\ref{fig:surv:oakMPDL} suggests that our method produces an overall adequate fit (except for the early 'wave' characteristic of the PFS Kaplan-Meier curve, which is a common phenomenon in oncology trials).

\bibliographystyle{ims}
{\small \setlength{\bibsep}{0.35em}
\bibliography{bibliography}}
\newpage

\begin{center}
{\LARGE \textbf{Technical Appendices}}
\end{center}


\def\thesection{\Alph{section}}
\setcounter{section}{0}
\def\theequation{A.\arabic{equation}}
\setcounter{equation}{0}
\section{Proof of main results}

In what follows, we prove the theoretical results from Section \ref{sec:theory}. Throughout this and the following appendix, the symbols $c$ and $C$ denote generic constants that may take a different value on each occurrence. Moreover, the symbols $c_j$ and $C_j$ with subscript $j$ (which may be either a natural number or a letter) are specific constants that are defined in the course of the appendices. Unless stated differently, the constants $c$, $C$, $c_j$ and $C_j$ do not depend on the sample size $n$.

\subsection*{Proof of Theorem \ref{theo:elementwise-bound}}

In order to prove Theorem \ref{theo:elementwise-bound}, we make use of the following general result on the fused lasso due to \cite{Zhang2023}; see also Theorem 3.1 in \cite{Zhang2019}. 
\begin{prop}\label{prop:fused-lasso-elementwise-error}
For any $\kappa_n > 0$ with the property that
\[ \max_{1 \le k \le \ell \le n} \left| \frac{1}{\sqrt{\ell-k+1}} \sum_{j=k}^{\ell} u_j \right| \le \kappa_n, \]
it holds that
\[ |\hat{\alpha}_{\lambda,j} - \haz_j| \le \max \left\{ \frac{\kappa_n}{\sqrt{d_j}}, \frac{\kappa_n^2}{4n\lambda},\frac{2n\lambda}{r_{k(j)}} + \frac{2\kappa_n}{\sqrt{r_{k(j)}}} \right\} \]
for $j \in \{1,\ldots,n\}$.
\end{prop}
To make the paper as self-contained as possible, we provide a proof of Proposition \ref{prop:fused-lasso-elementwise-error} in Appendix B. In order to apply Proposition \ref{prop:fused-lasso-elementwise-error}, we derive a (probabilistic) upper bound on the maximal partial sum $\max_{1 \le k \le \ell \le n} | (\ell-k+1)^{-1/2} \sum_{j=k}^{\ell} u_j|$ in the following proposition.

\begin{prop}\label{prop:Rosenthal-bound-error-partial-sum}
Assume that \ref{C0}--\ref{C6} are satisfied and let
\[ \kappa_n = c_n \max \left\{ n^{2/\nu}, \sqrt{n} \rho_n \right\}, \]
where $\nu > 4$ is specified in \ref{C5} and $\{c_n\}$ is a slowly diverging sequence (e.g.\ $c_n = c_0 \log \log n$ with some constant $c_0$). Then 
\[ \pr \left( \max_{1 \le k \le \ell \le n} \bigg| \frac{1}{\sqrt{\ell-k+1}} \sum_{j=k}^\ell u_j \bigg| > \kappa_n \right) = o(1). \] 
\end{prop}

\begin{proof}
Since $u_j = \Delta_j^\alpha + \Delta_j^J + \Delta_j^\beta + \eta_j$, it holds that
\[ \pr \left( \max_{1 \le k \le \ell \le n} \bigg| \frac{1}{\sqrt{\ell-k+1}} \sum_{j=k}^\ell u_j \bigg| > \kappa_n \right) \leq P_n^\alpha + P_n^J + P_n^\beta + P_n^\eta, \]
where
\begin{align*}
P_n^\alpha & = \pr \left( \max_{1 \le k \le \ell \le n} \bigg| \frac{1}{\sqrt{\ell-k+1}} \sum_{j=k}^\ell \Delta_j^\alpha \bigg| > \frac{\kappa_n}{4} \right) \\
P_n^J & = \pr \left( \max_{1 \le k \le \ell \le n} \bigg| \frac{1}{\sqrt{\ell-k+1}} \sum_{j=k}^\ell \Delta_j^J \bigg| > \frac{\kappa_n}{4} \right) \\
P_n^\beta & = \pr \left( \max_{1 \le k \le \ell \le n} \bigg| \frac{1}{\sqrt{\ell-k+1}} \sum_{j=k}^\ell \Delta_j^\beta \bigg| > \frac{\kappa_n}{4} \right) \\
P_n^\eta & = \pr \left( \max_{1 \le k \le \ell \le n} \bigg| \frac{1}{\sqrt{\ell-k+1}} \sum_{j=k}^\ell \eta_j \bigg| > \frac{\kappa_n}{4} \right).
\end{align*}
We analyze these four probabilities one after the other:
\begin{enumerate}[label=(\roman*)]

\item As the change points $\tau_1 < \ldots < \tau_K$ and their number $K$ are fixed, we have the following for sufficiently large $n$, where we set $\tau_0 := \tau_{\min}$ and $\tau_{K+1} := \tau_{\max}$ for notational convenience: for each $j \in \{1,\ldots,n\}$, either (a) $[t_{j-1}, t_j] \subseteq [\tau_{k-1},\tau_k]$ for some $k$ or (b) $t_{j-1} < \tau_k < t_j$ for some $k$. Case (a) occurs at least $n-K$ times, whereas case (b) occurs at most $K$ times. It is straightforward to see that in case (a), $\Delta_j^\alpha = 0$, whereas in case (b), 
\begin{align*}
|\Delta_j^\alpha|
 & = \left| \frac{\tau_{k}-t_{j-1}}{t_j-t_{j-1}} \mathfrak{a}_{k} + \frac{t_j-\tau_{k}}{t_j-t_{j-1}} \mathfrak{a}_{k+1} - \mathfrak{a}_{k+1} \right| \\
 & \le \frac{\tau_{k}-t_{j-1}}{t_j-t_{j-1}} |\mathfrak{a}_{k} - \mathfrak{a}_{k+1}| 
   \le \max_{1 \le k \le K} |\mathfrak{a}_{k} - \mathfrak{a}_{k+1}| =: C_{\mathfrak{a}} < \infty. 
\end{align*}
Since case (b) occurs at most $K$ times, we obtain that 
\[ \max_{1 \le k \le \ell \le n} \bigg| \frac{1}{\sqrt{\ell-k+1}} \sum_{j=k}^\ell \Delta_j^\alpha \bigg| \le K C_{\mathfrak{a}}. \]
With our choice of $\kappa_n$, this implies that $P_n^\alpha = 0$ for sufficiently large $n$.

\item Since $\Delta^J_j = \{ \Delta^J(t_j) - \Delta^J(t_{j-1}) \} / \{ t_j - t_{j-1} \}$ and $\Delta^J(t) =  \int_0^t (J(s,\bs{\beta})-1) \haz(s) ds$, we have   
\begin{align*}
P_n^J
 & \le \pr \left( \max_{1 \le k \le \ell \le n} \sqrt{\ell - k + 1} \max_{k \le j \le \ell} \big| \Delta_j^J \big| > \frac{\kappa_n}{4} \right) \\
 & \le \pr \left( \max_{1 \le j \le n} \big| \Delta_j^J \big| > \frac{\kappa_n}{4 \sqrt{n}} \right) \\
 & \le \pr \left( \max_{1 \le j \le n} \big|\Delta^J_j\big| > 0 \right) \\
 & \le \pr \left( \inf_{t \in [\tau_{\min},\tau_{\max}]} J(t,\bs{\beta}) = 0 \right) = o(1)
\end{align*}
by \ref{C4}.

\item By \ref{C6}, $\max_{1 \le j \le n} |\Delta^\beta_j| = O_p(\rho_n)$. Therefore,
\begin{align*}
P_n^\beta
 & \le \pr \left( \max_{1 \le k \le \ell \le n} \sqrt{\ell - k + 1} \max_{k \le j \le \ell} \big| \Delta_j^\beta \big| > \frac{\kappa_n}{4} \right) \\
 & \le \pr \left( \max_{1 \le j \le n} \big| \Delta_j^\beta \big| > \frac{\kappa_n}{4 \sqrt{n}} \right) = o(1),
\end{align*}
as long as $\kappa_n$ is such that $\kappa_n/\sqrt{n} \rho_n \to \infty$, which is ensured by our choice of $\kappa_n$.

\item A simple union bound yields that
\[ P_n^\eta \le \sum_{1 \le k \le \ell \le n} P_{k,\ell}^\eta \quad \text{with} \quad P_{k,\ell}^\eta = \pr \left( \bigg| \frac{1}{\sqrt{\ell-k+1}} \sum_{j=k}^\ell \eta_j \bigg| > \frac{\kappa_n}{4} \right). \]
Since $\eta_j = \{\eta(t_j) - \eta(t_{j-1})\} / \{t_j - t_{j-1}\} = n \{ \eta(t_j) - \eta(t_{j-1}) \}$, it holds that $\sum_{j=k}^\ell \eta_j = n ( \eta(t_\ell) - \eta(t_{k-1}) )$. Hence, we can apply Markov's inequality to get that 
\begin{align*}
P_{k,\ell}^\eta
 & \leq \frac{n^\nu \ex \left[ \big| \eta(t_\ell) - \eta(t_{k-1}) \big|^\nu \right]}{\{(\kappa_n/4) \sqrt{\ell-k+1}\}^\nu}
\end{align*}
for any $\nu > 0$. Our goal now is to derive a suitable upper bound on $\ex |\eta(t_\ell) - \eta(t_{k-1})|^\nu$. To do so, we regard $\eta(t_\ell) - \eta(t_{k-1})$ as a stochastic integral of the form 
\begin{equation*}
\eta(t_\ell) - \eta(t_{k-1}) = \int_0^\tau H_{k,\ell}(s) d\widebar{M}(s)
\end{equation*}
with
\[ H_{k,\ell}(s) = \ind(s \in (t_{k-1},t_\ell]) \frac{J(s,\bs{\beta})}{\widebar{Z}(s,\bs{\beta})} \]
and apply a version of Rosenthal's inequality to it. In particular, we use the following result due to \cite{Wood1999}. To formulate the result, we denote the intensity of the counting process $\widebar{N}$
by $\widebar{\lambda}$, that is, $\widebar{N}(t) = \int_0^t \widebar{\lambda}(s) ds + \widebar{M}(t)$ with $\widebar{\lambda}(s) = \widebar{Z}(s,\bs{\beta}) \haz(s)$. 
\begin{prop}\label{prop:Rosenthal}
Consider the stochastic integral 
\[ I_H(t) = \int_0^t H(s) d\widebar{M}(s), \]
where $H = \{H(t): t \in [0,\tau] \}$ is an $\{\mathcal{F}_t\}$-predictable process. If
(a) $\int_0^t \widebar{\lambda}(s) ds < \infty$ almost surely for any time point $t \in [0,\tau]$,
(b) $\ex [\int_0^\tau |H(s)| \widebar{\lambda}(s) ds ] < \infty$ and
(c) $\sup_{t \in [0,\tau]} \ex[I_H^2(t)] < \infty$, then  
\[ \ex \big[ |I_H(\tau)|^\nu \big] \leq C_\nu \ex \left[ \{\quadvar{I_H}(\tau)\}^{\nu/2} + \int_0^{\tau} |H(s)|^\nu \widebar{\lambda}(s) ds \right] \]
for any $\nu > 2$, where $C_\nu$ is a finite constant depending only on $\nu$.
\end{prop}
In order to apply this lemma with $H = H_{k,\ell}$, we verify its conditions: As $\widebar{Z}(\cdot,\bs{\beta})$ is predictable by \ref{C2}, so is the process $H_{k,\ell}$. Condition (a) is satisfied automatically in our framework \citep[see e.g.\ Theorem 2.3.1 in][]{FlemingHarrington2005}. Moreover, (b) holds since
\begin{align*}
\ex \int_0^\tau |H_{k,\ell}(s)| \widebar{\lambda}(s) ds 
 & \le \ex \int_0^\tau \ind(s \in (t_{k-1},t_\ell]) J(s,\bs{\beta})\haz(s) ds \\
 & \le \int_0^\tau \haz(s) ds = \cumhaz(\tau) 
\end{align*}
and (c) follows from the bound
\begin{align*}
\ex \bigg\{ \int_0^t H_{k,\ell}(s) d\widebar{M}(s) \bigg\}^2
& = \ex \int_0^t \left(  \ind(s \in (t_{k-1},t_\ell]) \frac{J(s,\bs{\beta})}{\widebar{Z}(s,\bs{\beta})} \right)^2 d\quadvar{\widebar{M}}(s) \\
& = \ex \int_0^t \left(  \ind(s \in (t_{k-1},t_\ell]) \frac{J(s,\bs{\beta})}{\widebar{Z}(s,\bs{\beta})} \right)^2 \widebar{\lambda}(s) ds \\
& = \ex \int_0^t |H_{k,\ell}(s)| \haz(s) ds \\
& \le \ex \int_{\tau_{\min}}^{\tau_{\max}} \bigg|\frac{J(s,\bs{\beta})}{\widebar{Z}(s,\bs{\beta})}\bigg| \haz(s) ds \\
& \le \left\{ \int_{\tau_{\min}}^{\tau_{\max}} \haz(s) ds \right\} \ex \sup_{s \in [\tau_{\min},\tau_{\max}]} \bigg|\frac{J(s,\bs{\beta})}{\widebar{Z}(s,\bs{\beta})}\bigg| < \infty,
\end{align*}
where $\ex \sup_{s \in [\tau_{\min},\tau_{\max}]} |J(s,\bs{\beta})/\widebar{Z}(s,\bs{\beta})| < \infty$ by \ref{C5}. We now compute the upper bound of the Rosenthal inequality from Proposition \ref{prop:Rosenthal} in our case. Let $\nu$ be a natural number strictly larger than $4$. Since
\[ \ex \bigg\| \frac{J(\cdot,\bs{\beta})}{\widebar{Z}(\cdot,\bs{\beta})} \bigg\|_{\infty}^{\nu} \le \frac{C_{\infty,\nu}}{n^\nu} \]
by \ref{C5}, it holds that
\begin{align*}
\ex \{\quadvar{I_{H_{k,\ell}}}(\tau)\}^{\nu/2}
 & = \ex \left\{ \int_0^\tau H_{k,\ell}^2(s) d\quadvar{\widebar{M}}(s) \right\}^{\nu/2} \\
 & = \ex \left\{ \int_0^\tau \left( \ind_{(t_{k-1},t_\ell]}(s) \frac{J(s,\bs{\beta})}{\widebar{Z}(s,\bs{\beta})} \right)^2 \widebar{\lambda}(s) ds \right\}^{\nu/2} \\
 & = \ex \left\{ \int_0^\tau \ind_{(t_{k-1},t_\ell]}(s) \bigg|\frac{J(s,\bs{\beta})}{\widebar{Z}(s,\bs{\beta})}\bigg| \haz(s) ds \right\}^{\nu/2} \\
 & \le \bigg\{ \int_0^\tau \ind_{(t_{k-1},t_\ell]}(s) \haz(s) ds \bigg\}^{\nu/2} \ex \bigg\| \frac{J(\cdot,\bs{\beta})}{\widebar{Z}(\cdot,\bs{\beta})} \bigg\|_{\infty}^{\nu/2} \\
 & \le \bigg\{ \frac{(\ell - k + 1)}{n} \| \haz \|_\infty \bigg\}^{\nu/2} \ex \bigg\| \frac{J(\cdot,\bs{\beta})}{\widebar{Z}(\cdot,\bs{\beta})} \bigg\|_{\infty}^{\nu/2} \\
 & \le C_{\infty,\nu/2} \bigg\{ \frac{(\ell - k + 1)}{n^2} \| \haz \|_\infty \bigg\}^{\nu/2} 
\end{align*}
and
\begin{align*}
\ex \int_0^\tau |H_{k,\ell}(s)|^\nu \widebar{\lambda}(s) ds 
 & \le \ex \int_0^\tau \ind_{(t_{k-1},t_\ell]}(s) \bigg| \frac{J(s,\bs{\beta})}{\widebar{Z}(s,\bs{\beta})} \bigg|^{\nu-1} \haz(s) ds \\
 & \le \bigg\{ \int_0^\tau \ind_{(t_{k-1},t_\ell]}(s) \haz(s) ds \bigg\} \ex \bigg\| \frac{J(\cdot,\bs{\beta})}{\widebar{Z}(\cdot,\bs{\beta})} \bigg\|_{\infty}^{\nu-1} \\
 & \le \bigg\{ \frac{(\ell - k + 1)}{n} \| \haz \|_\infty \bigg\} \ex \bigg\| \frac{J(\cdot,\bs{\beta})}{\widebar{Z}(\cdot,\bs{\beta})} \bigg\|_{\infty}^{\nu-1} \\
 & \le C_{\infty,\nu-1} \bigg\{ \frac{(\ell - k + 1)}{n^{\nu}} \| \haz \|_\infty \bigg\} 
\end{align*}
with $\| \haz \|_\infty = \sup_{t \in [0,1]} |\haz(t)|$. We thus arrive at the bound
\begin{align*}
\ex \big| \eta(t_\ell) - \eta(t_{k-1}) \big|^\nu
 & = \ex \bigg|\int_0^\tau H_{k,\ell}(s) d\widebar{M}(s)\bigg|^{\nu} \\
 & \le C_\nu \max\{C_{\infty,\nu/2},C_{\infty,\nu-1}\} \\ & \quad \times \left[ \bigg\{ \frac{(\ell - k + 1)}{n^2} \| \haz \|_\infty \bigg\}^{\nu/2} + \frac{(\ell - k + 1)}{n^{\nu}} \| \haz \|_\infty \right].
\end{align*}
Using this bound, we finally obtain that 
\begin{align*}
P_n^\eta
 & \le \sum_{1 \le k \le \ell \le n}  \frac{n^\nu \ex \left[ \big| \eta(t_\ell) - \eta(t_{k-1}) \big|^\nu \right]}{\{(\kappa_n/4) \sqrt{\ell-k+1}\}^\nu} 
 \le \frac{Cn^2}{\kappa_n^\nu} 
\end{align*}
with the constant $C = 2^{2\nu+1} C_\nu \max\{C_{\infty,\nu/2}, C_{\infty,\nu-1}\} \max\{\| \haz \|_\infty, \| \haz \|_\infty^{\nu/2}\}$. With our choice of $\kappa_n$, this implies that $P_n^\eta = o(1)$. \qedhere 
\end{enumerate}
\end{proof}

\subsection*{Proof of Theorem \ref{theo:l2-bound}}

By Theorem \ref{theo:elementwise-bound}, the following holds with probability tending to 1:  
\begin{align*}
\sum_{j=n_k}^{n_{k+1}-1} \big( \hat{\alpha}_{\lambda,j} - \haz_j \big)^2 
 & \le \sum_{j=1}^{r_k} \bigg\{ \bigg[\frac{\kappa_n}{\sqrt{j}}\bigg]^2 + \bigg[\frac{\kappa_n}{\sqrt{r_k +1 - j}}\bigg]^2 \\ & \phantom{\le \sum_{j=1}^{r_k} \bigg\{} + \frac{\kappa_n^4}{16(n\lambda)^2} + \frac{16(n\lambda)^2}{r_k^2} + \frac{16\kappa_n^2}{r_k} \bigg\} \\
 & \le C \left\{ \log(n) \kappa_n^2 + \frac{\kappa_n^4}{n \lambda^2} + n \lambda^2 + \kappa_n^2  \right\}
\end{align*}
for all $k \in \{0,\ldots,K\}$ with some sufficiently large constant $C$, where we have used that $\sum_{j=1}^m j^{-1} \le \log(m) + 1$ for any $m \in \naturals$ and $0 < c_r \le r_k/n \le C_r$ for all $k$ and sufficiently large $n$ with suitable constants $0 < c_r \le C_r < \infty$. We thus obtain that 
\begin{align*}
\frac{1}{n} \norm{\hat{\bs{\alpha}}_\lambda - \hazvec}_2^2  
 & = \frac{1}{n} \sum_{j=1}^n \big( \hat{\alpha}_{\lambda,j} - \haz_j \big)^2 \\
 & = \frac{1}{n} \sum_{k=0}^{K}\sum_{j=n_k}^{n_{k+1}-1} \big( \hat{\alpha}_{\lambda,j} - \haz_j \big)^2 \\
 & \le \frac{C}{n} \left\{ \log(n) \kappa_n^2 + \frac{\kappa_n^4}{n \lambda^2} + n \lambda^2 + \kappa_n^2  \right\}
\end{align*}
with probability tending to $1$. Choosing $\lambda$ such that $n \lambda^2 = \kappa_n^2$, that is, $\lambda = \kappa_n /\sqrt{n}$, we arrive at  
\begin{align*}
\frac{1}{n} \norm{\hat{\bs{\alpha}}_\lambda - \hazvec}_2^2
 & \le \frac{C \log(n) \kappa_n^2}{n}  = C \log(n) c_n^2 \max\{  n^{-1 + \frac{4}{\nu}}, \rho_n^2 \}
\end{align*}
with probability tending to $1$. This completes the proof.

\subsection*{Proof of Theorem \ref{theo:change-point-I}}

Let $J^* = \{ j \in \{2,\ldots, n\} : \haz_{j-1} \neq \haz_j \}$ and note that $J^* = \{ n_1,\ldots,n_K \}$ with $n_k = \lceil (\tau_k - \tau_{\min}) n \rceil$ for $k=1,\ldots,K$. Setting $n_0 = 1$ and $n_{K+1} = n+1$ for convenience, we denote the smallest distance between two jump indices $n_k$ and $n_{k+1}$ of $\hazvec$ by
\[ D_n = \min_{0 \le k \le K} (n_{k+1} - n_k) \]
and the smallest jump height by
\[ H_n = \min_{j \in J^*} |\haz_j - \haz_{j-1}|. \]
Theorem \ref{theo:change-point-I} is a direct consequence of Theorem 4 from \cite{LinSharpnackRinaldoTibshirani2017}, which says the following:

\begin{prop}\label{theo:LinSharpnackRinaldoTibshirani2017}
Let $\tilde{\bs{\alpha}} = (\tilde{\alpha}_1,\ldots,\tilde{\alpha}_n)^\top$ be an estimator of $\hazvec = (\haz_1,\ldots,\haz_n)^\top$ with the property that $n^{-1} \norm{\tilde{\bs{\alpha}} - \hazvec}_2^2 = O_p(R_n)$ and let $\tilde{J} = \{ j \in \{2,\ldots, n\} : \tilde{\alpha}_{j-1} \neq \tilde{\alpha}_j \}$. Assume that $n R_n / H_n^2 = o(D_n)$. Then
\[ d(\tilde{J} \, | \, J^* ) = O_p \Big( \frac{n R_n}{H_n^2} \Big). \]
\end{prop}

A proof of Proposition \ref{theo:LinSharpnackRinaldoTibshirani2017} is provided in Appendix B for completeness. We now verify that the assumptions of Proposition \ref{theo:LinSharpnackRinaldoTibshirani2017} are fulfilled in our case. As the number of jumps $K$ is fixed and finite, it holds that $D_n = \min_{0 \le k \le K} (\tau_{k+1} - \tau_k) n \ge c_D n$ with some constant $c_D > 0$. Similarly, $H_n \ge c_H$ with some $c_H > 0$. By Theorem \ref{theo:l2-bound},  
\[ \frac{1}{n} \norm{\hat{\bs{\alpha}}_\lambda - \hazvec}_2^2 = O_p(R_n) \quad \text{with} \quad R_n = c_n^2 \log(n) \max\big\{  n^{-1 + \frac{4}{\nu}}, \rho_n^2 \big\}. \]
Consequently,
\[ \frac{n R_n}{H_n^2} \le \frac{n R_n}{c_H^2} = \frac{1}{c_H^2} c_n^2 \log(n) \max\big\{  n^{\frac{4}{\nu}}, n \rho_n^2 \big\} = o(D_n), \]
which shows that the conditions of Proposition \ref{theo:LinSharpnackRinaldoTibshirani2017} are fulfilled. Defining $\hat{J}_\lambda = \{ j \in \{2,\ldots, n\} : \hat{\alpha}_{\lambda,j-1} \neq \hat{\alpha}_{\lambda,j} \}$, we can now apply Proposition \ref{theo:LinSharpnackRinaldoTibshirani2017} to obtain that 
\[ d(\hat{J}_\lambda \, | \, J^*) = O_p \Big( c_n^2 \log(n) \max\big\{  n^{\frac{4}{\nu}}, n \rho_n^2 \big\} \Big), \]
which is equivalent to
\[ d(\hat{\mathcal{S}}_\lambda \, | \, \mathcal{S}^*) = O_p\Big( c_n^2 \log(n) \max\big\{ n^{-1+\frac{4}{\nu}}, \rho_n^2 \big\} \Big). \]

\subsection*{Proof of Theorem \ref{theo:change-point-II}}

Theorem \ref{theo:elementwise-bound} immediately yields that with probability tending to $1$,
\begin{align*}
|\hat{\alpha}_{\lambda,\hat{\tau}n} - \hat{\alpha}_{\lambda,\hat{\tau}n-1}| 
 & \le |\hat{\alpha}_{\lambda,\hat{\tau}n} - \haz_{\hat{\tau}n}| + |\hat{\alpha}_{\lambda,\hat{\tau}n-1} - \haz_{\hat{\tau}n-1}| \\
 & \le 2 \max_{j \in \{\hat{\tau}n-1, \hat{\tau}n\}} \left\{ \frac{\kappa_n}{\sqrt{d_j}}, \frac{\kappa_n^2}{4\ngrid\lambda},\frac{2\ngrid\lambda}{r_{k(j)}} + \frac{2\kappa_n}{\sqrt{r_{k(j)}}} \right\} \\
 & \le C \max \left\{ \frac{\kappa_n}{\sqrt{\Delta_n}}, \frac{\kappa_n}{\sqrt{n}} \right\} =  C \frac{\kappa_n}{\sqrt{\Delta_n}}
\end{align*}
for all $\hat{\tau} \in \hat{\mathcal{S}}_\lambda^{\text{far}}$, where we have used that $\min_{1\le k \le K} |\hat{\tau} - \tau_k| \ge \Delta_n/n$, $0 < c_r \le r_k/n \le C_r$ for all $k$ and sufficiently large $n$ with suitable constants $0 < c_r \le C_r < \infty$, $\lambda =  \kappa_n / \sqrt{n}$, and $\Delta_n$ can always be chosen such that $\Delta_n \le n$.

\subsection*{Proof of Proposition \ref{prop:settingB}}

Proposition \ref{prop:settingB} can be proven by a simplified version of the technical arguments for Proposition \ref{prop:settingC}. Hence, we only prove the latter.

\subsection*{Proof of Proposition \ref{prop:settingC}}

In what follows, all probabilistic statements are to be understood with respect to the conditional probability measure $\pr^B$ given study entry. Formally, this measure is defined by $\pr^{B}(A) = \pr(A \cap B)/\pr(B)$ for all $A \in \mathcal{F}$, where $B = \bigcap_{i=1}^n B_i$ with $B_i = \{T_i^* > L_i\}$. In addition to $\pr^B$, we use the symbol $\ex^B$ to denote the expectation given study entry. As each competing risk $r \in \{1,\ldots,R\}$ is analyzed separately, we let $r$ be fixed throughout the proof and consider the corresponding multivariate counting process $N_{1:n,r} = (N_{1r},\ldots,N_{nr})$ on the filtered probability space $(\Omega,\mathcal{F},\{\mathcal{F}_t\},\pr^B)$. Here, $\mathcal{F}_t$ is the $\sigma$-algebra generated by the information available up to time $t$. Formally speaking, $\mathcal{F}_t = \sigma(N_{ir^\prime}(s),\ind(L_i < s \le T_i): 1 \le i \le n, \, 1 \le r^\prime \le R, \, 0 \le s \le t) \lor \sigma(\bs{W}_1,\ldots,\bs{W}_n)$.

Before we verify conditions \ref{C0}--\ref{C6}, we prove an auxiliary lemma. 
\begin{lemma}\label{lemma:binomial}
Let $Z \sim \textnormal{Bin}(n,q)$ be a binomially distributed random variable with $q \in (0,1)$, that is,
\[ \pr^B(Z=k) = \binom{n}{k} q^k (1-q)^{n-k} \quad \text{for } k=0,\ldots,n. \]
Then for any $\nu \in \naturals$,
\[ \ex^B \left[ \frac{\ind(Z>0)}{Z^\nu} \right] \leq \frac{(\nu+1)!}{(nq)^{\nu}}, \]
where we define
\[ \frac{\ind(Z>0)}{Z^\nu} =
\begin{cases}
1/Z^\nu & \text{for } Z > 0 \\
0      & \text{for } Z = 0.
\end{cases}
\]
\end{lemma}

\begin{proof}
We have
\begin{align*}
\ex^B \left[ \frac{\ind(Z > 0)}{Z^\nu} \right] 
& = \sum_{k=1}^n \frac{1}{k^\nu} \binom{n}{k} q^k (1-q)^{n-k} \\
& = \sum_{k=1}^n \frac{1}{k^\nu} \frac{n!}{(n-k)! k!} q^k (1-q)^{n-k} \\
& = \sum_{k=1}^n \frac{1}{k^\nu} \left( \prod_{i=1}^\nu \frac{(k+i)}{(k+i)} \right) \frac{n!}{(n-k)! k!} q^k (1-q)^{n-k} \\
& = \sum_{k=1}^n \underbrace{ \prod_{i=1}^\nu \frac{k+i}{k}}_{\leq \prod_{i=1}^{\nu} (1+i) = (\nu+1)!} \frac{n!}{(n-k)! (k+\nu)!} q^k (1-q)^{n-k} \\
& \leq (\nu+1)! \sum_{k=1}^n \frac{n!}{(n-k)! (k+\nu)!} q^k (1-q)^{n-k} \\
& = (\nu+1)! \sum_{k=1}^n \frac{n!}{[(n+\nu)-(k+\nu)]! (k+\nu)!} q^k (1-q)^{n-k} \\
& = (\nu+1)! \sum_{k=1}^n \left( \prod_{i=1}^\nu \frac{1}{n+i} \right) \frac{(n+\nu)!}{[(n+\nu)-(k+\nu)]! (k+\nu)!} \\
& \hspace{6.2cm} \times q^{k+\nu} q^{-\nu} (1-q)^{(n+\nu)-(k+\nu)} \\
& = (\nu+1)! q^{-\nu} \left( \prod_{i=1}^\nu \frac{1}{n+i} \right) \underbrace{\sum_{k=1}^n \binom{n+\nu}{k+\nu} q^{k+\nu} (1-q)^{(n+\nu) - (k+\nu)}}_{\begin{array}{l}
= \sum_{\ell=\nu+1}^{n+\nu} \binom{n+\nu}{\ell} q^\ell (1-q)^{(n+\nu)-\ell} \\
\leq \sum_{\ell=0}^{n+\nu} \binom{n+\nu}{\ell} q^\ell (1-q)^{(n+\nu)-\ell} = 1
\end{array}} \\
& \leq (\nu+1)! q^{-\nu} \left( \prod_{i=1}^\nu \frac{1}{n+i} \right) \\
& \leq (\nu+1)! (nq)^{-\nu}. \qedhere
\end{align*}
\end{proof}

We now verify conditions \ref{C0}--\ref{C6}.

\begin{proof}[Proof of \ref{C0}]
We need to show that the family of $\sigma$-algebras $\{\mathcal{F}_t\}$ is increasing, right-continuous and complete (i.e.\ can be completed). This follows from standard arguments \citep[see e.g.\ Section II.2 in][]{AndersenGill1993}.
\end{proof}

\begin{proof}[Proof of \ref{C1}]
By \ref{C1:settingC}, the data vectors $(T_i,\delta_i \cdot \varepsilon_i, L_i, \bs{W}_i)$ are i.i.d.\ across $i$ (with respect to $\pr^B$). This implies that the counting processes $N_{ir}$ defined by $N_{ir}(t) = \ind(L_i < T_i \le t, \delta_i\cdot\varepsilon_i = r)$ are i.i.d.\ across $i$ as well (with respect to $\pr^B$). 
\end{proof}

\begin{proof}[Proof of \ref{C2}]
Condition \ref{C2} is a direct consequence of \ref{C3:settingC}. In particular, since the at-risk process $\ind(L_i < \cdot \le T_i)$ is adapted to $\{\mathcal{F}_t\}$ and left-continuous, it is $\{\mathcal{F}_t\}$-predictable, which immediately implies that $Z_i(\cdot,\bs{\beta}_r) = \ind(L_i < \cdot \le T_i) \exp(\bs{\beta}_r^\top \bs{W}_i)$ and $\lambda_{ir}$ are $\{\mathcal{F}_t\}$-predictable as well. 
\end{proof}

\begin{proof}[Proof of \ref{C3}]
We have $\widebar{Z}(t,\bs{\beta}_r) = \sum_{i=1}^n Z_i(t,\bs{\beta}_r)$ with $Z_i(t,\bs{\beta}_r) = \ind(L_i < t \le T_i)$ $\exp(\bs{\beta}_r^\top \bs{W}_i)$ and $J(t,\bs{\beta}_r) = \ind(\widebar{Z}(t,\bs{\beta}_r) > 0)$. Since $\|\bs{W}_i\|_{\infty} \le C_W < \infty$ by \ref{C5:settingC} and $\| \bs{\beta}_r\|_1 \le C_\beta < \infty$ by \ref{C6:settingC}, it holds that $| \bs{\beta}_r^\top \bs{W}_i | \le \| \bs{\beta}_r \|_1 \| \bs{W}_i \|_{\infty} \le C_\beta C_W < \infty$. From this, it immediately follows that the process $\widebar{Z}(\cdot,\bs{\beta}_r)$ is  bounded. With the convention $0/0 := 0$, the process $J(\cdot,\bs{\beta}_r)/\widebar{Z}(\cdot,\bs{\beta}_r)$ is bounded as well. This in particular implies that $\widebar{Z}(\cdot,\bs{\beta}_r) $ and $J(\cdot,\bs{\beta}_r)/\widebar{Z}(\cdot,\bs{\beta}_r)$ are locally bounded. 
\end{proof}

\begin{proof}[Proof of \ref{C4}]
Let $J(t) = \ind( \widebar{Z}(t) > 0 )$ with $\widebar{Z}(t) = \sum_{i=1}^n \ind(L_i < t \le T_i)$ and write $\mathcal{T} = [\tau_{\min},\tau_{\max}]$ for short. Since $0 < \min_{1 \le i \le n} \exp(\bs{\beta}_r^\top \bs{W}_i) \le \max_{1 \le i \le n} \exp(\bs{\beta}_r^\top \bs{W}_i) < \infty$ with probability $1$, it holds that
\begin{equation*}
\pr^B \big( J(t,\bs{\beta}_r) = J(t) \text{ for all } t \in \mathcal{T} \big) = 1,
\end{equation*}
which implies that
\[ \pr^B \left( \inf_{t \in \mathcal{T}} J(t,\bs{\beta}_r) = 0 \right) = \pr^B \left( \inf_{t \in \mathcal{T}} J(t) = 0 \right). \]
Moreover, since $\widebar{Z}(t) = \sum_{i=1}^n \ind(L_i < t \land T_i \ge t) \ge \sum_{i=1}^n \ind(L_i < \tau_{\min} \land T_i \ge \tau_{\max})$ for all $t \in \mathcal{T}$,
\begin{align*}
\pr^B \left( \inf_{t \in \mathcal{T}} J(t) = 0 \right) 
& \le \pr^B \left( \sum_{i=1}^n \ind\big(L_i < \tau_{\min} \land T_i \ge \tau_{\max}\big) = 0 \right) \\
& = \prod_{i=1}^n \Big\{ 1  - \pr^B\big(L_i < \tau_{\min} \land T_i \ge \tau_{\max}\big) \Big\} \\
& = (1-p)^n = o(1),
\end{align*}
where we have used \ref{C4:settingC} in the last line. 
\end{proof}

\begin{proof}[Proof of \ref{C5}]
As before, let $\widebar{Z}(t) = \sum_{i=1}^n \ind(L_i < t \le T_i)$, $J(t) = \ind( \widebar{Z}(t) > 0 )$ and $\mathcal{T} = [\tau_{\min},\tau_{\max}]$. Since $\pr^B ( J(t,\bs{\beta}_r) = J(t) \text{ for all } t \in \mathcal{T} ) = 1$ as already shown in the proof of \ref{C4} and $| \bs{\beta}_r^\top \bs{W}_i | \le \| \bs{\beta}_r \|_1 \| \bs{W}_i \|_\infty \le C_\beta C_W$ by \ref{C5:settingC} and \ref{C6:settingC}, it holds that 
\begin{align*}
 & \ex^B \left[ \bigg\| \frac{J(\cdot,\bs{\beta}_r)}{\widebar{Z}(\cdot,\bs{\beta}_r)} \bigg\|_{\infty}^{\nu} \, \bigg| \, \bs{W}_1,\ldots, \bs{W}_n \right] \\
 & \le \frac{1}{\min_{1 \le i \le n} \exp(\nu \bs{\beta}_r^\top \bs{W}_i)} \, \ex^B \left[ \sup_{t \in \mathcal{T}} \bigg| \frac{J(t)}{\widebar{Z}(t)} \bigg|^{\nu} \, \bigg| \, \bs{W}_1,\ldots, \bs{W}_n \right] \\
 & \le \max_{1 \le i \le n} \exp(\nu |\bs{\beta}_r^\top \bs{W}_i|) \, \ex^B \left[ \sup_{t \in  \mathcal{T}} \bigg| \frac{J(t)}{\widebar{Z}(t)} \bigg|^{\nu} \, \bigg| \, \bs{W}_1,\ldots, \bs{W}_n \right] \\
 & \le \exp(\nu C_\beta C_W) \, \ex^B \left[ \sup_{t \in  \mathcal{T}} \bigg| \frac{J(t)}{\widebar{Z}(t)} \bigg|^{\nu} \, \bigg| \, \bs{W}_1,\ldots, \bs{W}_n \right] 
\end{align*}
almost surely and therefore
\[ \ex^B \left[ \bigg\| \frac{J(\cdot,\bs{\beta}_r)}{\widebar{Z}(\cdot,\bs{\beta}_r)} \bigg\|_{\infty}^{\nu} \right] \le \exp(\nu C_\beta C_W) \, \ex^B \left[ \sup_{t \in \mathcal{T}} \bigg| \frac{J(t)}{\widebar{Z}(t)} \bigg|^{\nu} \right]. \]
Next define $\widebar{Z}_{\mathcal{T}} := \sum_{i=1}^n \ind(L_i < \tau_{\min} \land T_i \ge \tau_{\max})$ and $J_{\mathcal{T}} := \ind(\widebar{Z}_{\mathcal{T}} > 0)$. Since $\widebar{Z}(t) \ge \widebar{Z}_{\mathcal{T}}$ and $J(t) \ge J_{\mathcal{T}}$ for any $t \in \mathcal{T}$ and $J(t)/\widebar{Z}(t)$ is upper bounded by $1$ (with the convention $0/0 := 0$),
\begin{align*}
\ex^B \left[ \sup_{t \in \mathcal{T}} \bigg| \frac{J(t)}{\widebar{Z}(t)} \bigg|^{\nu} \right]
 & = \ex^B \left[ \sup_{t \in \mathcal{T}} \bigg| \frac{J(t)}{\widebar{Z}(t)} \bigg|^{\nu} \ind\Big( \inf_{t \in \mathcal{T}} J(t) > 0\Big) \right] \\
 & \quad + \ex^B \left[ \sup_{t \in \mathcal{T}} \bigg| \frac{J(t)}{\widebar{Z}(t)} \bigg|^{\nu} \ind\Big( \inf_{t \in \mathcal{T}} J(t) = 0\Big) \right] \\
 & \le \ex^B \left[ \sup_{t \in \mathcal{T}} \bigg| \frac{J(t)}{\widebar{Z}(t)} \bigg|^{\nu} \ind\Big( \inf_{t \in \mathcal{T}} J(t) > 0\Big) \right] \\
 & \quad + \pr^B \left( \inf_{t \in \mathcal{T}} J(t) = 0 \right) \\
 & \le \ex^B \left[ \sup_{t \in \mathcal{T}} \bigg| \frac{J(t)}{\widebar{Z}(t)} \bigg|^{\nu} \ind\Big( \inf_{t \in \mathcal{T}} J(t) > 0 \land \inf_{t \in \mathcal{T}} J(t) = J_{\mathcal{T}} \Big) \right] \\
 & \quad + \ex^B \left[ \sup_{t \in \mathcal{T}} \bigg| \frac{J(t)}{\widebar{Z}(t)} \bigg|^{\nu} \ind\Big( \inf_{t \in \mathcal{T}} J(t) > 0 \land \inf_{t \in \mathcal{T}} J(t) > J_{\mathcal{T}} \Big) \right] \\
 & \quad + \pr^B \left( \inf_{t \in \mathcal{T}} J(t) = 0 \right) \\
 & \le \ex^B \left[ \bigg| \frac{J_{\mathcal{T}}}{\widebar{Z}_{\mathcal{T}}} \bigg|^{\nu} \right] + \pr^B \big( J_{\mathcal{T}} = 0 \big)  + \pr^B \left( \inf_{t \in \mathcal{T}} J(t) = 0 \right). 
\end{align*}
From the proof of \ref{C4}, we already know that $\pr^B(J_{\mathcal{T}} = 0) \le (1-p)^n$ as well as $\pr^B(\inf_{t \in \mathcal{T}} J(t) = 0) \le (1-p)^n$. Moreover, since $\widebar{Z}_{\mathcal{T}} = \sum_{i=1}^n \ind(L_i < \tau_{\min} \land T_i \ge \tau_{\max})$ follows a binomial distribution $\text{Bin}(n,q)$ with $q = \pr^B(L_i < \tau_{\min} \land T_i \ge \tau_{\max}) = p$ under \ref{C4:settingC}, we can apply Lemma \ref{lemma:binomial} to get that
\[ \ex^B \left[ \bigg| \frac{J_{\mathcal{T}}}{\widebar{Z}_{\mathcal{T}}} \bigg|^{\nu} \right] \leq \frac{(\nu+1)!}{\{np\}^{\nu}}. \]
Putting everything together, we arrive at
\[ \ex^B \left[ \bigg\| \frac{J(\cdot,\bs{\beta}_r)}{\widebar{Z}(\cdot,\bs{\beta}_r)} \bigg\|_{\infty}^{\nu} \right] \le \exp(\nu C_\beta C_W) \left\{ \frac{(\nu+1)!}{\{np\}^{\nu}} + 2(1-p)^n \right\} \]
for any $\nu \in \naturals$. As $(1-p)^n$ converges to zero faster than $n^{-\nu}$, \ref{C5} is an immediate consequence of this bound.  
\end{proof}

\begin{proof}[Proof of \ref{C6}]
We show that
\begin{equation*}
\max_{1 \le k \le n} \big|\Delta_k^\beta\big| = O_p\Big( \frac{\log(n)}{\sqrt{n}} \Big) 
\end{equation*}
(where the notation $O_p(\cdot)$ and $o_p(\cdot)$ is to be understood with respect to $\pr^B$). Our proof exploits the following facts about the random variables $\widebar{Z}(t,\bs{b}) = \sum_{i=1}^n \ind(L_i < t \le T_i) \exp(\bs{b}^\top \bs{W}_i)$ and $J(t,\bs{b}) = \ind(\widebar{Z}(t,\bs{b}) > 0)$ with $\bs{b} \in \reals^d$: 
\begin{enumerate}[label=(\alph*),leftmargin=0.75cm]

\item \label{proof:settingB-C6-a} It holds that
\[ \pr^B \big( J(T_i, \hat{\bs{\beta}}_r) = J(T_i,\bs{\beta}_r) \text{ for all } i =1,\ldots,n \big) = 1. \]
\begin{proof}
Since $J(t,\bs{b}) = \ind( \sum_{i=1}^n \ind(L_i < t \le T_i) \exp(\bs{b}^\top \bs{W}_i) > 0)$, \ref{proof:settingB-C6-a} is a direct consequence of the fact that $0 < \min_{1 \le i \le n} \exp(\bs{\beta}^\top \bs{W}_i) \le \max_{1 \le i \le n} \exp(\bs{\beta}^\top \bs{W}_i) < \infty$ and $0 < \min_{1 \le i \le n} \exp(\hat{\bs{\beta}}_r^\top \bs{W}_i)  \le \max_{1 \le i \le n} \exp(\hat{\bs{\beta}}_r^\top \bs{W}_i) < \infty$ with probability $1$. 
\end{proof}

\item \label{proof:settingB-C6-b} It holds that 
\[ \sup_{t \in \mathcal{T}} \big| \widebar{Z}(t,\hat{\bs{\beta}}_r) - \widebar{Z}(t,\bs{\beta}_r) \big| = O_p\big(\log(n) \sqrt{n}\big). \]
\begin{proof} 
With the help of a Taylor expansion, we get that 
\begin{align}
 & \sup_{t \in \mathcal{T}} \big| \widebar{Z}(t,\hat{\bs{\beta}}_r) - \widebar{Z}(t,\bs{\beta}_r) \big| \nonumber \\*
 & \le \sup_{t \in  \mathcal{T}} \sum_{i=1}^n \ind(L_i < t \le T_i) \big| \exp(\hat{\bs{\beta}}_r^\top \bs{W}_i) - \exp(\bs{\beta}_r^\top \bs{W}_i) \big| \nonumber \\
 & \le \sum_{i=1}^n \big| \exp(\hat{\bs{\beta}}_r^\top \bs{W}_i) - \exp(\bs{\beta}_r^\top \bs{W}_i) \big| \nonumber \\
 & = \sum_{i=1}^n \big| \exp(\tilde{\bs{\beta}}_r^\top \bs{W}_i) \bs{W}_i^\top (\hat{\bs{\beta}}_r - \bs{\beta}_r) \big| \nonumber \\
 & \le \sum_{i=1}^n \exp(\tilde{\bs{\beta}}_r^\top \bs{W}_i) \| \bs{W}_i \|_\infty \| \hat{\bs{\beta}}_r - \bs{\beta}_r \|_1 \nonumber \\
 & = \sum_{i=1}^n \exp(\tilde{\bs{\beta}}_r^\top \bs{W}_i) \| \bs{W}_i \|_\infty \| \hat{\bs{\beta}}_r - \bs{\beta}_r \|_1 \, \ind\left( \| \hat{\bs{\beta}}_r - \bs{\beta}_r \|_1 \le \frac{\log(n)}{\sqrt{n}} \right) \nonumber \\
 & \quad + \sum_{i=1}^n \exp(\tilde{\bs{\beta}}_r^\top \bs{W}_i) \| \bs{W}_i \|_\infty \| \hat{\bs{\beta}}_r - \bs{\beta}_r \|_1 \, \ind\left( \| \hat{\bs{\beta}}_r - \bs{\beta}_r \|_1 > \frac{\log(n)}{\sqrt{n}}\right), \label{proof:settingB-C6-b-bound1}
\end{align}
where $\tilde{\bs{\beta}}_r$ is an intermediate value between $\hat{\bs{\beta}}_r$ and $\bs{\beta}_r$. In the case that $\| \hat{\bs{\beta}}_r - \bs{\beta}_r \|_1 \le \log(n) n^{-1/2}$, $\exp(\tilde{\bs{\beta}}_r^\top \bs{W}_i) \le \exp(\|\tilde{\bs{\beta}}_r\|_1 \|\bs{W}_i\|_\infty) \le \exp(2 C_\beta C_W)$ for sufficiently large $n$ because $\tilde{\bs{\beta}}_r$ lies between $\hat{\bs{\beta}}_r$ and $\bs{\beta}_r$ and thus $\| \tilde{\bs{\beta}}_r \|_1 \le \max \{ \| \bs{\beta}_r \|_1,$ $\| \hat{\bs{\beta}}_r \|_1 \} \le \| \bs{\beta}_r \|_1 + \log(n) n^{-1/2} \le 2 C_\beta$ for sufficiently large $n$. Consequently, 
\begin{align}
\sum_{i=1}^n \exp & (\tilde{\bs{\beta}}_r^\top \bs{W}_i) \| \bs{W}_i \|_\infty \| \hat{\bs{\beta}}_r - \bs{\beta}_r \|_1 \, \ind\left( \| \hat{\bs{\beta}}_r - \bs{\beta}_r \|_1 \le \frac{\log(n)}{\sqrt{n}} \right) \nonumber \\
 & \le \frac{\exp(2 C_\beta C_W) \log(n)}{\sqrt{n}} \sum_{i=1}^n \| \bs{W}_i \|_\infty = O_p(\log(n) \sqrt{n}). \label{proof:settingB-C6-b-bound2}
\end{align}
Moreover, for any given $\varepsilon > 0$,
\begin{align*}
 & \pr \left( \sum_{i=1}^n \exp(\tilde{\bs{\beta}}_r^\top \bs{W}_i) \| \bs{W}_i \|_\infty \| \hat{\bs{\beta}}_r - \bs{\beta}_r \|_1 \, \ind\left( \| \hat{\bs{\beta}}_r - \bs{\beta}_r \|_1 > \frac{\log(n)}{\sqrt{n}} \right) > \varepsilon \right) \\
 & \hspace{6.2cm} \le \pr \left( \| \hat{\bs{\beta}}_r - \bs{\beta}_r \|_1 > \frac{\log(n)}{\sqrt{n}} \right) = o(1),
\end{align*}
where the last equality follows from \ref{C6:settingC}, and thus
\begin{equation}\label{proof:settingB-C6-b-bound3}
\sum_{i=1}^n \exp(\tilde{\bs{\beta}}_r^\top \bs{W}_i) \| \bs{W}_i \|_\infty \| \hat{\bs{\beta}}_r - \bs{\beta}_r \|_1 \, \ind\left( \| \hat{\bs{\beta}}_r - \bs{\beta}_r \|_1 > \frac{\log(n)}{\sqrt{n}} \right) = o_p(1). 
\end{equation}
Plugging \eqref{proof:settingB-C6-b-bound2} and \eqref{proof:settingB-C6-b-bound3} into \eqref{proof:settingB-C6-b-bound1} yields \ref{proof:settingB-C6-b}.
\end{proof}

\item \label{proof:settingB-C6-c} It holds that $\inf_{t \in \mathcal{T}} \widebar{Z}(t,\bs{\beta}_r) \ge c n$ with probability tending to $1$, where $c$ is any positive constant with $c < p \exp(- C_\beta C_W)$.
\begin{proof}
Under our assumptions, 
\begin{align*}
\Big| \frac{\inf_{t \in \mathcal{T}} \widebar{Z}(t,\bs{\beta}_r)}{n} \Big|
 & = \inf_{t \in \mathcal{T}} \left\{ \frac{1}{n} \sum_{i=1}^n \exp(\bs{\beta}_r^\top \bs{W}_i) \ind(L_i < t \le T_i) \right\} \\
 & \ge \frac{1}{n} \sum_{i=1}^n \exp(\bs{\beta}_r^\top \bs{W}_i) \ind(L_i < \tau_{\min} \land T_i \ge \tau_{\max}) \\
 & = \frac{1}{n} \sum_{i=1}^n \ex \Big[ \exp(\bs{\beta}_r^\top \bs{W}_i) \ind(L_i < \tau_{\min} \land T_i \ge \tau_{\max}) \Big] + O_p\Big( \frac{1}{\sqrt{n}} \Big) \\
 & \ge \exp(-C_\beta C_W) \frac{1}{n} \sum_{i=1}^n \ex \big[ \ind(L_i < \tau_{\min} \land T_i \ge \tau_{\max}) \big] + O_p\Big( \frac{1}{\sqrt{n}} \Big) \\
 & = p \exp(-C_\beta C_W) + O_p\Big( \frac{1}{\sqrt{n}} \Big),
\end{align*}
where the second equality follows from a simple application of the weak law of large numbers. This immediately implies \ref{proof:settingB-C6-c}.
\end{proof}

\item \label{proof:settingB-C6-d} It holds that $\inf_{t \in \mathcal{T}} \widebar{Z}(t,\hat{\bs{\beta}}_r) \ge c n$ with probability tending to $1$, where $c$ is any positive constant with $c < p \exp(- C_\beta C_W)$.
\begin{proof}
As by (b), $\inf_{t \in \mathcal{T}} \widebar{Z}(t,\hat{\bs{\beta}}_r) \ge \inf_{t \in \mathcal{T}} \widebar{Z}(t,\bs{\beta}_r) - \sup_{t \in \mathcal{T}} |\widebar{Z}(t,\hat{\bs{\beta}}_r) - \widebar{Z}(t,\bs{\beta}_r)| = \inf_{t \in \mathcal{T}} \widebar{Z}(t,\bs{\beta}_r) + O_p(\log(n) \sqrt{n})$, we can use (c) to complete the proof.
\end{proof}

\end{enumerate}
We now bound $\max_{1 \le k \le n} |\Delta_k^\beta|$ with the help of (a)--(d). Using (a), we obtain that
\begin{align}
\max_{1 \le k \le n} \big|\Delta_k^\beta\big| \nonumber
 & = n \max_{1 \le k \le n} \bigg| \int \ind(s \in (t_{k-1},t_k]) \frac{J(s,\hat{\bs{\beta}}_r)}{\widebar{Z}(s,\hat{\bs{\beta}}_r)} d\widebar{N}_r(s) \\
 & \phantom{= n \max_{1 \le k \le n} \bigg| \quad} - \int \ind(s \in (t_{k-1},t_k]) \frac{J(s,\bs{\beta}_r)}{\widebar{Z}(s,\bs{\beta}_r)} d\widebar{N}_r(s) \bigg| \nonumber \\
 & = n \max_{1 \le k \le n} \bigg| \sum_{\{ i: \, t_{k-1} < T_i \le t_k, \, \delta_i \cdot \varepsilon_i = r\}} \bigg\{ \frac{J(T_i,\hat{\bs{\beta}}_r)}{\widebar{Z}(T_i,\hat{\bs{\beta}}_r)} - \frac{J(T_i,\bs{\beta}_r)}{\widebar{Z}(T_i,\bs{\beta}_r)} \bigg\} \bigg| \nonumber \\
 & = n \max_{1 \le k \le n} \bigg| \sum_{\{ i: \, t_{k-1} < T_i \le t_k, \, \delta_i \cdot \varepsilon_i = r\}} J(T_i,\bs{\beta}_r) \bigg\{ \frac{1}{\widebar{Z}(T_i,\hat{\bs{\beta}}_r)} - \frac{1}{\widebar{Z}(T_i,\bs{\beta}_r)} \bigg\} \bigg| \nonumber \\
 & \le n \max_{1 \le k \le n} \bigg\{ \sum_{\{ i: \, t_{k-1} < T_i \le t_k, \, \delta_i \cdot \varepsilon_i = r\}} \bigg| \frac{1}{\widebar{Z}(T_i,\hat{\bs{\beta}}_r)} - \frac{1}{\widebar{Z}(T_i,\bs{\beta}_r)} \bigg| \bigg\} \nonumber \\
 & = n \max_{1 \le k \le n} \bigg\{ \sum_{\{ i: \, t_{k-1} < T_i \le t_k, \, \delta_i \cdot \varepsilon_i = r\}} \bigg| \frac{\widebar{Z}(T_i,\bs{\beta}_r) - \widebar{Z}(T_i,\hat{\bs{\beta}}_r)}{\widebar{Z}(T_i,\hat{\bs{\beta}}_r) \widebar{Z}(T_i,\bs{\beta}_r)} \bigg| \bigg\} \nonumber \\
 & \le \frac{n}{\{ \inf_{t \in \mathcal{T}} \widebar{Z}(t,\hat{\bs{\beta}}_r) \} \{ \inf_{t \in \mathcal{T}} \widebar{Z}(t,\bs{\beta}_r) \}} \nonumber \\
 & \qquad \times \max_{1 \le k \le n} \bigg\{ \sum_{\{ i: \, t_{k-1} < T_i \le t_k, \delta_i \cdot \varepsilon_i = r\}} \big| \widebar{Z}(T_i,\bs{\beta}_r) - \widebar{Z}(T_i,\hat{\bs{\beta}}_r) \big| \bigg\} \label{eq:proof_C5_settingB_bound1}
\end{align}
with probability $1$. From (c) and (d), it directly follows that
\begin{equation}\label{eq:proof_C5_settingB_bound2}
 \frac{n}{\{ \inf_{t \in \mathcal{T}} \widebar{Z}(t,\hat{\bs{\beta}}_r) \} \{ \inf_{t \in \mathcal{T}} \widebar{Z}(t,\bs{\beta}_r) \}} = O_p\Big(\frac{1}{n}\Big).
\end{equation}  
Moreover, condition \ref{C2:settingC} and standard arguments for uniform convergence yield that $\max_{1 \le k \le n} \{ \sum_{i=1}^n \ind(t_{k-1} < T_i^* \le t_k) \} = O_p(1)$. Combining this with (b), we get that 
\begin{align}
\max_{1 \le k \le n} & \bigg\{ \sum_{\{ i: \, t_{k-1} < T_i \le t_k, \delta_i \cdot \varepsilon_i = r\}} \big| \widebar{Z}(T_i,\bs{\beta}_r) - \widebar{Z}(T_i,\hat{\bs{\beta}}_r) \big| \bigg\} \nonumber \\
 & \le O_p\big(\log(n) \sqrt{n}\big) \max_{1 \le k \le n} \Big\{ \sum_{i=1}^n \ind(t_{k-1} < T_i^* \le t_k) \Big\} = O_p(\log(n) \sqrt{n}). \label{eq:proof_C5_settingB_bound3}
\end{align}
Plugging \eqref{eq:proof_C5_settingB_bound3} and \eqref{eq:proof_C5_settingB_bound2} into \eqref{eq:proof_C5_settingB_bound1}, we finally arrive at the desired result 
\begin{equation*}
\max_{1 \le k \le n} \big|\Delta_k^\beta\big| = O_p\Big( \frac{\log(n)}{\sqrt{n}} \Big). \qedhere
\end{equation*}
\end{proof}

\subsection*{Proof of Proposition \ref{prop:settingD}}

As each direct $\ell \to m$ transition with $\ell, m \in \{0,\ldots,R\}$, $\ell \ne m$, is analyzed separately, we let the pair $(\ell,m)$ be fixed throughout the proof and consider the corresponding multivariate counting process $N_{1:n,\ell \to m} = (N_{1,\ell \to m},\ldots,N_{n,\ell \to m})$ on the filtered probability space $(\Omega,\mathcal{F},\{\mathcal{F}_t\},\pr)$. Here, $\mathcal{F}_t$ is the $\sigma$-algebra generated by the information available up to time $t$, which is given by $\mathcal{F}_t = \sigma(X_1(s),\ldots,X_n(s): 0 \le s \le t)$.

\begin{proof}[Proof of \ref{C0}]
Standard arguments show that $\{\mathcal{F}_t\}$ is increasing, right-continuous and complete (i.e.\ can be completed); see Section II.2 in \cite{AndersenGill1993}.
\end{proof}

\begin{proof}[Proof of \ref{C1}]
By \ref{C1:settingD}, the Markov processes $X_i$ are i.i.d.\ across $i$, which implies that the counting processes $N_{i,\ell \to m}$ defined by $N_{i,\ell \to m}(t) = \# \{s \le t:  X_i(s-) = \ell \text{ and } X_i(s) = m \}$ are i.i.d.\ across $i$ as well. 
\end{proof}

\begin{proof}[Proof of \ref{C2}]
By Theorem II.6.8 in \cite{AndersenGill1993} and absolute continuity of the intensity measure as assumed in \ref{C1:settingD}, the process $N_{i,\ell \to m}$ has $\{\mathcal{F}_t\}$-compensator $\Lambda_{i,\ell \to m}(t) = \int_0^t \lambda_{i,\ell \to m}(s) ds$, where $\lambda_{i,\ell \to m}(s) = Z_{i,\ell}(s) \haz_{\ell \to m}(s)$ with $Z_{i,\ell}(s) = \ind(X_i(s-) = \ell)$ and a deterministic, non-negative function $\haz_{\ell \to m}$. Since $Z_{i,\ell}$ is adapted to $\{\mathcal{F}_t\}$ and left-continuous, it is $\{\mathcal{F}_t\}$-predictable, which immediately implies that $\lambda_{i,\ell \to m}$ is $\{\mathcal{F}_t\}$-predictable as well.
\end{proof}

\begin{proof}[Proof of \ref{C3}]
We have $\widebar{Z}_{\ell}(t) = \sum_{i=1}^n Z_{i,\ell}(t)$ with $Z_{i,\ell}(t) = \ind(X_i(t-) = \ell)$ and $J_\ell(t) = \ind(\widebar{Z}_\ell(t) > 0)$. Hence, with the convention $0/0 := 0$, the processes $\widebar{Z}_\ell$ and $J_\ell/\widebar{Z}_\ell$ are obviously bounded, which in particular implies that they are locally bounded. 
\end{proof}

\begin{proof}[Proof of \ref{C4}]
With $\mathcal{T} = [\tau_{\min},\tau_{\max}]$, it holds that
\begin{align*}
\pr \left( \inf_{t \in \mathcal{T}} J_\ell(t) = 0 \right) 
& = \pr \left( \inf_{t \in \mathcal{T}} \sum_{i=1}^n \ind(X_i(t-) = \ell) = 0 \right) \\
& \le \pr \left( \sum_{i=1}^n \inf_{t \in \mathcal{T}} \ind(X_i(t-) = \ell) = 0 \right) \\
& = \prod_{i=1}^n \Big\{ 1 - \pr \big(X_i(t-) = \ell \text{ for all } t \in \mathcal{T}\big) \Big\} \\
& = (1-p_\ell)^n = o(1), 
\end{align*}
where we have used \ref{C2:settingD} in the last line. 
\end{proof}

\begin{proof}[Proof of \ref{C5}]
As before, let $\widebar{Z}_{\ell}(t) = \sum_{i=1}^n Z_{i,\ell}(t)$ with $Z_{i,\ell}(t) = \ind(X_i(t-) = \ell)$ and $J_\ell(t) = \ind(\widebar{Z}_\ell(t) > 0)$. Moreover, define $\widebar{Z}_{\ell,\mathcal{T}} = \sum_{i=1}^n \inf_{t \in \mathcal{T}} Z_{i,\ell}(t)$ and $J_{\ell,\mathcal{T}} = \ind(\widebar{Z}_{\ell,\mathcal{T}} > 0)$. With the convention $0/0 := 0$, it holds that 
\begin{align*}
\ex \left[ \bigg\| \frac{J_\ell}{\widebar{Z}_\ell} \bigg\|_{\infty}^{\nu} \right]
 & = \ex \left[ \sup_{t \in \mathcal{T}} \bigg| \frac{J_\ell(t)}{\widebar{Z}_\ell(t)} \bigg|^{\nu} \ind\Big( \inf_{t \in \mathcal{T}} J_\ell(t) > 0\Big) \right] 
   + \ex \left[ \sup_{t \in \mathcal{T}} \bigg| \frac{J_\ell(t)}{\widebar{Z}_\ell(t)} \bigg|^{\nu} \ind\Big( \inf_{t \in \mathcal{T}} J_\ell(t) = 0\Big) \right] \\
 & \le \ex \left[ \sup_{t \in \mathcal{T}} \bigg| \frac{J_\ell(t)}{\widebar{Z}_\ell(t)} \bigg|^{\nu} \ind\Big( \inf_{t \in \mathcal{T}} J_\ell(t) > 0\Big) \right] 
   + \pr \left( \inf_{t \in \mathcal{T}} J_\ell(t) = 0 \right) \\
 & \le \ex \left[ \sup_{t \in \mathcal{T}} \bigg| \frac{J_\ell(t)}{\widebar{Z}_\ell(t)} \bigg|^{\nu} \ind\Big( \inf_{t \in \mathcal{T}} J_\ell(t) > 0 \land \inf_{t \in \mathcal{T}} J_\ell(t) = J_{\ell,\mathcal{T}} \Big) \right] \\
 & \quad + \ex \left[ \sup_{t \in \mathcal{T}} \bigg| \frac{J_\ell(t)}{\widebar{Z}_\ell(t)} \bigg|^{\nu} \ind\Big( \inf_{t \in \mathcal{T}} J_\ell(t) > 0 \land \inf_{t \in \mathcal{T}} J_\ell(t) > J_{\ell,\mathcal{T}} \Big) \right] \\
 & \quad + \pr \left( \inf_{t \in \mathcal{T}} J_\ell(t) = 0 \right) \\
 & \le \ex \left[ \bigg| \frac{J_{\ell,\mathcal{T}}}{\widebar{Z}_{\ell,\mathcal{T}}} \bigg|^{\nu} \right] + \pr \big( J_{\ell,\mathcal{T}} = 0 \big)  + \pr \left( \inf_{t \in \mathcal{T}} J_\ell(t) = 0 \right). 
\end{align*}
From the proof of \ref{C4}, we already know that $\pr(J_{\ell,\mathcal{T}} = 0) \le (1-p_\ell)^n$ and $\pr(\inf_{t \in \mathcal{T}} J_\ell(t) = 0) \le (1-p_\ell)^n$. Moreover, since $\widebar{Z}_{\ell,\mathcal{T}} = \sum_{i=1}^n \inf_{t \in \mathcal{T}} \ind(X_i(t-) = \ell)$ follows a binomial distribution $\text{Bin}(n,q)$ with $q = \pr(X_i(t-) = \ell \text{ for all } t \in \mathcal{T}) = p_\ell$ under \ref{C2:settingD}, we can apply Lemma \ref{lemma:binomial} to get that
\[ \ex \left[ \bigg| \frac{J_{\ell,\mathcal{T}}}{\widebar{Z}_{\ell,\mathcal{T}}} \bigg|^{\nu} \right] \leq \frac{(\nu+1)!}{\{np_\ell\}^{\nu}}. \]
Putting everything together, we arrive at
\[ \ex \left[ \bigg\| \frac{J_\ell}{\widebar{Z}_\ell} \bigg\|_{\infty}^{\nu} \right] \le \frac{(\nu+1)!}{\{np_\ell\}^{\nu}} + 2(1-p_\ell)^n \]
for any $\nu \in \naturals$. As $(1-p_\ell)^n$ converges to zero faster than $n^{-\nu}$, \ref{C5} is an immediate consequence of this bound.  
\end{proof}

\begin{proof}[Proof of \ref{C6}]
Since $\Delta_k^\beta = 0$ for all $k$ in the present case (without covariates), \ref{C6} is trivially fulfilled.  
\end{proof}


\def\thesection{\Alph{section}}
\setcounter{section}{1}
\def\theequation{B.\arabic{equation}}
\setcounter{equation}{0}
\section{Proof of auxiliary results}

For completeness, we now provide proofs of Propositions \ref{prop:fused-lasso-elementwise-error} and \ref{theo:LinSharpnackRinaldoTibshirani2017}.

\subsection*{Proof of Proposition \ref{prop:fused-lasso-elementwise-error}}

The proof is essentially a reformulation of the arguments for Theorem 3.1 in \cite{Zhang2019}. As a preliminary step, we analyze the minimization problem
\begin{equation}\label{eq:aux-min-problem}
\underset{a_1,\ldots,a_r \in \reals}{\textnormal{argmin}} \mathcal{L}(z_1,\ldots,z_r, a_1,\ldots,a_r \, | \, \underline{a},\overline{a}) 
\end{equation}
for given values $z_1,\ldots,z_r,\underline{a},\overline{a} \in \reals$, where  
\begin{align}
 & \mathcal{L}(z_1,\ldots,z_r,a_1,\ldots,a_r \, | \, \underline{a},\overline{a}) \nonumber \\ & \quad = \sum_{j=1}^r (z_j - a_j)^2 + \gamma \bigg\{ \sum_{j=1}^{r-1} |a_{j+1} - a_j| + |a_1 - \underline{a}| + |a_r - \overline{a}| \bigg\} \label{eq:aux-min-problem-crit}
\end{align}
with some penalization parameter $\gamma > 0$. We first derive certain statements about the structure of the solution vector
\[ (\hat{a}_1,\ldots,\hat{a}_r) \in \underset{a_1,\ldots,a_r \in \reals}{\textnormal{argmin}} \mathcal{L}(z_1,\ldots,z_r, a_1,\ldots,a_r \, | \, \underline{a},\overline{a}) \]
in Lemmas \ref{lemma:lemma3.6-Zhang2019} and \ref{lemma:theo3.4-Zhang2019} and then use these statements to obtain bounds on the components $\hat{a}_k$ for $k \in \{1,\ldots,r\}$ in Lemma \ref{prop:theo3.4-Zhang2019}. For ease of notation, we write $\hat{a}_0 := \underline{a}$ and $\hat{a}_{r+1} := \overline{a}$ in what follows.

\begin{lemma}\label{lemma:lemma3.6-Zhang2019}
Assume that 
\begin{equation}\label{eq:condition-lemma3.6-Zhang2019}
\bigg| \frac{1}{\sqrt{r_2 - r_1 + 1}} \sum_{j=r_1}^{r_2} z_j \bigg| \leq C_z 
\end{equation}
for some $r_1, r_2$ with $1 \le r_1 \le r_2 \le r$. Then the solution vector $(\hat{a}_1,\ldots,\hat{a}_r)$ has the following properties:
\begin{enumerate}[label=(\roman*)]
\item If $\hat{a}_{r_1-1} < \hat{a}_{r_1} = \ldots = \hat{a}_{r_2} > \hat{a}_{r_{2}+1}$, then
\begin{equation*}
\hat{a}_{r_1} = \ldots = \hat{a}_{r_2} \leq \frac{C_z^2}{4 \gamma}.
\end{equation*}
\item If $\hat{a}_{r_1-1} > \hat{a}_{r_1} = \ldots = \hat{a}_{r_2} < \hat{a}_{r_2+1}$, then
\begin{equation*}
\hat{a}_{r_1} = \ldots = \hat{a}_{r_2} \geq - \frac{C_z^2}{4 \gamma}.
\end{equation*}
\end{enumerate}
\end{lemma}

Lemma \ref{lemma:lemma3.6-Zhang2019} essentially says the following: if condition \eqref{eq:condition-lemma3.6-Zhang2019} is satisfied for all $r_1$, $r_2$ with $1 \le r_1 \le r_2 \le r$, then the sequence $\hat{a}_0, \hat{a}_1,\ldots,\hat{a}_r,\hat{a}_{r+1}$ cannot have a local minimum smaller than $-C_z^2/(4\gamma)$ or a local maximum larger than $C_z^2/(4\gamma)$.

\begin{proof}
We only prove (i) since (ii) can be verified by analogous arguments. For any $\varepsilon > 0$ smaller than $\min\{ \hat{a}_{r_1} - \hat{a}_{r_1-1}, \hat{a}_{r_2} - \hat{a}_{r_{2}+1} \}$, define 
\begin{equation}
\tilde{a}_j =
\begin{cases}
\hat{a}_j - \varepsilon & \text{for } j \in \{r_1, \ldots, r_2\} \\ 
\hat{a}_j & \text{for } j \notin \{r_1, \ldots, r_2\}.
\end{cases}
\end{equation}
Direct calculations yield that
\begin{align*}
\sum_{j=1}^r (z_j - \hat{a}_j)^2 - \sum_{j=1}^r (z_j - \tilde{a}_j)^2 
 & = 2 \varepsilon \sum_{j=r_1}^{r_2} (\hat{a}_j - z_j) - (r_2 - r_1 + 1) \varepsilon^2
\end{align*}
and 
\begin{align*}
\sum_{j=0}^r |\hat{a}_{j+1} - \hat{a}_j| - \sum_{j=0}^r | \tilde{a}_{j+1} - \tilde{a}_j| 
 & = |\hat{a}_{r_1} - \hat{a}_{r_1-1}| - |\hat{a}_{r_1} - \hat{a}_{r_1-1} - \varepsilon| \\[-0.25cm]
 & \quad + |\hat{a}_{r_2+1} - \hat{a}_{r_2}| - |\hat{a}_{r_2+1} - \hat{a}_{r_2} + \varepsilon| = 2 \varepsilon,
\end{align*}
where the last equality uses that $\hat{a}_{r_1} - \hat{a}_{r_1 - 1} > 0$, $\hat{a}_{r_2} - \hat{a}_{r_2 + 1} > 0$ and $0 < \varepsilon < \min\{ \hat{a}_{r_1} - \hat{a}_{r_1-1}, \hat{a}_{r_2} - \hat{a}_{r_{2}+1} \}$. By definition of $(\hat{a}_1,\ldots,\hat{a}_r)$, it holds that
\[ \mathcal{L}(z_1,\ldots,z_r, \hat{a}_1,\ldots,\hat{a}_r \, | \, \underline{a},\overline{a}) \le \mathcal{L}(z_1,\ldots,z_r, \tilde{a}_1,\ldots,\tilde{a}_r \, | \, \underline{a},\overline{a}). \]
Hence, it follows that
\begin{align*}
0 & \geq \mathcal{L}(z_1,\ldots,z_r, \hat{a}_1,\ldots,\hat{a}_r \, | \, \underline{a},\overline{a}) - \mathcal{L}(z_1,\ldots,z_r, \tilde{a}_1,\ldots,\tilde{a}_r \, | \, \underline{a},\overline{a})  \\
  & = \left[ \sum_{j=1}^r (z_j - \hat{a}_j)^2 + \gamma \left( |\hat{a}_1 - \underline{a}| + |\hat{a}_r - \overline{a}| + \sum_{j=1}^{r-1} |\hat{a}_{j+1} - \hat{a}_j| \right) \right] \\
  & \quad - \left[ \sum_{j=1}^r (z_j - \tilde{a}_j)^2 + \gamma \left( |\tilde{a}_1 - \underline{a}| + |\tilde{a}_r - \overline{a}| + \sum_{j=1}^{r-1} |\tilde{a}_{j+1} - \tilde{a}_j| \right) \right] \\
  & = \left[ \sum_{j=1}^r (z_j - \hat{a}_j)^2 - \sum_{j=1}^r (z_j - \tilde{a}_j)^2 \right] + \gamma \left[ \sum_{j=0}^{r} |\hat{a}_{j+1} - \hat{a}_j| - \sum_{j=0}^{r}|\tilde{a}_{j+1} - \tilde{a}_j| \right] \\
  & = 2 \varepsilon \sum_{j=r_1}^{r_2} (\hat{a}_j - z_j) - (r_2 - r_1 + 1) \varepsilon^2 + 2 \gamma \varepsilon, 
\end{align*}
that is,
\[ 0 \geq \sum_{j=r_1}^{r_2} (\hat{a}_j - z_j) - \frac{(r_2 - r_1 + 1) \varepsilon}{2} + \gamma. \]
Letting $\varepsilon \to 0$, we obtain that
\begin{align*}
0 & \ge \sum_{j=r_1}^{r_2} (\hat{a}_j - z_j) + \gamma =  (r_2 - r_1 + 1) \hat{a}_{r_1} - \sum_{j=r_1}^{r_2} z_j + \gamma,
\end{align*}
which in turn implies that
\[ \hat{a}_{r_1} \leq \frac{\sum_{j=r_1}^{r_2} z_j - \gamma}{r_2 - r_1 + 1}. \]
Finally, using the assumption that $\sum_{j=r_1}^{r_2} z_j \le C_z \sqrt{r_2-r_1+1}$ and the inequality $\frac{a}{x} - \frac{b}{x^2} \leq \frac{a^2}{4b}$ for $x > 0$, we arrive at   
\[ \hat{a}_{r_1} \leq \frac{C_z \sqrt{r_2 - r_1 + 1} - \gamma}{r_2 - r_1 + 1} \leq \frac{C_z^2}{4 \gamma}. \qedhere \]
\end{proof}

\begin{lemma}\label{lemma:theo3.4-Zhang2019}
Assume that
\begin{equation}\label{eq:condition-theo3.4-Zhang2019}
\max_{1 \le r_1 \le r_2 \le r} \bigg| \frac{1}{\sqrt{r_2 - r_1 + 1}} \sum_{j=r_1}^{r_2} z_j \bigg| \leq C_z. 
\end{equation}
Then the solution vector $(\hat{a}_1,\ldots,\hat{a}_r)$ has one of the following patterns:
\pagebreak
\begin{enumerate}[label=Pattern \Alph*:,leftmargin=2.15cm]
\item 
There exist $r_1, r_2 \in \{0,\ldots,r+1\}$ with $r_1 \leq r_2$ such that
\begin{itemize}[label=,leftmargin=0.1cm]
\item $\hat{a}_0 \geq \ldots \geq \hat{a}_{r_{1}-1} > \frac{C_z^2}{4 \gamma}$ or $\hat{a}_0 \leq \ldots \leq \hat{a}_{r_1-1} < - \frac{C_z^2}{4 \gamma}$,
\item $|\hat{a}_j| \leq \frac{C_z^2}{4 \gamma}$ for all $j \in \{r_1,\ldots,r_2\}$, and 
\item $\frac{C_z^2}{4 \gamma} < \hat{a}_{r_2+1} \leq \ldots \leq \hat{a}_{r+1}$ or $-\frac{C_z^2}{4 \gamma} > \hat{a}_{r_2+1} \geq \ldots \geq \hat{a}_{r+1}$.
\end{itemize}
\item There exists $r_0 \in \{0,\ldots,r\}$ such that either (a) or (b) holds:
\begin{enumerate}[label=(\alph*),leftmargin=0.75cm]
\item $\hat{a}_0 \geq \ldots \geq \hat{a}_{r_0} > \frac{C_z^2}{4 \gamma}$ and $- \frac{C_z^2}{4 \gamma} > \hat{a}_{r_0+1} \geq \ldots \geq \hat{a}_{r+1}$
\item $\hat{a}_0 \leq \ldots \leq \hat{a}_{r_0} < - \frac{C_z^2}{4 \gamma}$ and $\frac{C_z^2}{4 \gamma} < \hat{a}_{r_0+1} \leq \ldots \leq \hat{a}_{r+1}$.
\end{enumerate}
\item There exist $r_1, r_2 \in \{0,\ldots,r+1\}$ with $r_1 \leq r_2$ such that either (a) or (b) holds:
\begin{enumerate}[label=(\alph*),leftmargin=0.75cm]
\item $\hat{a}_j > \frac{C_z^2}{4 \gamma}$ for all $j$ and $\hat{a}_0 \geq \ldots > \hat{a}_{r_1} = \ldots = \hat{a}_{r_2} < \ldots \leq \hat{a}_{r+1}$
\item $\hat{a}_j < -\frac{C_z^2}{4 \gamma}$ for all $j$ and $\hat{a}_0 \leq \ldots < \hat{a}_{r_1} = \ldots = \hat{a}_{r_2} > \ldots \geq \hat{a}_{r+1}$.
\end{enumerate}
\end{enumerate}
\end{lemma}

\begin{proof}
Let $\mathcal{I} = \{ 0 \le j \le r+1: |\hat{a}_j| \le C_z^2 / (4\gamma) \}$. From Lemma \ref{lemma:lemma3.6-Zhang2019}, we already know that the sequence $\hat{a}_0, \hat{a}_1,\ldots,\hat{a}_r,\hat{a}_{r+1}$ cannot have a local minimum smaller than $-C_z^2/(4\gamma)$ or a local maximum larger than $C_z^2/(4\gamma)$ under condition \eqref{eq:condition-theo3.4-Zhang2019}. This implies that $\mathcal{I}$ must be an interval of the form $\mathcal{I} = \{r_1,\ldots,r_2\}$ for some $0 \le r_1 \le r_2 \le r+1$ or the empty set.

We first examine the case where $\mathcal{I}$ is a non-empty interval of the form $\mathcal{I} = \{r_1,\ldots,r_2\}$. Since the sequence $\hat{a}_0, \hat{a}_1,\ldots,\hat{a}_r,\hat{a}_{r+1}$ cannot have a local maximum (minimum) larger than $C_z^2/(4\gamma)$ (smaller than $-C_z^2/(4\gamma)$), the values $\hat{a}_0, \ldots, \hat{a}_{r_{1}-1}$ must either form a decreasing sequence with $\hat{a}_0 \geq \ldots \geq \hat{a}_{r_{1}-1} > C_z^2 / (4 \gamma)$ or an increasing sequence with $\hat{a}_0 \leq \ldots \leq \hat{a}_{r_1-1} < - C_z^2 / (4 \gamma)$. The values $\hat{a}_{r_2+1}, \ldots, \hat{a}_{r+1}$ must fulfill an analogous restriction. We thus end up with pattern A.

Next consider the case where $\mathcal{I}$ is the empty set and thus $|\hat{a}_j| > C_z^2 / (4 \gamma)$ for all $j \in \{0,\ldots,r+1\}$. There are two subcases: either there are $\hat{a}_j$'s with different signs or all $\hat{a}_j$'s have the same sign. 
We start with the subcase where the $\hat{a}_j$'s have different signs. To avoid a violation of Lemma \ref{lemma:lemma3.6-Zhang2019} by producing a local maximum/minimum which is too large/small, the following must hold: The signs of the $\hat{a}_j$'s can only change once, that is, there exists $r_0$ such that $\hat{a}_j$ is positive (negative) for all $j \le r_0$ and negative (positive) for all $j > r_0$. Moreover, when the signs are first positive (negative) and then negative (positive), the sequence $\hat{a}_0, \hat{a}_1,\ldots,\hat{a}_r,\hat{a}_{r+1}$ must be decreasing (increasing). We thus end up with pattern B. 
We now turn to the second subcase where all $\hat{a}_j$'s have the same sign. If $\hat{a}_j > C_z^2 / (4 \gamma)$ for all $j$, we must have that $\hat{a}_0 \geq \ldots > \hat{a}_{r_1} = \ldots = \hat{a}_{r_2} < \ldots \leq \hat{a}_{r+1}$ for some indices $r_1$ and $r_2$ because otherwise the sequence $\hat{a}_0, \hat{a}_1,\ldots,\hat{a}_r,\hat{a}_{r+1}$ would have a local maximum larger than $C_z^2/(4\gamma)$. Analogously, if $\hat{a}_j < -C_z^2 / (4 \gamma)$ for all $j$, we must have that $\hat{a}_0 \leq \ldots < \hat{a}_{r_1} = \ldots = \hat{a}_{r_2} > \ldots \geq \hat{a}_{r+1}$. As a result, we obtain pattern C. 
\end{proof}

With the help of Lemmas \ref{lemma:lemma3.6-Zhang2019} and \ref{lemma:theo3.4-Zhang2019}, we now prove that the entries of the solution vector $(\hat{a}_1,\ldots,\hat{a}_r)$ can be bounded as follows.
\begin{lemma}\label{prop:theo3.4-Zhang2019}
If
\begin{equation}\label{eq:bound-theo3.4-Zhang2019}
\max_{1 \le r_1 \le r_2 \le r} \left| \frac{1}{\sqrt{r_2-r_1+1}} \sum_{j=r_1}^{r_2} z_j \right| \le C_z, 
\end{equation}
then
\begin{equation}\label{eq:statement-theo3.4-Zhang2019}
|\hat{a}_j| \le \max \left\{ \frac{C_z}{\sqrt{j}}, \frac{C_z}{\sqrt{r+1-j}},\frac{C_z^2}{4\gamma},\frac{2\gamma}{r} + \frac{2C_z}{\sqrt{r}} \right\} 
\end{equation}
for all $j \in \{1,\ldots,r\}$. 
\end{lemma}

\begin{proof}
From Lemma \ref{lemma:theo3.4-Zhang2019}, we already know that the vector $(\hat{a}_0,\hat{a}_1,\ldots,\hat{a}_r,\hat{a}_{r+1})$ has pattern A, B or C. We now examine these three patterns one after the other.

First suppose that the vector has pattern A. Without loss of generality, assume that $\hat{a}_0 \geq \ldots \geq \hat{a}_{r_{1}-1} > C_z^2 / (4 \gamma)$ and $C_z^2 / (4 \gamma) < \hat{a}_{r_2+1} \leq \ldots \leq \hat{a}_{r+1}$ with $1 \le r_1 \le r_2 \le r$. The other possibilities entailed by pattern A can be treated analogously. To start with, we have a closer look at the partial sequence $\hat{a}_0,\ldots,\hat{a}_{r_1-1}$. For any $k \leq r_1 - 1$, define $k' := k'(k)$ as the biggest index such that $\hat{a}_k = \hat{a}_{k+1} = \ldots = \hat{a}_{k'} \neq \hat{a}_{k'+1}$. Note that $k' \leq r_1 - 1$ since $\hat{a}_{r_1-1} \neq \hat{a}_{r_1}$. Defining 
\[ \tilde{a}_j = 
\begin{cases} 
\hat{a}_j - \varepsilon & \text{for } 1 \le j \leq k' \\ 
\hat{a}_j & \text{else} 
\end{cases} 
\]
with $\varepsilon > 0$, we obtain that 
\begin{align*}
\sum_{j=0}^r |\hat{a}_{j+1} - \hat{a}_j| - \sum_{j=0}^r | \tilde{a}_{j+1} - \tilde{a}_j| 
 & = | \hat{a}_1 - \hat{a}_0 | - | \hat{a}_1 - \hat{a}_0 - \varepsilon | \\[-0.3cm]
 & \qquad + |\hat{a}_{k'+1} - \hat{a}_{k'}| - |\hat{a}_{k'+1} - \hat{a}_{k'} + \varepsilon| = 0 
\end{align*}
for sufficiently small $\varepsilon > 0$ (taking into account that $\hat{a}_{k'+1} - \hat{a}_{k'} < 0$ and $\hat{a}_1 - \hat{a}_0 \le 0$) and
\begin{align*}
\sum_{j=1}^r (z_j - \hat{a}_j)^2 - \sum_{j=1}^r (z_j - \tilde{a}_j)^2 
 & = 2 \varepsilon \sum_{j=1}^{k'} (\hat{a}_j - z_j) - k' \varepsilon^2.
\end{align*}
This implies that 
\begin{align*}
0 & \ge \mathcal{L}(z_1,\ldots,z_r, \hat{a}_1,\ldots,\hat{a}_r \, | \, \underline{a},\overline{a}) - \mathcal{L}(z_1,\ldots,z_r, \tilde{a}_1,\ldots,\tilde{a}_r \, | \, \underline{a},\overline{a}) \\
  & = \left[ \sum_{j=1}^r (z_j - \hat{a}_j)^2 - \sum_{j=1}^r (z_j - \tilde{a}_j)^2 \right] + \gamma \left[ \sum_{j=0}^{r} |\hat{a}_{j+1} - \hat{a}_j| - \sum_{j=0}^{r}|\tilde{a}_{j+1} - \tilde{a}_j| \right] \\*
  & = 2 \varepsilon \sum_{j=1}^{k'} (\hat{a}_j - z_j) - k'\varepsilon^2 
\end{align*}
and thus
\[ 0 \geq \sum_{j=1}^{k'} (\hat{a}_j - z_j) - \frac{k' \varepsilon}{2}. \]
Letting $\varepsilon \to 0$, we further get that
\[ 0 \geq \sum_{j=1}^{k'} (\hat{a}_j - z_j) \ge  k' \hat{a}_{k'} - \sum_{j=1}^{k'} z_j = k' \hat{a}_k - \sum_{j=1}^{k'} z_j. \]
Rearranging, inserting the bound \eqref{eq:bound-theo3.4-Zhang2019} and taking into account that $k \le k'$ yields 
\begin{equation}\label{eq:theo3.4-Zhang2019-structureA-bound1}
\hat{a}_{k} \leq \frac{C_z \sqrt{k'}}{k'} = \frac{C_z}{\sqrt{k'}} \leq \frac{C_z}{\sqrt{k}} 
\end{equation}
for any $k \in \{1,\ldots,r_1-1\}$. Analyzing the partial sequence $\hat{a}_{r_2+1},\ldots,\hat{a}_{r+1}$ in a similar way shows that 
\begin{equation}\label{eq:theo3.4-Zhang2019-structureA-bound2} 
\hat{a}_{k} \leq \frac{C_z}{\sqrt{r-k+1}} 
\end{equation}
for any $k \in \{r_2+1,\ldots,r\}$. Combining the bounds \eqref{eq:theo3.4-Zhang2019-structureA-bound1} and \eqref{eq:theo3.4-Zhang2019-structureA-bound2} with the fact that $|\hat{a}_j| \le C_z^2/(4\gamma)$ for $j \in \{r_1,\ldots,r_2\}$, we finally arrive at the desired bound \eqref{eq:statement-theo3.4-Zhang2019}.

We next suppose that the vector $(\hat{a}_0,\hat{a}_1,\ldots,\hat{a}_r,\hat{a}_{r+1})$ has pattern B. In this case, we can apply the same arguments as for pattern A to the partial sequences $\hat{a}_0,\ldots,\hat{a}_{r_0}$ and $\hat{a}_{r_0+1},\ldots,,\hat{a}_{r+1}$ to verify the bound \eqref{eq:statement-theo3.4-Zhang2019}.

We finally turn to the case that $(\hat{a}_0,\hat{a}_1,\ldots,\hat{a}_r,\hat{a}_{r+1})$ has pattern C. Without loss of generality, we let $1 \le r_1 \le r_2 \le r$ and consider case (a), that is, $\hat{a}_j > C_z^2 / (4 \gamma)$ for all $j$ and $\hat{a}_0 \geq \ldots > \hat{a}_{r_1} = \ldots = \hat{a}_{r_2} < \ldots \leq \hat{a}_{r+1}$. The other possibilities entailed by pattern C can be treated analogously. The same arguments as for pattern A yield that 
\[ \hat{a}_{k} \leq \frac{C_z}{\sqrt{k}} \quad \text{for } k \in \{1,\ldots,r_1-1\} \]
and 
\[ \hat{a}_{k} \leq \frac{C_z}{\sqrt{r-k+1}} \quad \text{for } k \in \{r_2+1,\ldots,r\}. \]
It thus remains to bound $\hat{a}_k$ for $k \in \mathcal{I} = \{r_1, \ldots, r_2\}$. Defining
\[ \tilde{a}_j = \begin{cases} 
\hat{a}_j - \varepsilon & \text{for } j \in \mathcal{I} \\ 
\hat{a}_j & \text{for } j \notin \mathcal{I} 
\end{cases} \]
with $\varepsilon > 0$, we get that
\begin{align*}
 & \sum_{j=0}^r |\hat{a}_{j+1} - \hat{a}_j| - \sum_{j=0}^r | \tilde{a}_{j+1} - \tilde{a}_j| \\
 & \quad = |\hat{a}_{r_1} - \hat{a}_{r_1-1}| - |\hat{a}_{r_1} - \hat{a}_{r_1-1} - \varepsilon| + |\hat{a}_{r_2+1} - \hat{a}_{r_2}| - |\hat{a}_{r_2+1} - \hat{a}_{r_2} + \varepsilon| = - 2\varepsilon
\end{align*}
for sufficiently small $\varepsilon > 0$ (taking into consideration that $\hat{a}_{r_{1}} - \hat{a}_{r_{1}-1} < 0$ and $\hat{a}_{r_{2}+1} - \hat{a}_{r_{2}} > 0$) and 
\begin{align*}
\sum_{j=1}^r (z_j - \hat{a}_j)^2 - \sum_{j=1}^r (z_j - \tilde{a}_j)^2 
 & = 2 \varepsilon \sum_{j=r_{1}}^{r_{2}} (\hat{a}_j - z_j) - (r_2 - r_1 + 1) \varepsilon^2.
\end{align*}
We thus obtain that
\begin{align*}
0 & \ge \mathcal{L}(z_1,\ldots,z_r, \hat{a}_1,\ldots,\hat{a}_r \, | \, \underline{a},\overline{a}) - \mathcal{L}(z_1,\ldots,z_r, \tilde{a}_1,\ldots,\tilde{a}_r \, | \, \underline{a},\overline{a}) \\
  & = \left[ \sum_{j=1}^r (z_j - \hat{a}_j)^2 - \sum_{j=1}^r (z_j - \tilde{a}_j)^2 \right] + \gamma \left[ \sum_{j=0}^{r} |\hat{a}_{j+1} - \hat{a}_j| - \sum_{j=0}^{r}|\tilde{a}_{j+1} - \tilde{a}_j| \right] \\
  & = 2 \varepsilon \sum_{j=r_{1}}^{r_{2}} (\hat{a}_j - z_j) - (r_2 - r_1 + 1) \varepsilon^2 - 2 \gamma \varepsilon, 
\end{align*}
which in turn yields that
\[ 0 \geq \sum_{j=r_{1}}^{r_{2}} (\hat{a}_j - z_j) - \frac{(r_2 - r_1 + 1) \varepsilon}{2} - \gamma. \]
Letting $\varepsilon \to 0$, we can infer that
\[ 0 \geq \sum_{j=r_{1}}^{r_{2}} (\hat{a}_j - z_j) - \gamma = (r_2 - r_1 + 1) \hat{a}_{r_{1}} - \sum_{j=r_{1}}^{r_{2}} z_j - \gamma \]
and thus
\[ \hat{a}_{r_{1}} \leq \frac{\sum_{j=r_{1}}^{r_{2}} z_j + \gamma}{r_2 - r_1 + 1}. \]
Rearranging this inequality, recalling that $\hat{a}_{r_1} = \hat{a}_{k}$ for all $k \in \mathcal{I}$ and inserting the bound \eqref{eq:bound-theo3.4-Zhang2019} yields
\begin{equation}\label{eq:theo3.4-Zhang2019-structureC-bound1}
\hat{a}_{k} \leq \frac{C_z \sqrt{r_2 - r_1 + 1} + \gamma}{r_2 - r_1 + 1} = \frac{C_z}{\sqrt{r_2-r_1+1}} + \frac{\gamma}{r_2-r_1+1} 
\end{equation}
for any $k \in \mathcal{I}$. Since we deal with pattern 3(a), it additionally holds that
\begin{equation}\label{eq:theo3.4-Zhang2019-structureC-bound2}
\hat{a}_k \leq \min \left\{ \frac{C_z}{\sqrt{r_1-1}}, \frac{C_z}{\sqrt{r-r_2}} \right\} 
\end{equation}
for any $k \in \mathcal{I}$. Combining the bounds \eqref{eq:theo3.4-Zhang2019-structureC-bound1} and \eqref{eq:theo3.4-Zhang2019-structureC-bound2}, we can conclude the following:
\begin{itemize}
\item If $r_1 - 1 \geq r/4$, then $\hat{a}_k \leq \frac{2 C_z}{\sqrt{r}}$ for $k \in \mathcal{I}$.
\item If $r - r_2 \geq r/4$, then $\hat{a}_k \leq \frac{2 C_z}{\sqrt{r}}$ for $k \in \mathcal{I}$.
\item If $r_1 - 1 \leq r/4$ and $r - r_2 \leq r/4$, then $r_2 - r_1 + 1 \geq r/2$ and thus $\hat{a}_k \leq \frac{C_z}{\sqrt{r/2}} + \frac{2 \gamma}{r}$ for $k \in \mathcal{I}$.
\end{itemize}
Combining these bounds for $k \in \mathcal{I}$ with those derived for $k \notin \mathcal{I}$, we finally arrive at the desired bound \eqref{eq:statement-theo3.4-Zhang2019}.
\end{proof}

To prove Proposition \ref{prop:fused-lasso-elementwise-error}, we now connect the minimization problem underlying the fused lasso to the minimization problem \eqref{eq:aux-min-problem} analyzed above. For simplicity, we denote the fused lasso by $\hat{\bs{\alpha}}$, that is, we drop the subscript $\lambda$. It holds that
\[ \hat{\bs{\alpha}} \in  \underset{\bs{a} \in \reals^{n}}{\textnormal{argmin}} \ \mathcal{L}(\bs{Y},\bs{a}) \]
with
\[ \mathcal{L}(\bs{Y},\bs{a}) = \frac{1}{n} \sum_{j=1}^{n} (Y_j - a_j)^2 + \lambda \sum_{j=1}^{n-1} |a_{j+1} - a_j|. \]
Let $1 < n_1 < \ldots < n_K \le n$ be the indices where the signal vector $\hazvec = (\haz_1,\ldots,\haz_{n})$ has a jump (i.e., $\haz_{n_k-1} \ne \haz_{n_k}$ for $1 \le k \le K$), and additionally set $n_0 = 1$ along with $n_{K+1} = n + 1$. Given the indices $n_1,\ldots,n_K$, any vector $\bs{x} = (x_1,\ldots,x_n) \in \reals^{n}$ can be represented as $\bs{x} = (\bs{x}_0,\bs{x}_1,\ldots,\bs{x}_K)$ with $\bs{x}_k = (x_{n_k},\ldots,x_{n_{k+1}-1})$ for $k \in \{0,\ldots,K\}$. Obviously, for any $k \in \{1,\ldots,K-1\}$, it holds that 
\begin{align*}
\hat{\bs{\alpha}}_k 
 & \in \underset{\bs{a}_k \in \reals^{n_{k+1} - n_k + 1}}{\textnormal{argmin}} \mathcal{L}(\bs{Y},\hat{\bs{\alpha}}_0,\ldots,\hat{\bs{\alpha}}_{k-1}, \ \bs{a}_k, \ \hat{\bs{\alpha}}_{k+1},\ldots,\hat{\bs{\alpha}}_K) \\
 & = \underset{\bs{a}_k \in \reals^{n_{k+1} - n_k + 1}}{\textnormal{argmin}} \mathcal{L}(\bs{Y}_k,\bs{a}_k \, | \, \hat{\alpha}_{n_k-1}, \hat{\alpha}_{n_{k+1}}),
\end{align*}
where 
\begin{align*}
 & \mathcal{L}(\bs{Y}_k, \bs{a}_k \, | \, \hat{\alpha}_{n_k-1}, \hat{\alpha}_{n_{k+1}}) \\*
 & \quad = \sum_{j=n_k}^{n_{k+1}-1} (Y_j - a_j)^2 + \gamma \sum_{j=n_k}^{n_{k+1}-2} |a_{j+1} - a_j| \\* & \qquad \qquad + \gamma \big\{ |a_{n_k} - \hat{\alpha}_{n_k - 1}| + |\hat{\alpha}_{n_{k+1}} - a_{n_{k+1}-1}| \big\}
\end{align*}
is the criterion function defined in \eqref{eq:aux-min-problem-crit} with $\gamma = n \lambda$. This shows that for any $k \in \{1,\ldots,K-1\}$, the subvector $\hat{\bs{\alpha}}_k$ can be obtained by solving a subproblem with correctly specified end points, in particular, by minimizing the criterion function $ \mathcal{L}(\bs{Y}_k,\bs{a}_k \, | \, \underline{a}, \overline{a} )$ with the end points $\underline{a}=\hat{\alpha}_{n_k-1}$ and $\overline{a}=\hat{\alpha}_{n_{k+1}}$. It is easy to see that for any constant vector $\bs{c}_k = (c,\ldots,c) \in \reals^{n_{k+1} - n_k + 1}$,
\[ (\hat{\bs{\alpha}}_k - \bs{c}_k) \in 
\underset{\bs{a}_k \in \reals^{n_{k+1} - n_k + 1}}{\textnormal{argmin}}\mathcal{L}(\bs{Y}_k - \bs{c}_k, \bs{a}_k \, | \, \hat{\alpha}_{n_k-1}-c, \hat{\alpha}_{n_{k+1}}-c). \]
Hence, if we shift the input vector $\bs{Y}_k$ as well as the end points by a constant $c$, the solution vector $\hat{\bs{\alpha}}_k$ gets shifted by the same amount. Since the $k$-th part of the signal vector $\hazvec_k$ is constant, this in particular implies that 
\begin{align}
(\hat{\bs{\alpha}}_k - \hazvec_k) 
 & \in \underset{\bs{a}_k \in \reals^{n_{k+1} - n_k + 1}}{\textnormal{argmin}}\mathcal{L}(\bs{u}_k, \bs{a}_k \, | \, \hat{\alpha}_{n_k-1}-\haz_{n_k}, \hat{\alpha}_{n_{k+1}}-\haz_{n_k}),  \label{eq:min-subvector-k}
\end{align}
where $\bs{u}_k = \bs{Y}_k - \hazvec_k$. Assuming that  
\[ \max_{1 \le k \le \ell \le n} \left| \frac{1}{\sqrt{\ell-k+1}} \sum_{j=k}^{\ell} u_j \right| \le C_u, \]
we can now apply Lemma \ref{prop:theo3.4-Zhang2019} to obtain that  
\begin{equation}\label{eq:theo3.4-Zhang2019-desired-bound}
|\hat{\alpha}_j - \haz_j| \le \max \left\{ \frac{C_u}{\sqrt{d_j}}, \frac{C_u^2}{4n\lambda},\frac{2n\lambda}{r_{k(j)}} + \frac{2C_u}{\sqrt{r_{k(j)}}} \right\} 
\end{equation}
for $j \in \{n_1,\ldots,n_{K}-1\}$.

The subvectors $\hat{\bs{\alpha}}_0 - \hazvec_0$ and $\hat{\bs{\alpha}}_K - \hazvec_K$ solve a minimization problem similar to \eqref{eq:min-subvector-k}, where only one instead of two endpoints are given. Repeating the above arguments for these slightly different minimization problems yields the bound \eqref{eq:theo3.4-Zhang2019-desired-bound} also for $j \in \{1,\ldots,n_1-1\} \cup \{n_K,\ldots,n\}$. This completes the proof.

\subsection*{Proof of Proposition \ref{theo:LinSharpnackRinaldoTibshirani2017}}

Fix some $\varepsilon > 0$. By assumption, we know that there is a constant $C = C(\varepsilon)$ and an integer $N_1 = N_1(\varepsilon)$ such that 
\[ \pr \left( \frac{1}{n} \norm{\tilde{\bs{\alpha}} - \hazvec}_2^2 \ge \frac{C}{2} R_n \right) \leq \varepsilon \]
for all $n \geq N_1$. We also know that there is an integer $N_2 = N_2(\varepsilon) > 0$ with $2 \lceil C n R_n / H_n^2 \rceil < D_n$ for all $n \geq N_2$. Let $N = \max\{N_1,N_2\}$, take $n \geq N$ and let $r_n = \lceil C n R_n / H_n^2 \rceil$.
Suppose that $d(\tilde{J} \, | \, J^*) > r_n$. Then by definition, there exists a jump index $n_k \in J^*$ such that no jump points of $\tilde{\bs{\alpha}}$ are within $r_n$ of $n_k$, which implies that $\tilde{\alpha}_j$ is constant over $j \in \{n_k - r_n, \ldots, n_k + r_n - 1\}$. With the notation
\[ z = \tilde{\alpha}_{n_k - r_{n}} = \ldots = \tilde{\alpha}_{n_k-1} = \tilde{\alpha}_{n_k} = \ldots = \tilde{\alpha}_{n_k + r_{n} - 1}, \]
we obtain the lower bound
\[ \frac{1}{n} \sum_{j = n_k - r_{n}}^{n_k + r_{n} - 1} \left( \tilde{\alpha}_j - \haz_j \right)^2 = \frac{r_n}{n} \left( z - \haz_{n_k-1} \right)^2 + \frac{r_n}{n} \left( z - \haz_{n_k} \right)^2 \geq \frac{r_n H_n^2}{2n} \ge \frac{C}{2} R_n , \]
where the first inequality holds because $(x-a)^2 + (x-b)^2 \geq (a-b)^2/2$ for all $x$ (the quadratic in $x$ here is minimized at $x = (a+b)/2$) and the second one holds because $r_n = \lceil C n R_n / H_n^2 \rceil$. As a consequence, we see that $d(\tilde{J} \, | \, J^*) > r_n$ implies
\[ \frac{1}{n} \norm{\tilde{\bs{\alpha}} - \hazvec}_2^2 \geq \frac{1}{n} \sum_{j = n_k - r_{n}}^{n_k + r_{n} - 1} \left( \tilde{\alpha}_j - \alpha_j^* \right)^2 \ge \frac{C}{2} R_n, \]
which yields that
\[ \pr \left( d(\tilde{J} \, | \, J^*) > r_n \right) \leq \pr \left( \frac{1}{n} \norm{\tilde{\bs{\alpha}} - \hazvec}_2^2 \ge \frac{C}{2} R_n \right) \leq \varepsilon \]
for all $n \geq N$. This completes the proof.

\end{document}